\documentclass[MNA]{mna} % 

\usepackage{titlesec}
\usepackage{hyphenat}
\SHORTTITLE{Ordinal Characterization of Similarity Judgments}

\TITLE{Ordinal Characterization of Similarity Judgments} % \thanks is optional. Insert line breaks with \\

\AUTHORS{%
  Jonathan D.~Victor\footnote{Weill Cornell Medical College.
    \EMAIL{jdvicto@med.cornell.edu}}
    \and %% remove this line and below if single author
 Guillermo~Aguilar\footnote{Technische Universität Berlin.
 	\EMAIL{guillermo.aguilar@mail.tu-berlin.de}}
  \and %% remove this line and below if single author
 Suniyya A.~Waraich\footnote{University of California -- Davis.
 	\EMAIL{swaraich@ucdavis.edu}}}%AUTHORS
\SHORTAUTHOR{J. D. Victor, G. Aguilar and S. A. Waraich}

\KEYWORDS{perceptual spaces; maximum likelihood estimation; ultrametric space; additive trees;
triads; multidimensional scaling} % 
\AMSSUBJ{91E30; 62P15; 92-08; 92-10; 51-08} % Edit. Separate items with ;
\SUBMITTED{October 23, 2023} % Edit.
\ACCEPTED{February 13, 2025} % Edit.

\ARXIVID{2310.07543} % Edit.
\VOLUME{5} %5
\YEAR{2025}
\PAPERNUM{1} %1
\DOI{10.46298/mna.12457} % 10.46298/mna.12457

\ABSTRACT{ Characterizing judgments of similarity within a perceptual or semantic domain, and making inferences about the underlying structure of this domain from these judgments, has an increasingly important role in cognitive and systems neuroscience. We present a new framework for this purpose that makes limited assumptions about how perceptual distances are converted into similarity judgments. The approach starts from a dataset of empirical judgments of relative similarities: the fraction of times that a subject chooses one of two comparison stimuli to be more similar to a reference stimulus. These empirical judgments provide Bayesian estimates of underling choice probabilities. From these estimates, we derive indices that characterize the set of judgments in three ways: compatibility with a symmetric dis-similarity, compatibility with an ultrametric space, and compatibility with an additive tree. Each of the indices is derived from rank-order relationships among the choice probabilities that, as we show, are necessary and sufficient for local consistency with the three respective characteristics. We illustrate this approach with simulations and example psychophysical datasets of dis-similarity judgments in several visual domains and provide code that implements the analyses at \url{https://github.com/jvlab/simrank}.}

\begin{document}

%%%%%%%%%%%%%%%%%%%%%%%%%%%%%%%%%%%%%%%%%%%%%%%%%%%%%%%%%%%%%%%%%%%
%%                                                               %%
%% No need for \maketitle.                                       %%
%%                                                               %%
%%%%%%%%%%%%%%%%%%%%%%%%%%%%%%%%%%%%%%%%%%%%%%%%%%%%%%%%%%%%%%%%%%%

%%%%%%%%%%%%%%%%%%%%%%%%%%%%%%%%%%%%%%%%%%%%%%%%%%%%%%%%%%%%%%%%%%%
%%                                                               %%
%% Please replace what follows by the body of your article       %%
%% (up to the bibliography):                                     %%
%%                                                               %%
%%%%%%%%%%%%%%%%%%%%%%%%%%%%%%%%%%%%%%%%%%%%%%%%%%%%%%%%%%%%%%%%%%%

\section{Introduction}

Characterization of the similarities between elements of a domain of sensory or semantic information is important for many reasons. First, these similarities, and the relationships between them, \cite{Edelman1998, Kemp2008, Tversky1977}, reveal the cognitive structure of the domain. Similarities are functionally important as they are the substrate for learning, generalization, and categorization \cite{Kemp2008, Saxe2019, Shepard1958, Zaidi2013}. At a mechanistic level, the quantification of similarities provides a way to test hypotheses concerning their neural substrates \cite{Kriegeskorte2013}. Thus, measuring of perceptual similarities, and using these judgments to make inferences about the geometry of the underlying perceptual spaces, plays an important role in cognitive and systems neuroscience. \par
The goal of this work is to present a novel approach that complements the standard strategies used for this purpose. The starting point for the present approach, in common with standard strategies, is a set of triadic similarity judgments: is stimulus $x$ or stimulus $y$ more similar to a reference stimulus $r$? To make geometric inferences from such data, one standard approach is to make use of a variant of multidimensional scaling \cite{deLeeuw1982, Knoblauch2008, Maloney2003, Tsogo2000, Victor2017, Waraich2022}, i.e., to associate the stimuli with points in a space, so that the distances between the points account for the perceptual similarities. Once these points are determined, inferences can be made about the dimensionality of the space, its curvature, and its topology. A second approach, topological data analysis, makes use of the distances directly, and then invokes graph-theoretic procedures \cite{Dabaghian2012, Giusti2015, Guidolin2022, Singh2008, Zhou2018} to infer these geometric features. \par
In applying these approaches to experimental data, one must deal with the fact that even if a forced-choice response is required, the response likely represents an underlying choice probability -- and that this choice probability may depend on sensory noise, noise in how distances are mentally computed and transformed into dis-similarities, and noise in the decision process in which dis-similarities are compared. As a consequence, analysis of an experimental dataset requires, at least implicitly, substantial assumptions. Such assumptions are not always benign: a monotonic transformation of distances -- which preserves binary similarity judgments -- can alter the dimensionality and curvature of a multidimensional scaling model \cite{Kruskal1978}. Topological data analysis via persistent homologies, which also only makes use of rank orders of distances, is invariant to a global monotonic transformation of distances, but makes other assumptions (for example, that this transformation is the same across the domain), and does not typically take into account a noise model. \par
With these considerations in mind, here we pursue an approach to make inferences from the choice probabilities themselves, as estimated from repeated triadic judgments. Our main assumption is that if, for any particular triad, comparison stimulus $x$ is chosen more often than comparison stimulus $y$ as closer to a reference stimulus $r$, then the distance between $x$ and $r$ is less than the distance between $y$ and $r$. Note that we do not make any assumptions about how relative or absolute distances are transformed into choice probabilities within an individual triad, or whether this transformation is the same across the domain.\par
As we show, despite the relative paucity of assumptions, the approach nevertheless provides indices that characterize a set of similarity judgments in three useful ways. The first index quantifies the extent to which the similarity judgments are consistent with conditions necessary for symmetry  -- i.e., that in judging similarity, reference and comparison stimuli are treated in the same way. While at first glance one might expect that similarity judgments always reflect symmetric distances and most models (including those considered here) make this assumption, this need not be the case \cite{Tversky1977, Tversky1982} -- so it is useful to have a way to determine whether an experimental dataset implies violations of symmetry. Applying our approach to test for symmetry is also the simplest of the three we consider. \par
The second index quantifies consistency with an ultrametric model, a geometry that is a formalization of strict hierarchical structure \cite{Semmes2007}. In an ultrametric model, elements of the space correspond to the leaf nodes on a rooted tree, and distance between two nodes is determined by the height of the first common ancestor. Such structure has been postulated for perceptual domains with semantic content \cite{Kemp2008, Saxe2019, Treves1997, Tversky1986}. Consideration of ultrametric models is also motivated by evidence that representation of olfactory \cite{Zhou2018} and physical \cite{Zhang2023} space may have a hyperbolic geometry, as hyperbolic geometries provide natural embeddings of trees \cite{DeSa2018}. \par
The third index quantifies consistency with an addtree (or additive tree) model, a generalization of the ultrametric model. In an addtree model, the distance between two points is determined by the length of the path on an acyclic graph. An addtree model allows for more flexible clustering and may be characteristic of semantic domains \cite{Kemp2008, Sattath1977}. Moreover, since a one-dimensional domain is a special case of an addtree \cite{Sattath1977}, exclusion of addtree structure implies that a one-dimensional model cannot account for the rank order of similarity judgments -- and therefore, cannot account for any model in which perceptual distances are monotonically related to these judgments. Addtree models are in a sense piecewise linear, and thus also may be appropriate models for olfactory perceptual spaces \cite{Wilson2017}. \par
The organization of this paper is as follows. The first section sets out a formal framework of the approach and includes rigorous mathematical results. The second section moves from this framework to procedures that can be applied to experimental data, yielding indices that characterize consistency with symmetry, ultrametric structure and addtree structure. To deal with the fact that the choice probabilities are estimated quantities, we use a Bayesian approach, and details of its implementation differ for the three indices. In particular, integration over a local prior extends exactly to a global prior for the symmetry index and the ultrametric index, but is only approximate for addtree. We then apply the method to synthetic datasets. These examples demonstrate the ability of the indices to characterize similarity structure, that the characterization is largely insensitive to the Bayesian prior, and illustrate how interpretation can be augmented by analysis of surrogate datasets. The final results section applies the method to experimental data from three visual domains. The Discussion considers caveats, limitations, alternative strategies, implications for experimental design, and avenues for further development.

\section{Theory}
\subsection*{Overview, key terms, and preliminaries}
Our goal is to develop indices that characterize a dataset of triadic similarity judgments, in a way
that provides insight into the structure of the underlying perceptual space. Our central assumption is that,
within a triadic judgment, the probability that a participant chooses one pair of stimuli as more similar
than an alternative is monotonically related to the similarity. Typical datasets include large numbers of
similarity judgments of overlapping triads, and the relationship between these judgments contains
information about the underlying perceptual space. We show how this information can be accessed,
without further making assumptions about the specifics of the monotonic relationship between choice
probability and (dis)-similarity, whether it is constant throughout the space, or the decision process
itself.\par

\begin{figure}[htbp]
    \centering % gives better spacing than \begin{center}...\end{center}
    \includegraphics[width=0.7\linewidth]{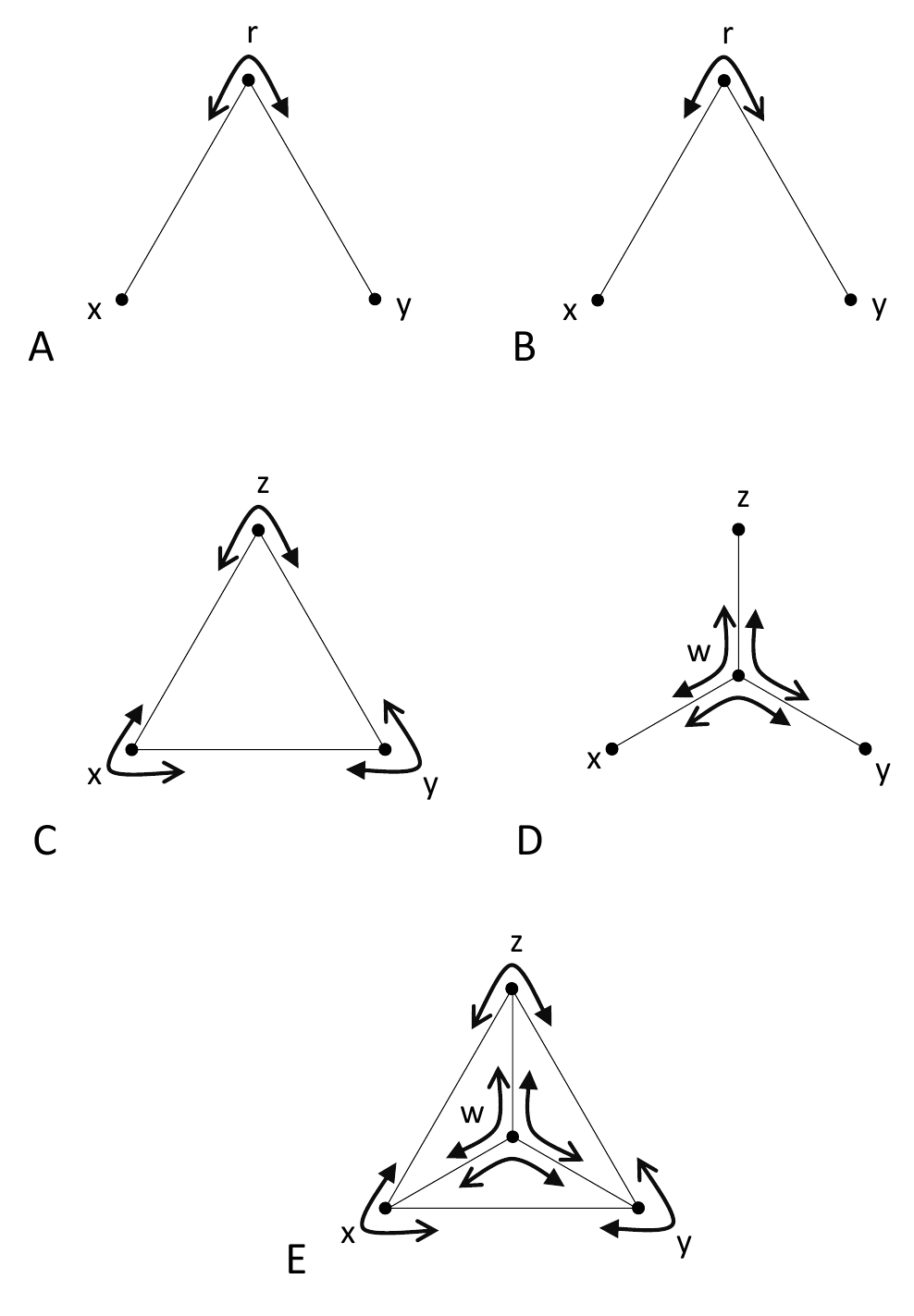}
    \caption{Panel A. A triad $(r;x,y)$ consists of a reference stimulus $r$ and two comparison stimuli, $x$ and $y$; the first comparison stimulus is indicated by the open arrow and the second by the closed arrow. Panel B. In the mirror triad $(r;y,x)$, the roles of the two comparison stimuli are reversed (note the arrowheads near the reference stimulus $r$). Panel C. A triplet is a set of three triads $\{(z;x,y),(x;y,z),(y;z,x)\}$ that can be formed from three stimuli, with each stimulus serving as a reference for the other two. Panel D. A tripod (shown here in ``top view'') is a set of three triads constructed with a single stimulus (here, $w$) serving as a reference for the other three stimuli taken in pairs. Panel E. A tent is a union of a triplet and a tripod.}
    \label{fig:1}
\end{figure}

We define a \textit{triad}, denoted $(r;x,y)$, to be an ordered triple of stimuli, consisting of a reference stimulus $r$ and two comparison stimuli, $x$ and $y$, all drawn from a set $S$ of all stimuli (Figure~\ref{fig:1}A). A triad forms the basic unit of data collection in an experiment: the participant is asked to decide, in a
forced-choice response, which of the two comparison stimuli is more similar to the reference. The \textit{mirror} (Figure~\ref{fig:1}B) of the triad $(r;x,y)$ is defined as the triad $(r;y,x)$. Several groupings of triads will be useful below. A \textit{triplet} (Figure~\ref{fig:1}C) is defined as a set of the triads $\{(z;x,y),(x;y,z),(y ;z,x)\}$ that can be formed from three stimuli, with each stimulus serving as the reference for one of the triads. A \textit{tripod}
(Figure~\ref{fig:1}D) is defined as a set of triads $\{(w;x,y),(w;y,z),(w;z,x)\}$. A \textit{tent} (Figure~\ref{fig:1}E) is defined as the
union of a triplet (the \textit{base} of the tent) and a tripod whose comparison stimuli are the stimuli in the triplet and whose common reference stimulus is the \textit{vertex} of the tent, and is denoted $\{w;x,y,z\}$, where $w$ is the vertex and the base is the triplet $\{(z;x,y),(x;y,z),(y;x,z)\}$. \par

We define the \textit{choice probability} for the triad $(r;x,y)$, denoted $R(r;x,y)$, to be the probability
that a participant judges $x$ as more similar to $r$ than $y$ is to $r$. We assume that the two comparison
stimuli in a triad are treated equivalently, i.e., that the response to a triad can also be considered to be the
alternative response to its mirror. That is, 
\begin{equation}\label{eq:1.1}
    R(r;x,y) + R(r;y,x) = 1.
\end{equation}

We view the choice probability $R(r;x,y)$ as an unknown to be estimated from an experiment in which the triad $(r;x,y)$, among others, is presented; this is discussed in detail in the
Implementation section below. We do not require that the experiment explore each triad or its mirror (in an experiment with $M$ stimuli, there are $M\binom{M-1}{2} = \frac{1}{2}M(M-1)(M-2)$ pairs of triads and their mirrors), though clearly greater coverage of the triads, and more repeats of each, will lead to better estimates of the choice probabilities. 
We also assume that an experimenter chooses, randomly, whether to present a triad or its mirror. \par
We assume that a triadic judgment is the result of a two-step process: first, estimation of the dis-similarity between the reference and each of the comparison stimuli, and second, comparison of these dis-similarities. 
We denote the dis-similarity of a comparison stimulus $z$ to a reference stimulus $r$ by $D(r,z)$. Our central assumption is that a participant is more likely to judge that $x$ is more similar to $r$ than $y$ is to $r$ if, and only if, $D(r,x) < D(r, y)$. That is,
\begin{equation}\label{eq:1.2}
    R(r;x,y) > \frac{1}{2} \iff D(r,x) < D(r,y).
\end{equation}
Because $\eqref{eq:1.1}$ holds for every triad and its mirror, an immediate consequence is that a choice probability of exactly $\frac{1}{2}$ only occurs when dis-similarities are exactly equal: combining $\eqref{eq:1.1}$ with $\eqref{eq:1.2}$ yields

\begin{equation}\label{eq:1.3}
    R(r;x,y) = \frac{1}{2} \iff D(r,x) = D(r,y).    
\end{equation}
We also assume that dis-similarities are non-negative, and that a dis-similarity of zero only occurs for
stimuli that are identical:
\begin{equation}\label{eq:1.4}
    D(x,y) \geq 0 \quad \text{and} \quad D(x,y) = 0 \iff x=y.
\end{equation}
Note, however, that we do not assume that dis-similarities are symmetric, i.e., that $D(x,y) = D(y,x).$ Rather, we will test whether a dataset is compatible with this constraint -- essentially, whether the
reference and comparison stimuli are treated equivalently. We use the term \textit{``symmetric distance''} for a function $d(x,y)$ that satisfies $\eqref{eq:1.4}$ and also $d(x,y)=d(y,x)$. Note that a symmetric distance need not satisfy the triangle inequality, and therefore need not be a metric (see Appendix~\ref{appendix1}). \par

Eqs. \eqref{eq:1.2} and \eqref{eq:1.3} formalize our focus on rank order of dis-similarities. That is, rather than assume (or attempt to infer) a quantitative relationship between the choice probability and perceived dis-
similarities, we only make use of the sign of comparisons -- for example, that an alligator and toothpaste are more dis-similar than an alligator and a panda -- but we do not attempt to infer the size of this difference, in absolute terms, relative to other dis-similarities not included in triads that have been presented, or relative to an internal noise. \par
The analyses below will ask whether the choice probabilities, and hence, the dis-similarities, suffice to rule out a model whose distances conform to a specific kind of geometric space. We consider two ways of formalizing this notion, each based on a set of stimuli, a set $T$ of triads formed from them, and a function $f$ on (reference, comparison) pairs within the triads. Given this set-up, we define a basic notion of compatibility: a set of dis-similarities $D$ is \textit{pointwise compatible} with $f$ if, for all the triads $(r;x,y) \in T$, the rank-order of pairwise dis-similarities is the same as the rank-order of values assigned by $f$:
\begin{equation}\label{eq:1.5}
    D(r,x) < D(r,y) \iff f(r,x) < f(r,y).
\end{equation}
This contrasts with a stronger definition: the dis-similarities $D$ is \textit{setwise compatible} with $f$ if for all
$(r;x,y) \in T$, there is a monotonic-increasing function $F$ for which
\begin{equation}\label{eq:1.6}
    f(x,y) = F(D(x,y)).
\end{equation}
We focus primarily on pointwise compatibility, a focus that considers only the relationships between dis-similarities that jointly occur in triadic comparisons. Since setwise compatibility implies pointwise compatibility, a set of choice probabilities that rules out pointwise compatibility necessarily rules out setwise compatibility. \par
The present analysis, which applies most directly to a paradigm in which each trial is devoted to judgment of a single triad, is applicable to other paradigms in which individual trials yield judgments about more than one triad, provided that each judgment can be considered to be independent of the context in which it is made. For example, in the ``odd one out'' paradigm (also known as the ``oddity task'' \cite{Kingdom2016}), three stimuli are presented and the participant is asked to choose which one is the outlier. Here, a selection of a stimulus $x_j$ out of a set $\{x_j, x_k, x_l\}$ can be interpreted as a judgment that $D(x_k, x_j) > D(x_k, x_l)$ and also that $D(x_l, x_j) > D(x_l, x_k)$, and thus contributes to estimates of choice probabilities for two triads, $(x_k; x_j, x_l)$ and $(x_l;x_j,x_k)$. The analysis is also applicable to paradigms in which they are to rank $m$ comparison stimuli $x_1,...,x_m$ in order of similarity to a shared reference stimulus $r$ \cite{Waraich2022}. The ranking obtained on each trial then contributes to estimation of choice probabilities for $\binom{m}{2}$ triads $(r;x_k, x_l)$, one for each pair of comparison stimuli.\par

\subsection*{Choice probabilities and conditions for compatibility with distance-based models}
The elements introduced above enable characterization of the choice probabilities in three ways: compatibility of dis-similarities with a symmetric distance, compatibility with an ultrametric model, and compatibility with an addtree model. The characterizations for symmetry and ultrametric structure are the most straightforward since these properties do not make use of additive structure.

\subsubsection*{Symmetry}

We consider symmetry first, as it is a fundamental property of distances and it is also the simplest of the three characterizations. We focus on the choice probabilities among members of a triplet of stimuli, $\{x, y, z\}$, which, for brevity, we denote $R_1 \triangleq R(x;y,z), R_2 \triangleq R(y;z,x), R_3 \triangleq R(z;x,y)$. These values fully characterize the choice probabilities among the three stimuli, as the three other triads
consisting of these stimuli are all mirrors of one of the $R_i$.\par
\begin{proposition}[ordinal conditions for symmetry]\label{pr:1}
    For any triplet composed of the stimuli $\{x, y, z\}$, the dis-similarities are pointwise compatible with a symmetric distance if, and only if, one of the following conditions hold. With $n_{half}$ as the number of $R_i$ that are exactly equal to $\frac{1}{2}$,
    \begin{enumerate}
        \item $n_{half} \in \{0, 1\}$ and the $R_i-\frac{1}{2}$ include both positive and negative values \\
        or
        \item $n_{half} = 3.$
    \end{enumerate}
\end{proposition}
\begin{proof}
    To show that these conditions are required for pointwise compatibility with symmetry: If none of them hold, then either $n_{half} = 2$, or $n_{half} \in \{0, 1\}$, and the nonzero values of $R_i - \frac{1}{2}$ are all of the same sign. \par
    Case 1: $n_{half} = 2.$ Without loss of generality (WLOG), say $R_1 = R_2$ but $R_3 \neq \frac{1}{2}.$ Then $D(x, y) = D(x, z), D(y, z) = D(y, x).$ If $D$ were pointwise compatible with a symmetric $d$, then
    \begin{equation}\label{eq:2.1}
        d(z,x) = d(x,z) = d(x,y) = d(y,x) = d(y,z) = d(z,y).
    \end{equation}
    Pointwise compatibility then requires that $D(z,x)=D(z,y),$ which contradicts $R_3 \neq \frac{1}{2}.$\par
    Case 2: $n_{half} \in \{0, 1\}.$ There are at least two nonzero values of $R_i - \frac{1}{2},$ and these nonzero values are all the same sign. WLOG assume that two of the $R_i$ are $> \frac{1}{2}$, and $R_2$ is one of these. Then, $R_1 \geq \frac{1}{2}, R_2 > \frac{1}{2}, R_3 \geq \frac{1}{2}$ requires that
    \begin{equation}\label{eq:2.2}
    \left.\begin{aligned}
        D(x, y) \leq D(x, z) \\
        D(z, x) \leq D(z,y) \\
        D(y, z) < D(y, x)
    \end{aligned}\right\}.
    \end{equation}
    If $D$ is pointwise compatible with a symmetric distance $d$, then \eqref{eq:2.2} implies that
        \begin{equation}\label{eq:2.3}
    \left.\begin{aligned}
        d(x, y) \leq d(x, z) \\
        d(z, x) \leq d(z,y) \\
        d(y, z) < d(y, x)
    \end{aligned}\right\}.
    \end{equation}
    Symmetry of $d$ leads to a contradiction,
    \begin{equation}\label{eq:2.4}
        d(x,y) \leq d(x, z) = d(z, x) \leq d(z, y) = d(y, z) < d(y, x). 
    \end{equation}
    
    To show that these conditions suffice for pointwise compatibility with a symmetric distance: $n_{half} = 3$ is trivial; all dis-similarities are equal. For $n_{half} \in \{0, 1\}$:, WLOG assume that $R_1 - \frac{1}{2}$ and $R_3-\frac{1}{2}$ have opposite signs; $R_2$, which compares $D(y, x)$ with $D(y, z)$ is unconstrained. Choose $d(x,y) = d(y, x) = D(y, x)$, $d(z, y) = d(y, z) = D(y, z)$; this guarantees that the rank-order of dis-similarities required by $R_2$ is respected by $d$. 
    If $R_1 > \frac{1}{2}$ and $R_3 < \frac{1}{2}$, we have $D(x, y) < D(x, z)$, $D(z, x) > D(z, y)$. Thus choosing $d(x, z) = d(z, x) = max\{d(x, y), d(z, y)\}+k$ (where $k > 0$) yields pointwise compatibility with $D$ on $\{x, y, z\}$. If instead, $R_1 < \frac{1}{2}$ and $R_3 > \frac{1}{2}$, a similar argument applies: these require $D(x, y) > D(x, z)$, $D(z, x) < D(z, y)$. Set $d(x, z) = max\{D(x, z), D(z, x)\}$ and put $d(x, y) = d(y, x) = D(y, x) + k$, $d(z, y) = d(y, z) = D(y, z) + k$, where $k$ is large enough to ensure that $d(x, y)$ and $d(z, y)$ are both larger than $d(x, z)$. 
\end{proof}
\begin{remark}\label{rem:1}
    Appendix~\ref{appendix1} shows that compatibility with a symmetric distance implies compatibility with a metric. 
\end{remark}

Proposition~\ref{pr:1} can be rephrased as follows: a necessary and sufficient condition for a triad to be pointwise compatible with a symmetric distance is that either all choice probabilities are $\frac{1}{2}$, or if not, at least one of the choice probabilities $R_i$ is strictly greater than $\frac{1}{2}$ and at least one of them is strictly less than $\frac{1}{2}$. That is, the triplet of choice probabilities $R_i$ is compatible with a symmetric distance if the three choice probabilities lie in a subset $\Omega_{\text{sym}}$ of $[0, 1]^3$ consisting of the $[0, 1]^3$ cube from which two smaller cubes, $[0, \frac{1}{2}]^3$ ($R_i \leq \frac{1}{2}$) and $[\frac{1}{2},1]^3$ ($R_i \geq \frac{1}{2}$), are removed:
    \begin{equation}\label{eq:2.5}
        \Omega_{\text{sym}} = [0, 1]^3 \setminus\left([0, \frac{1}{2}]^3 \cup [\frac{1}{2}, 1]^3\right).
    \end{equation}

\begin{figure}[htbp]
    \centering 
    \includegraphics[width=\linewidth]{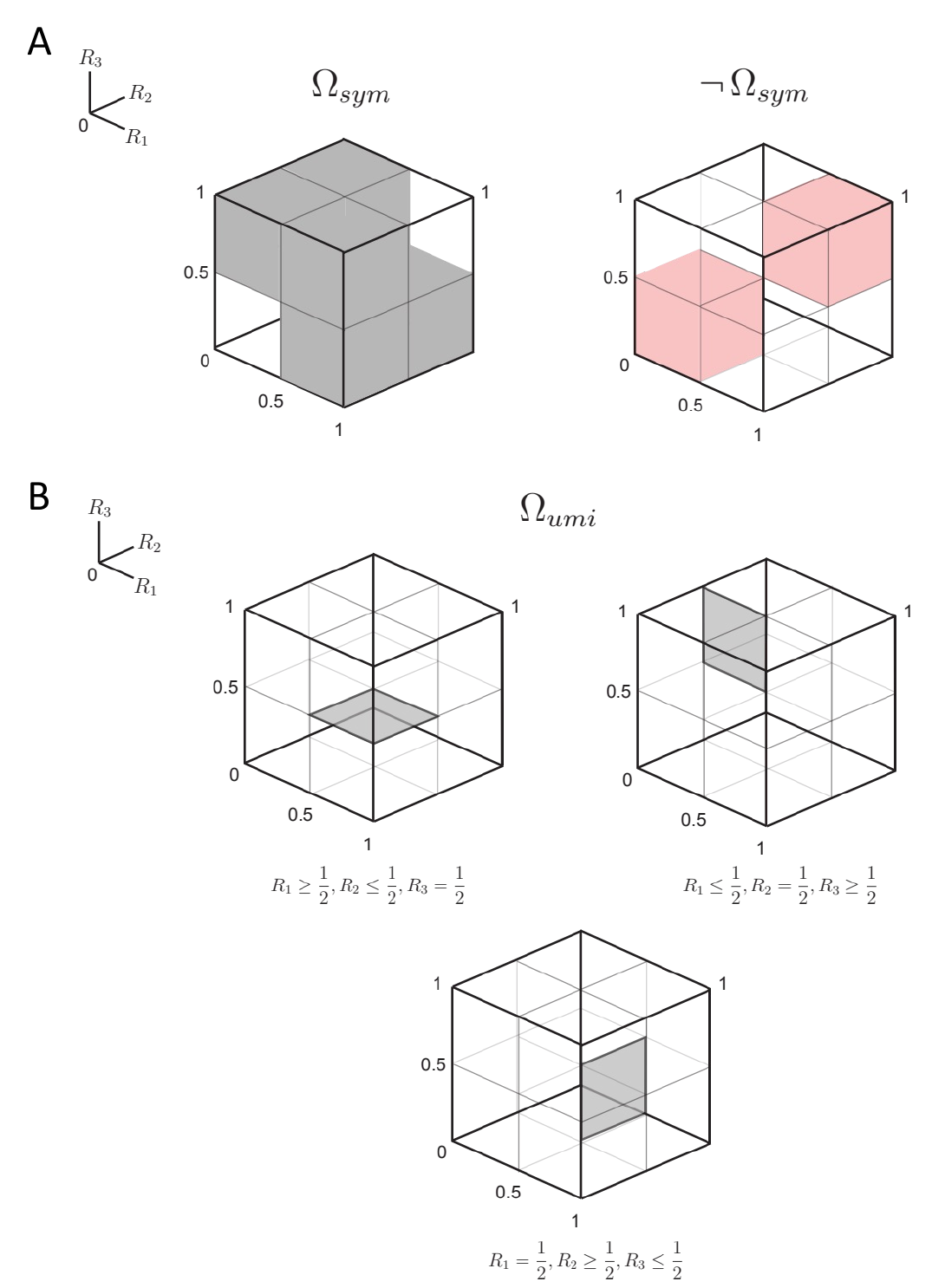}
    \caption{Panel A. The domain $\Omega_{\text{sym}}$ (eq. $\eqref{eq:2.5}$) in which the choice probabilities in a triplet are compatible with a symmetric distance (grey volumes, left) and its complement in $[0,1]^3$ (pink volumes, right). See Proposition $\ref{pr:1}$ for further details. Panel B. The domain $\Omega_{\text{umi}}$ (eq. \eqref{eq:2.9}, in which the choice probabilities in a triplet are compatible with an ultrametric distance. $\Omega_{\text{umi}}$ is the union of three orthogonal components (gray), each of which is a square of side length $\frac{1}{2}$ and its interior.}
    \label{fig:2}
\end{figure}

This domain is diagrammed in Figure~\ref{fig:2}A. \par
Note that the condition in Proposition~\ref{pr:1} is not sufficient for setwise compatibility across all stimuli, or even pointwise compatibility for sets with more than 3 stimuli. The chain of inequalities of \eqref{eq:2.2} or \eqref{eq:2.3} is the simplest of a series of necessary conditions for setwise compatibility: more generally, if, for any  $n$-cycle $(a_1, a_2, ..., a_n)$, the choice probabilities satisfy
\begin{equation}\label{eq:2.6}
     \left.\begin{aligned}
        R(a_1;a_2, a_n) \geq \frac{1}{2} \\
        R(a_n;a_1, a_{n-1}) \geq \frac{1}{2} \\
        \vdots \\
        R(a_2;a_3, a_1) \geq \frac{1}{2}
    \end{aligned}\right\}, \text{at least one inequality strict},
\end{equation}
then pointwise compatibility in $\{a_1, a_2, ..., a_n\}$ with a symmetric distance is impossible. For if \eqref{eq:2.6} holds, with the final inequality strict, then (generalizing \eqref{eq:2.4}), there would be a contradiction:
\begin{equation}\label{eq:2.7}
    d(a_1, a_2) \leq d(a_1, a_n) = d(a_n, a_1) \leq d(a_n, a_{n-1}) = d(a_{n-1}, a_n) \leq ... = d(a_2, a_3) < d(a_2, a_1).
\end{equation}

\begin{figure}
    \centering % gives better spacing than \begin{center}...\end{center}
    \includegraphics[scale=0.5]{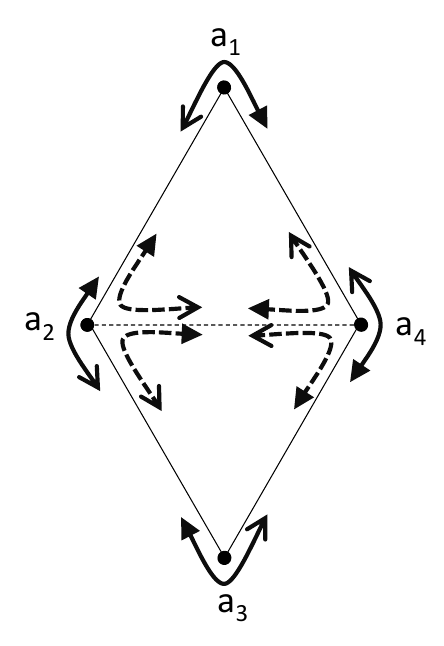}
    \caption{A configuration of four stimuli and four triads (solid arrows) corresponding to the comparisons of eq. \eqref{eq:2.6} for $n=4$. Dis-similarities among the four stimuli may fail to be pointwise compatible with a symmetric distance even though the dis-similarities among the two triplets (upper and lower triangles, and their associated solid and dashed arrows) are compatible with a symmetric distance. For further details, see text following eq. \eqref{eq:2.6}.}
    \label{fig:3}
\end{figure}

Figure~\ref{fig:3} illustrates the case of $n=4$, and shows that pointwise compatibility for comparisons among three stimuli does not imply pointwise compatibility for comparisons among four stimuli. In this diagram, eq. \eqref{eq:2.6} means that traversing the perimeter of the configuration yields four choice probabilities that are not all exactly $\frac{1}{2}$, but also not on both sides of $\frac{1}{2}$. As this leads to the contradiction \eqref{eq:2.7}, these dis-similarities among the four points and four triads cannot be pointwise compatible with symmetric distances. This scenario can occur even though Proposition~\ref{pr:1} holds for all triplets. To see this, set $D(a_2, a_4)$ and $D(a_4, a_2)$ to be larger than the any other dis-similarity.  This guarantees that the choice probabilities for triplets $\{(a_1;a_2, a_4), (a_2;a_4, a_1), (a_4;a_1, a_2)\}$ and $\{(a_3;a_4, a_2), (a_4;a_2, a_3), (a_2;a_3, a_4)\}$ include values that are both $< \frac{1}{2}$ and $> \frac{1}{2}$, which satisfies the conditions of Proposition~\ref{pr:1}.    

\subsubsection*{Ultrametric}

The motivation for considering compatibility with ultrametric distances begins with the observation that pointwise compatibility with a symmetric dis-similarity guarantees pointwise compatibility with a metric-space structure (Appendix~\ref{appendix1}). It is therefore natural to ask whether the dis-similarities have further properties associated with specific kinds of metric spaces. Ultrametric spaces \cite{Semmes2007} are one important such kind, as they abstract the notion of a hierarchical organization -- and have therefore been proposed as models for perceptual representations \cite{Kemp2008,Saxe2019, Treves1997, Tversky1986}. Points in an ultrametric space correspond to the terminal nodes of a tree, and the distance between two points corresponds to the height of their first common ancestor. Formally, a distance $d$ is said to satisfy the ultrametric inequality if, for any three points $x, y$, and $z$,
\begin{equation}\label{eq:2.8}
    d(x, y) \leq max\left(d(x, z), d(y, z)\right), 
\end{equation}
a condition that implies the triangle inequality (Appendix~\ref{appendix1}). Eq. \eqref{eq:2.8} states that for the three pairwise distances among three points, none can be strictly greater than the other two, i.e., at least two of the distances must be equal and the third cannot be longer.

\begin{proposition}[ordinal conditions for ultrametric]\label{pr:2}
    For a triplet composed of the stimuli $\{x, y, z\}$, the dis-similarities are pointwise compatible with an ultrametric distance if and only if the dis-similarities are symmetric and at least one of the following three hold:
    \begin{equation}\label{eq:2.9}
        \begin{aligned}
            R_1 \geq \frac{1}{2}, R_2 \leq \frac{1}{2}, R_3 = \frac{1}{2} \\
            R_1 = \frac{1}{2}, R_2 \geq \frac{1}{2}, R_3 \leq \frac{1}{2}\\
            R_1 \leq \frac{1}{2}, R_2 = \frac{1}{2}, R_3 \geq \frac{1}{2}
        \end{aligned}
    \end{equation}
    where, as before, $R_1 \triangleq R(x;y,z), R_2 \triangleq R(y;z,x), R_3\triangleq R(z;x,y).$ This domain, denoted $\Omega_{\text{umi}}$ is illustrated in Figure~\ref{fig:2}B.
\end{proposition}
\begin{proof}
    To show that \eqref{eq:2.9} and symmetry implies pointwise compatibility with an ultrametric: WLOG, assume the first set holds. Then 
    \begin{equation}\label{eq:2.10}
       \begin{aligned}
        R_1 \geq \frac{1}{2} \iff D(x, y) \leq D(x,z),\\
        R_2 \leq \frac{1}{2} \iff D(y, z) \geq D(y, x),\\
        R_3 = \frac{1}{2} \iff D(z, x) = D(z, y).
        \end{aligned}
    \end{equation}
    Using symmetry, this implies
    \begin{equation}\label{eq:2.11}
        D(x, y) \leq D(x, z) = D(y, z).
    \end{equation}
    Consider the transformation $d= G(D)$ where
    \begin{equation}\label{eq:2.12}
    G(D) = \left\{ \begin{aligned}
        0&,  D=0, \\
        1 + &\frac{D}{D+1}, D>0
    \end{aligned}\right.
    \end{equation}
    As this is strictly monotonic-increasing, it demonstrates \eqref{eq:2.8}. Moreover, $G$ satisfies the triangle inequality (see Appendix~\ref{appendix1}, eq. \eqref{eq:4.6}). Thus $G$ is a transformation that demonstrates setwise compatibility with an ultrametric distance on the triad, and, \textit{a fortiori}, pointwise compatibility. \par
    Conversely, if $d$ is an ultrametric distance and $D$ is pointwise compatible with $d$, then \eqref{eq:2.8}, which holds for $d$, must also hold for $D$. Therefore, of the three values $D(x, y), D(y, z), D(z, x)$, two must be equal and the third can be no larger than the others. WLOG, assume $D(y, z) = D(x, z)$ and $D(x, y)$ is no larger. Then the first alternative of \eqref{eq:2.9} holds.
\end{proof}
\begin{remark}
    If the conditions of Proposition~\ref{pr:2} hold for all triads composed of stimuli in $S$, the dis-similarities are setwise compatible with an ultrametric distance on $S$. This follows immediately from the proof of Proposition~\ref{pr:2}, as setwise compatibility with an ultrametric distance on each triplet is demonstrated via a transformation \eqref{eq:2.12} that is the same for all triplets.
\end{remark}

\subsubsection*{Addtree}

Like the ultrametric model, the additive similarity tree model \cite{Sattath1977} is a metric space model that places constraints on the properties of the distance, but these constraints are less-restrictive than the constraints of the ultrametric model (Appendix~\ref{appendix2}) and therefore may have greater suitability as a model for perceptual representations. In this model, here referred to as ``addtree'', the distance between two points is determined by a graph that has a tree structure, in which each link has a specified nonzero weight. The distance between two points is given by the total weight of the path that connects the points. Because of the requirement that the graph is a tree structure, there are no loops -- and this places constraints on the inter-relationships of the distances.\par
To determine the extent to which the dis-similarities implied by a set of triadic judgments are compatible with the distances in an addtree model, the starting point is a necessary and sufficient condition for distances in a metric space to be compatible with an addtree structure \cite{Buneman1974, Dobson1974, Sattath1977}. This condition, known as the ``four-point condition,'' is that given any four points $u, v, w$, and $x$,

\begin{equation}\label{eq:2.13}
    \text{none of the three quantities} \left\{ \begin{aligned}
        d(u, v) + d(w, x) \\
        d(u, w) + d(v, x) \\
        d(u, x) + d(v, w)
    \end{aligned} \right\} \text{ is strictly greater than the other two}.
\end{equation}
Put another way, of the three pairwise sums in eq. \eqref{eq:2.13}, two must be equal, and the third can be no larger. Appendix~\ref{appendix2} shows that this condition is weaker than the ultrametric inequality and stronger than the triangle inequality, and that a one-dimensional arrangement of points is always compatible with an addtree model. \par
Since the four-point condition is based on adding distances, we cannot apply it directly to dis-similarities -- as distances are linked to dis-similarity via an unknown monotonic function. However, there are conditions on the dis-similarities that are necessary for the four-point inequality to hold.
\begin{proposition}[necessary ordinal conditions for addtree]\label{pr:3}
    If, in a tent $\{z;a, b, c\}$, the inequalities
    \begin{equation}\label{eq:2.14}
        R(z;c,b) \leq \frac{1}{2}, R(z;c,a) \leq \frac{1}{2}, R(a;b,c) \leq \frac{1}{2}, R(b;a,c) \leq \frac{1}{2}.
    \end{equation}
    hold, along with strict inequalities 
    \begin{equation}\label{eq:2.15}
        \begin{aligned}
            & \left(R(z;c,b) < \frac{1}{2} \quad \text{or} \quad R(a;b,c) < \frac{1}{2} \right) \\
            & \text{and} \\
            & \left(R(z;c,a) < \frac{1}{2} \quad \text{or}\quad R(b;a,c) < \frac{1}{2} \right)
        \end{aligned}
    \end{equation}
    then the dis-similarities in the tent $\{z;a,b,c\}$ are not pointwise compatible with an addtree distance.
\end{proposition}
\begin{proof}
    Via eqs \eqref{eq:1.2} and \eqref{eq:1.3}, these conditions may be rewritten as inequalities among dis-similarities. The inequalities \eqref{eq:2.14} become
     \begin{equation}\label{eq:2.16}        
     	\begin{aligned}        D(z, c) \geq D(z, b) \quad\text{and}\quad D(z,c) \geq D(z, a) \quad\text{and}\\ 
     		       D(a, b) \geq D(c, a) \quad\text{and}\quad D(a, b) \geq D(b, c).       
     	 \end{aligned}    
      \end{equation}
    The inequalities \eqref{eq:2.15} become 
        \begin{equation}\label{eq:2.17}
        \begin{aligned}
            & \left(D(z, c) > D(z, b) \quad\text{or}\quad D(a,b) > D(a,c) \right) \\ 
            & \text{and}\\
            & \left(D(z, c) > D(z, a) \quad\text{or}\quad D(a,b) > D(b,c) \right).
        \end{aligned}
    \end{equation}
    Either alternative of the first portion of eq. \eqref{eq:2.17} leads to
    \begin{equation}\label{eq:2.18}
        D(z,c) + D(a,b) > D(z,b) + D(c,a).
    \end{equation}
    This follows from eqs. \eqref{eq:2.16} and \eqref{eq:2.17} via a term-by-term comparison of the two sides of eq. \eqref{eq:2.18}: \eqref{eq:2.16} implies that each term on the left of \eqref{eq:2.18} is no smaller than the corresponding term on the right, and \eqref{eq:2.17} implies that at least one of these inequalities is strict. Further, note that conditions \eqref{eq:2.16} are unchanged by swapping $a$ with $b$, and the two portions of \eqref{eq:2.17} are interchanged by this swap. Similarly, either alternative of the second portion of \eqref{eq:2.17} leads to
    \begin{equation}\label{eq:2.19}
        D(z,c) + D(a,b) > D(z,a) + D(c,b).
    \end{equation}
    If the dis-similarities $D$ were pointwise-compatible with an addtree distance $d$, then \eqref{eq:2.18} and \eqref{eq:2.19} would imply (via \eqref{eq:1.5}) that 
    \begin{equation}\label{eq:2.20}
        \begin{aligned}
            & d(z,c) + d(a,b) > d(z,b) + d(c,a) \\
            & \text{and} \quad \\
            & d(z,c) + d(a,b) > d(z,a) + d(c, b),
        \end{aligned}
    \end{equation}
    in contradiction to \eqref{eq:2.13}. 
\end{proof}
It is helpful to think of Proposition~\ref{pr:3} geometrically. If all dis-similarities are unequal, its conditions state that, for any tent, the largest dis-similarity in a tripod cannot be opposite the largest dis-similarity in the base. If some dis-similarities are equal, then  its conditions state that there is
 strict inequality for comparisons between one of the oppositely-paired dis-similarities and the other two oppositely-paired dis-similarities, either in the tripod or in the base (Figure~\ref{fig:1}C-E). \par
Proposition~\ref{pr:4} (Appendix~\ref{appendix3}) is a partial converse to Proposition~\ref{pr:3}: if the dis-similarities among four points are all unequal and the conjunction \eqref{eq:2.14} is false, we construct a monotonic transformation   that demonstrates setwise (and therefore pointwise) compatibility between the dis-similarities and an addtree distance. But note that even though the ``4-point'' condition \eqref{eq:2.13} on distances suffices to ensure a global addtree model, Proposition~\ref{pr:4} (Appendix~\ref{appendix3}) stops short of showing that the monotonic transformations needed to transform dis-similarities to distances for each quadruple of points can be made in a globally-consistent way -- though we do not have examples to the contrary. 

\section{Implementation}

In this section, we move from the results described above to procedures that can be applied to experimental data. 
Specifically, we develop indices, computable from a set of triadic dis-similarity judgments, that express the likelihood that these judgments are compatible with an underlying symmetric distance (the index \(I_{\text{sym}}\)), an ultrametric distance (the index \(I_{\text{umi}}\)), 
and an addtree distance (the index \(I_{\text{add}}\)). 
\(I_{\text{sym}}\) is the simplest and illustrates the basic strategy; \(I_{\text{umi}}\) and \(I_{\text{add}}\) each build on this strategy in different ways.

Common to all three indices is the hurdle that the choice probabilities \(R(r;x,y)\) are unknown, and must be estimated from experimental data. 
We denote the number of such trials in which the triad \((r;x,y)\) is presented by \(N(r;x,y)\), and the number of trials in which the participant judges \(x\) as more similar to \(r\)  than \(y\) is to \(r\)  by \(C(r;x,y)\). This provides a na\"ive estimate  of the choice probability: 

\begin{equation}\label{eq:3.1}
R_{\text{obs}}(r; x, y) = \frac{C(r;x,y)}{N(r;x,y)}.
\end{equation}
To ensure that the estimated choice probabilities obey \eqref{eq:1.1}, we consider any presentation of a triad \((r;x,y)\) to also be a presentation of its mirror, so that any response contributes to
\(C(r;x,y)\) (if \(x\) is judged as more similar to \(r\) than \(y\) is to \(r\)) or to \(C(r;y,x)\) (if otherwise). With this convention, 

\begin{equation}\label{eq:3.2}
    C(r;x,y) + C(r;y,x) = N(r;x,y) = N(r;y,x)
\end{equation}
and
\begin{equation}\label{eq:3.3}
    R_{\text{obs}}(r;x,y) + R_{\text{obs}}(r;y,x) = 1.
\end{equation}

However, \(R_{\text{obs}}(r;x,y)\) is only an estimate of the choice probability \(R(r;x,y)\). 
Thus, rather than determine whether the naïve estimates \(R_{\text{obs}}(r;x,y)\) satisfy the conditions of Propositions~\ref{pr:1}, \ref{pr:2}, \ref{pr:3} and \ref{pr:4}, we take a Bayesian approach: given the observed data, what is the likelihood that the underlying choice probabilities \(R(r;x,y)\) are consistent with the requisite inequalities?

The Bayesian approach requires a prior for the distribution of the set of choice probabilities. 
For the assessment of compatibility with symmetry between reference and comparison stimuli, the prior
assumes that the choice probabilities are independently drawn from a specified
univariate distribution (see below). 
For the assessment of compatibility with an ultrametric or an addtree structure, we then
modify the prior by eliminating combinations of choice probabilities that are inconsistent with
symmetry.

We implement this strategy using beta functions as a family of priors for choice probabilities.
Since we have assumed that an experimenter randomly chooses whether to present a triad or its mirror, 
\eqref{eq:1.1} means that appropriate priors for the distribution of choice probabilities, 
\(p(R)\), should satisfy \(p(R) = p(1-R)\). (This ``symmetry'' refers to interchangeability of the two comparison 
stimuli \(x\) and \(y\) in the triad \((r;x,y)\); not to interchange between reference and comparison, which is our focus.).
Thus, it suffices to consider the family of priors

\begin{equation}\label{eq:3.4}
p_{a}(R) \triangleq \frac{1}{B(a,a)} R^{a-1}(1-R)^{a-1},
\end{equation}
where \(B(a,a)\) is the symmetric specialization of the beta function 
\(B(a,b)\), defined in the standard fashion by
\begin{equation}\label{eq:3.5}
B(a,b) \triangleq \int\limits_0^1 u^{a-1}(1-u)^{b-1} du = \frac{\Gamma(a)\Gamma(b)}{\Gamma(a+b)}.
\end{equation}

As the parameter \(a\) varies over the positive reals, the shape of the prior \eqref{eq:3.4} 
changes from heavily weighted near the extremes of \(R=0\) and \(R=1\) (\(a\) near zero), 
to heavily weighted near \(R=1/2\) (\(a \gg 1\)), thus capturing scenarios ranging from those in which most judgments are near certainty, to those in which most judgments are close to equivocal.
In between, \(a=1\) corresponding to the scenario in which choice probabilities are 
evenly distributed in \([0,1]\).

While this choice of priors is fundamentally a heuristic one, it has both theoretical motivations and practical advantages. 
The beta distribution is the univariate case of the Dirichlet distribution, a distribution that has theoretical justification as a prior for multivariate probabilities \cite{Ferguson1973} -- so the beta distribution is a natural choice for a decision model in which an internal multivariate state is collapsed to a binary choice. 
In addition to encompassing a range of qualitatively different shapes, the family includes important special cases.
For \(a=1\), it is flat. For \(a=1/2\), \eqref{eq:3.4} is the Jeffreys prior for probabilities in 
the \([0,1]\) interval, i.e., the ``uninformative'' prior in an information-theoretic sense \cite{Jeffreys1946}. 
At \(a=0\), the \(R\)-dependence of \eqref{eq:3.4}, \(\frac{1}{R(1-R)}\), is the (improper) Haldane prior,
the unique prior for which the expected value of \(R\) is equal to the naïve estimate \(\frac{C}{N}\) \cite{Jeffreys1946}.
At a practical level, this family of priors has shown its utility as a prior distribution for choice
probabilities in the context of improving the estimation of psychometric functions \cite{Schutt2016}.
Finally, the key integrals involving \eqref{eq:3.4} are easy to compute 
(e.g., eq. \eqref{eq:3.7}, leading to eq. \eqref{eq:3.8}), providing a computationally efficient means to select the parameter
\(a\) that is most appropriate given the observed data. 
Note also that the prior \eqref{eq:3.4} is symmetric about \(R=1/2\), so that Bayesian inference of \(R\)
from an \(R_{\text{obs}}\) that satisfies \eqref{eq:3.3} will necessarily satisfy \eqref{eq:1.1}.

We use a maximum likelihood approach to determine the parameter \(a\) 
(alternative choices for \(a\) will be considered in the example applications below).
Specifically, we maximize the likelihood of the observed set of responses across the entire experiment, 
assuming that the individual responses for the \(i\)-th triad \((r_i; x_i, y_i)\) are independently drawn from a Bernoulli distribution with parameter
\(R_i = R(r_i; x_i, y_i)\), and that each \(R_i\) is independently drawn 
from the distribution \eqref{eq:3.4}. That is, for a given \(R_i\),
the probability that the subject reports \(D(r,x) < D(r, y)\) in \(C_i\) of \(N_i\) 
presentations is 

\begin{equation}\label{eq:3.6}
p(C_i | R_i, N_i) = \binom{N_i}{C_i} R_i^{C_i}(1-R_i)^{N_i - C_i}. 
\end{equation}
Integrating over the prior \eqref{eq:3.4} for \(R_i\) yields the probability of observing
\(C_i\) reports of \(D(r,x) < D(r, y)\) in \(N_i\) presentations, given the parameter \(a\):
\begin{equation}\label{eq:3.7}
\begin{aligned}
p(C_i | N_i, a) & = \binom{N_i}{C_i} \int\limits^1_0 \left(\frac{1}{B(a,a)}R^{a-1}(1-R_i)^{a-1} \right) R_i^{C_i}(1-R_i)^{N_i - C_i} dR_i  \\
& = \binom{N_i}{C_i} \frac{B(a+C_i, a+N_i-C_i)}{B(a,a)} 
\end{aligned}
\end{equation}
Making use of the independence of each triad yields the overall log-likelihood:

\begin{equation}\label{eq:3.8}
LL(a) = \log{\left( \prod_i p(C_i | N_i, a) \right)} = \log{K} + \sum_i \log{\frac{B(a+C_i, a+N_i-C_i)}{B(a,a)} },
\end{equation}
where \(K\) is a combinatorial factor independent of \( a \), and the sum ranges over all triads.
Maximizing \eqref{eq:3.8} then determines the value of \(a\) for 
which independent draws of choice probabilities are most likely to yield the experimental data.

We then use this value of \(a\) to determine the posterior likelihood for choice probabilities within
a triplet or a tent. 
That is, for a set \(\overrightarrow{R}\) of choice probabilities \(R_i\),
the prior is

\begin{equation}\label{eq:3.9}
P(\overrightarrow{R}) = \prod_i \left( \frac{1}{B(a,a)} R_i^{a-1}(1-R_i)^{a-1}\right).
\end{equation}
Via Bayes rule, this prior determines the posterior likelihood of a set of choice probabilities:

\begin{equation}\label{eq:3.10}
p(\overrightarrow{R} | \overrightarrow{C}, \overrightarrow{N}) = \frac{p(\overrightarrow{C} | \overrightarrow{R}, \overrightarrow{N}) P(\overrightarrow{R})}{p(\overrightarrow{C}|\overrightarrow{N})}.
\end{equation}
where \(\overrightarrow{C}\) denotes the responses \(C_i\) to each of the triads,
\(\overrightarrow{N}\) denotes the number of times that each triad \(N_i\) was presented,
and \(p(\overrightarrow{C} | \overrightarrow{R}, \overrightarrow{N})\) denotes a product of the terms specified by eq. \eqref{eq:3.6}.
\(p(\overrightarrow{C}|\overrightarrow{N})\), the \textit{a priori} probability of
\(\overrightarrow{C}\) and \(\overrightarrow{N}\), is unknown, but as is standard in Bayesian analyses, it is eliminated when likelihood ratios are calculated.

\subsubsection*{Symmetry}

We now use this machinery to estimate the probability that the choice probabilities underlying a
set of observations are compatible with a symmetric distance. We first focus on the data within a single
triplet, and then consider extension of the analysis to the entire dataset.

For the analysis within a single triplet (i.e., given the observations \(\overrightarrow{C}_T\) and \(\overrightarrow{N}_T\) for the triads in a triplet \(T\)), 
we compare the posterior likelihood for the corresponding choice probabilities \(\overrightarrow{R}_T\)
for which the inequalities of Proposition~\ref{pr:1} hold, to the likelihood that the observations
result from choice probabilities within the entire space of choice probabilities.
We denote the posterior likelihoods by
\begin{equation}\label{eq:3.11}
L_{\text{sym}}(T) = \int\limits_{\Omega_{\text{sym}}} p(\overrightarrow{R}_T | \overrightarrow{C}_T, \overrightarrow{N}_T) d\overrightarrow{R}_T,
\end{equation}
and
\begin{equation}\label{eq:3.12}
L(T) = \int\limits_{\Omega} p(\overrightarrow{R}_T | \overrightarrow{C}_T, \overrightarrow{N}_T) d\overrightarrow{R}_T,
\end{equation}
and their ratio by
\begin{equation}\label{eq:3.13}
LR_{\text{sym}}(T) = \frac{L_{\text{sym}}(T)}{L(T)},
\end{equation}
where \(\Omega\) is the space in which all choice probabilities in the triplet \(T\) 
range independently over \([0, 1]\), and \(\Omega_{\text{sym}}\) is the subset of the space consistent with the conditions of Proposition~\ref{pr:1}, namely, the cube \([0, 1]^3\) from which 
\([0, \frac{1}{2}]^3\) and \([\frac{1}{2}, 1]^3\) are removed (eq. \eqref{eq:2.5}, Figure~\ref{fig:2}A)).

Both quantities \eqref{eq:3.11} and \eqref{eq:3.12} can be expressed in terms of the prior via Bayes’ rule:
\begin{equation}\label{eq:3.14}
\begin{aligned}
p(\overrightarrow{R}_T | \overrightarrow{C}_T, \overrightarrow{N}_T)  & = \frac{p(\overrightarrow{C}_T | \overrightarrow{R}_T, \overrightarrow{N}_T) P(\overrightarrow{R}_T)}{p(\overrightarrow{C}_T | \overrightarrow{N}_T)} \\
& = \frac{1}{p(\overrightarrow{C}_T, \overrightarrow{N}_T)} \prod_{i \in T} \left( \binom{N_i}{C_i} R_i^{C_i} (1-R_i)^{N_i - C_i} p_{a}(R_i) \right),
\end{aligned}
\end{equation}
where the prior for each choice probability, \(p_{a}(R_i)\), is given by \eqref{eq:3.4}. 
The likelihood ratio \eqref{eq:3.13} \(LR_{\text{sym}}(T)\) is thus
\begin{equation}\label{eq:3.15}
LR_{\text{sym}}(T) = \frac{ \int\limits_{\Omega_{\text{sym}}} \prod_{i \in T} \left( R_i^{C_i} (1-R_i)^{N_i - C_i} p_{a}(R_i)\right) d\overrightarrow{R}_T }{\prod_{i \in T} \left( \int\limits_{0}^{1} R_i^{C_i} (1-R_i)^{N_i - C_i} p_{a}(R_i) dR_i\right)} 
\end{equation}
Based on \eqref{eq:2.5}, the numerator of \eqref{eq:3.15} is a combination of three terms,
\begin{equation}\label{eq:3.16}
\begin{aligned}
\int\limits_{\Omega_{\text{sym}}} \prod_{i \in T} \left( R_i^{C_i} (1-R_i)^{N_i - C_i} p_{a}(R_i)\right) d\overrightarrow{R}_T  = & \int\limits_{[0, 1]^3} \prod_{i \in T} \left( R_i^{C_i} (1-R_i)^{N_i - C_i} p_{a}(R_i)\right) d\overrightarrow{R}_T \\
& - \int\limits_{[\frac{1}{2}, 1]^3} \prod_{i \in T} \left( R_i^{C_i} (1-R_i)^{N_i - C_i} p_{a}(R_i)\right) d\overrightarrow{R}_T \\
& - \int\limits_{[0, \frac{1}{2}]^3} \prod_{i \in T} \left( R_i^{C_i} (1-R_i)^{N_i - C_i} p_{a}(R_i)\right) d\overrightarrow{R}_T,
\end{aligned}    
\end{equation}
each of which can be written in terms of incomplete beta functions:
\begin{equation}\label{eq:3.17}
\begin{aligned}
\int\limits_{v}^{w} R^C (1-R)^{N-C}p_a(R)dR & = \frac{1}{B(a,a)} \int\limits_{v}^{w}R^{C+a-1}(1-R)^{N-C+a-1} dR \\
& = \frac{B(w; a+C, a+N -C) - B(v; a+C, a+N-C)}{B(a,a)}
\end{aligned}
\end{equation}
where 
\begin{equation}\label{eq:3.18}
B(w; a,b) = \int\limits_{0}^{w} u^{a-1}(1-u)^{b-1} du.
\end{equation}
Each factor of the denominator of \eqref{eq:3.15} can also be expressed in terms of beta functions:
\begin{equation}\label{eq:3.19}
\begin{aligned}
\int\limits_{0}^{1} R_i^{C_i} (1-R_i)^{N_i - C_i} p_{a}(R_i) dR_i & = \frac{1}{B(a,a)} \int\limits_{0}^{1} \left( R_i^{C_i} (1-R_i)^{N_i - C_i} R_i^{a-1}(1-R_i)^{a-1} \right)dR_i \\
& = \frac{1}{B(a,a)} \int\limits_{0}^{1} \left( R_i^{C_i+a-1} (1-R_i)^{N_i - C_i +a -1} \right)dR_i \\
& = \frac{B(a+C_i, a+N_i-C_i)}{B(a,a)}
\end{aligned}
\end{equation}
Combining \eqref{eq:3.15} through \eqref{eq:3.19} yields

\begin{equation}\label{eq:3.20}
LR_{\text{sym}}(T) = 1 - \prod_{i \in T} \left( 1 - \frac{B(\frac{1}{2}; a+C_i, a+N_i-C_i)}{B(a+C_i, a+N_i-C_i)}\right) - \prod_{i \in T} \left( \frac{B(\frac{1}{2}; a+C_i, a+N_i-C_i)}{B(a+C_i, a+N_i-C_i)}\right).
\end{equation}
Many software packages (e.g., MATLAB) provide the normalized beta function

\begin{equation}\label{eq:3.21}
B_{\text{norm}}(w;a,b) \triangleq \frac{B(w;a,b)}{B(a,b)} = \frac{1}{B(a,b)} \int\limits_0^w u^{a-1}(1-u)^{b-1} du,
\end{equation}
which simplifies \eqref{eq:3.20} to

\begin{equation}\label{eq:3.22}
\begin{aligned}
LR_{\text{sym}}(T) = \quad & 1 - \prod_{i \in T} \left( 1 - B_{\text{norm}}(\frac{1}{2}; a+C_i, a+N_i-C_i) \right) \\
& - \prod_{i \in T} \left( B_{\text{norm}}(\frac{1}{2}; a+C_i, a+N_i-C_i)\right).
\end{aligned}
\end{equation}

To extend this analysis to the entire set of observations \(\overrightarrow{C}\) and \(\overrightarrow{N}\), 
we note that the conditions of Proposition~\ref{pr:1} concern the triads within a single triplet, and that the triads within distinct triplets are non-overlapping. 
Thus, the posterior likelihood for choice probabilities consistent with Proposition~\ref{pr:1} at all
triplets, \(L_{\text{sym}}(\overrightarrow{C}, \overrightarrow{N})\), is an integral over a product space
\(\Omega^*_{\text{sym}} = \prod_{T \in \mathfrak{trip}} \Omega_{\text{sym}}(T)\) with one component for each triplet.
The integral over the component corresponding to the triplet \(T\) is given by \eqref{eq:3.11}. 
The unrestricted posterior likelihood \(L(\overrightarrow{C}, \overrightarrow{N})\) 
is an integral over a product space \(\Omega^*\) of choice probabilities for all triads. 
Since each triad is a member of exactly one triplet, this integral also factors into a product
of terms corresponding to \eqref{eq:3.12}, one for each triplet.
Consequently, the likelihood ratio for the entire dataset is given by

% TODO: fix overflow on the right
\begin{equation}\label{eq:3.23}
\begin{aligned}
LR_{\text{sym}}(\overrightarrow{C}, \overrightarrow{N}) = \frac{L_{\text{sym}}(\overrightarrow{C}, \overrightarrow{N})}{ L(\overrightarrow{C}, \overrightarrow{N})} = 
 \prod_{T \in \mathfrak{trip}}  \left( 1 - \prod_{i \in T} \left( 1 - B_{\text{norm}}(\frac{1}{2}; a+C_i, a+N_i-C_i) \right) \right.& \\
 \left. - \prod_{i \in T} \left( B_{\text{norm}}(\frac{1}{2}; a+C_i, a+N_i-C_i) \right) \right) &
\end{aligned}
\end{equation}
In view of \eqref{eq:3.10}, \(\frac{LR_{\text{sym}}}{1-LR_{\text{sym}}}\) (and the analogous likelihood ratios for the indices below) may be interpreted as a Bayes factor for the triplet: 
the ratio between the likelihood of the data \((\overrightarrow{C}, \overrightarrow{N})\) given that
all of the choice probabilities are pointwise-compatible with symmetry, vs. the likelihood of the data
given that the choice probabilities are drawn from the complementary portion of the prior.

While it is thus tempting to interpret \eqref{eq:3.23} as a global measure of compatibility with symmetry,
there is an important caveat: as mentioned above, pointwise compatibility at each triplet does not
guarantee compatibility across the entire stimulus space \(S\). To emphasize this point, we normalize
\(LR_{\text{sym}}\) by the number of triplets, yielding a quantity that can rigorously be interpreted an average
measure of pointwise compatibility within triplets:

\begin{equation}\label{eq:3.24}
I_{\text{sym}}(\overrightarrow{C}, \overrightarrow{N}) = \frac{1}{\#(\mathfrak{trip})} \log LR_{\text{sym}}(\overrightarrow{C}, \overrightarrow{N}).
\end{equation}
Values of \( I_{\text{sym}}(\overrightarrow{C}, \overrightarrow{N})\) that are close to zero indicate that nearly all of the posterior distribution of choice probabilities lies within the product space
\(\Omega^*_{\text{sym}} = \prod_{T \in \mathfrak{trip}} \Omega_{\text{sym}}(T)\)
(i.e., is pointwise compatible with a symmetric distance); progressively more negative values indicate that the posterior shifts into its complement
\( \Omega \backslash \Omega_{\text{sym}} \) in which symmetry is necessarily violated.

A useful benchmark in interpreting \( I_{\text{sym}}(\overrightarrow{C}, \overrightarrow{N})\) 
is its value in the absence of any data (i.e., \( \overrightarrow{C} = \overrightarrow{N} = \overrightarrow{0}\)).
In that case, each of the normalized beta functions has a value of 
\(B_{\text{norm}}(\frac{1}{2}; a, a) = \frac{1}{2}\), so

\begin{equation}\label{eq:3.25}
I_{\text{sym}}(\overrightarrow{0}, \overrightarrow{0}) = \frac{1}{\#(\mathfrak{trip})} \log \prod_{T \in \mathfrak{trip}} \left( 1 - \left(  1 - \frac{1}{2}\right)^3 - \left( \frac{1}{2} \right)^3 \right) = \log \frac{3}{4} \approx -0.2877.
\end{equation}
Thus, values of \(I_{\text{sym}}(\overrightarrow{C}, \overrightarrow{N})\) greater than
\(\log \frac{3}{4}\) are more compatible with symmetry than an index derived from choice probabilities drawn randomly from the prior.

Note also that deviations from this \textit{a priori} value can only be driven by triplets in which there are
observations for at least two of the triads.
This follows from the fact that \(B_{\text{norm}}(\frac{1}{2}; a, a) = \frac{1}{2}\), so that if
only one triad \(k \in T\) has a nonzero number of observations,

\begin{equation}\label{eq:3.26}
\begin{aligned}
& 1 - \prod_{i \in T} \left( 1 - B_{\text{norm}}(\frac{1}{2}; a+C_i, a+N_i-C_i) \right)
- \prod_{i \in T} \left( B_{\text{norm}}(\frac{1}{2}; a+C_i, a+N_i-C_i) \right) \\
& = 1 - \frac{1}{2} \cdot \frac{1}{2} \cdot \left( 1 - B_{\text{norm}}(\frac{1}{2}; a+C_k, a+N_k - C_k) \right) \\
& \quad - \frac{1}{2} \cdot \frac{1}{2} \left( B_{\text{norm}}(\frac{1}{2}; a+C_k, a+N_k - C_k) \right) \\
& = \frac{3}{4}.
\end{aligned}
\end{equation}
This is a reassuring result: we should only be able to make inferences about the structure of the 
dis-similarity judgments if there is experimental data about more than one triad within a triplet. 
With data about only one triad, knowing the sign of the comparison is useless since this sign is 
arbitrarily determined by how the triad is labeled, i.e., \((r; x, y)\) vs \((r; y, x)\).

\subsubsection*{Ultrametric}

We use a parallel strategy to apply the present approach to evaluate compatibility with an
ultrametric distance. Necessary and sufficient criteria for compatibility with an ultrametric distance
concern a single triplet (Proposition~\ref{pr:2}), as is the case for symmetry (Proposition~\ref{pr:1}).
Thus, we first focus on a single triplet and then extend to the entire dataset.

An immediate difficulty arises: the conditions of Proposition~\ref{pr:2} are only satisfied on a set of measure zero, since at least one of the \(R_i\) must be exactly equal to \(\frac{1}{2}\).
So a Bayesian analysis based on a continuous prior (including the beta-function prior of \eqref{eq:3.4}) 
will always lead to a likelihood ratio of zero, since a posterior derived from a continuous prior via Bayes rule and Bernoulli trials cannot have a discrete mass at \(\frac{1}{2}\).

It is nevertheless possible to capture the spirit of ultrametric behavior in a rigorous way, and at
the same time, address a way in which beta-function prior may be unrealistic. 
To do this, we posit that there is a fraction of triads for which the underlying
choice probability is exactly \(\frac{1}{2}\).
For example, such triads could consist of stimuli for which there is no basis for comparison: is a toothbrush or a mountain more similar to an orange? 
But we don’t know, \textit{a priori}, what fraction \(h\) of the triads have this property.
To take this into account, we generalize the prior for choice probabilities to be a sum of two 
components: one component is the beta-function prior used above \eqref{eq:3.4}, normalized to \(1-h\);
the second component is a point mass at \(\frac{1}{2}\), normalized to \(h\):

\begin{equation}\label{eq:3.27}
p_{a,h}(R) \triangleq (1-h)p_a(R) + h\delta(R-\frac{1}{2}) = \frac{1-h}{B(a,a)}R^{a-1}(1-R)^{a-1} + h\delta(R - \frac{1}{2}).
\end{equation}
For fixed \(h\), the parameter \(a\) is determined by maximizing the likelihood of the observed responses
(generalizing \eqref{eq:3.8}):
\begin{equation}\label{eq:3.28}
\begin{aligned}
LL(a;h) & = \log \left( \prod_i p(C_i | N_i, a;h)\right) \\ 
& = \log K + \sum_i \log \left( (1-h) \frac{B(a+C_i, a+N_i-C_i)}{B(a,a)} + 2^{-N_i}h\right),
\end{aligned}
\end{equation}

With this prior, we can then determine the likelihood ratio as a function of \(h\), in the limit that 
\(h\) approaches zero.
We anticipate (and will confirm below) that for small values of \(h\), the likelihood ratio
will be proportional to \(h\), since the mass in the prior at \(\frac{1}{2}\) is proportional to \(h\).
This proportionality serves as an index of compatibility with the ultrametric property: how quickly does the likelihood ratio increase, if a small fraction of the underlying choice probabilities are exactly \(\frac{1}{2}\).
An alternative approach (not taken here) is that if the experimental dataset suggests that a prior
\(p_{a,h}(R)\) with \(h>0\) is a substantially better fit to the distribution of choice probabilities than
\(p_{a,0}(R)\), this prior can be used directly to calculate a likelihood ratio, and the best-fitting value of
\(h\) then provides an additional descriptor of the dataset.

To implement this strategy for individual triads, we define
\begin{equation}\label{eq:3.29}
L_{\text{umi}}(T;h) = \int\limits_{\Omega_{\text{umi}}} p(\overrightarrow{R}_T | \overrightarrow{C}_T, \overrightarrow{N}_T) d \overrightarrow{R}_T \quad,
\end{equation}
and consider the likelihood ratio,

\begin{equation}\label{eq:3.30}
LR_{\text{umi}}(T;h)= \frac{L_{\text{umi}}(T;h)}{L_{\text{sym}}(T;h)},
\end{equation}
where \(\Omega_{\text{umi}}\) is the space in which all choice probabilities in the triplet \(T\)
satisfy the conditions \eqref{eq:2.9} of Proposition~\ref{pr:2} (Figure~\ref{fig:2}B). 
Note that since ultrametric behavior is only defined with respect to a
symmetric distance, the denominator in the likelihood ratio is given by \eqref{eq:3.11} 
(but using the more general prior \eqref{eq:3.27}), 
which only considers combinations of choice probabilities in \(\Omega_{\text{sym}}\).

In view of Proposition~\ref{pr:2}, \(\Omega_{\text{umi}}\) is a union of regions defined 
by combinations of the signs of 
\(R_i-\frac{1}{2}\), for \(i\in \{1,2,3\}\) (Figure~\ref{fig:2}B), and 
Proposition~\ref{pr:1} implies that the same is true for \(\Omega_{\text{sym}}\).
Thus, numerator and denominator of \eqref{eq:3.30} can be written
\begin{equation}\label{eq:3.31}
\begin{aligned}
    L_q(T; h) = \int\limits_{[0, 1]^k} & V_q \left( sgn(R_1 - \frac{1}{2}), \dots, sgn(R_k - \frac{1}{2})\right) R_1^{C_1} (1-R_1)^{N_1-C_1} \bullet \dots \\ 
    & \bullet R_k^{C_k} (1-R_k)^{N_k - C_k} \bullet p_{a,h}(R_1) \bullet \dots \bullet p_{a,h}(R_k) dR_1 \dots dR_k
\end{aligned}
\end{equation}
where \(k=3\), \(R_i\) are the choice probabilities in the triad \(T\), \(C_i\) and \(N_i\) tally the responses, and \(V_q(\sigma_1, \dots , \sigma_k)\), 
for \(q=\text{umi}\) or \(q=\text{sym}\), is an indicator function over the configuration of signs, which defines either \(\Omega_{\text{umi}}\)
or \(\Omega_{\text{sym}}\). (We write \eqref{eq:3.31} in a general form; \(k=3\) suffices for analyzing ultrametric behavior and symmetry but the analysis of addtree behavior will require
\(k=6\).)

For the numerator of \eqref{eq:3.30}, Proposition~\ref{pr:2} dictates that the nonzero values of
\(V_{\text{umi}}\) are:

\begin{equation}\label{eq:3.32}
\begin{aligned}
V_{\text{umi}}(+1, -1, 0) & = 1, \quad \text{corresponding to} \quad R_1 >\frac{1}{2}, R_2 <\frac{1}{2}, R_3 = \frac{1}{2} \\
V_{\text{umi}}(0, +1, -1) &= 1, \quad \text{corresponding to} \quad R_1 = \frac{1}{2}, R_2 >\frac{1}{2}, R_3 < \frac{1}{2}\\
V_{\text{umi}}(-1, 0, +1) &= 1, \quad \text{corresponding to} \quad R_1 < \frac{1}{2}, R_2 =\frac{1}{2}, R_3 > \frac{1}{2} \\
V_{\text{umi}}(0, 0, 0) &= 1, \quad \text{corresponding to} \quad R_1 = 
\frac{1}{2}, R_2 =\frac{1}{2}, R_3 = \frac{1}{2} .
\end{aligned}
\end{equation}
All other values of \(V_{\text{umi}}( \overrightarrow{\sigma} )\) are zero, since either they don’t correspond to any of the conditions, or to exactly two of those conditions.
The latter is impossible, as it would require two equalities and one strict
inequality among the dis-similarities.

For the denominator of \eqref{eq:3.30}, we find the nonzero values of \(V_{\text{sym}}\) from Proposition~\ref{pr:1}:

\begin{equation}\label{eq:3.33}
\begin{aligned}
n_{half} =0: & \quad V_{\text{sym}}(\pm1, \pm1, \mp1) = V_{\text{sym}}(\pm1, \mp1, \pm1) = V_{\text{sym}}(\mp1, \pm1, \pm1) = 1 \\
n_{half} =1: & \quad V_{\text{sym}}(\pm1, \mp1, 0) = V_{\text{sym}}(0, \pm1, \mp1) = V_{\text{sym}}(\mp1, 0, \pm1) = 1\\
n_{half} =3: & \quad V_{\text{sym}}(0, 0, 0) = 1,
\end{aligned}
\end{equation}
where, as before, \(n_{half}\) is the number of \(R_i\) that are exactly equal to \(\frac{1}{2}\).

To establish the behavior of the likelihood ratio \eqref{eq:3.30} as \(h \rightarrow 0\), we use
\eqref{eq:3.27} to isolate the dependence of integrals \eqref{eq:3.31} on \(h\). This is a polynomial:

\begin{equation}\label{eq:3.34}
L_q(T; h) = \sum_{\overrightarrow{\sigma}} h^{Z(\overrightarrow{\sigma})} (1-h)^{k-Z(\overrightarrow{\sigma})} V_q(\sigma_1, \dots, \sigma_k) W(\sigma_1; C_1, N_1) \bullet \dots \bullet W(\sigma_k; C_k, N_k),
\end{equation}
where the sum is over all \(3^k\) assignments of the elements of \(\overrightarrow{\sigma} = (\sigma_1, \dots, \sigma_k)\) to \(\{-1,0,+1\}\), 
\(Z(\overrightarrow{\sigma})\) is the number of entries in \(\overrightarrow{\sigma}\)
that are equal to zero (each such entry incurring a factor of \(h\)), 
and \(W(\sigma; C, N)\) is the integral of the prior \eqref{eq:3.27}, weighted by the experimental data, over one segment of the domain:

\begin{equation}\label{eq:3.35}
    W(\sigma; C, N) = \left\{ 
    \begin{aligned}
        \int\limits_0^{1/2} R^C (1-R)^{N-C} p_a(R) dR, \quad & \sigma = -1\\
        \left. R^C (1-R)^{N-C} \right\rvert_{R=\frac{1}{2}} ,  \quad & \sigma = 0\\
        \int\limits_{1/2}^{1} R^C (1-R)^{N-C} p_a(R) dR, \quad & \sigma = +1\\
    \end{aligned}
    \right.
\end{equation}
These evaluate to

\begin{equation}\label{eq:3.36}
    W(\sigma; C, N) = \left\{ 
    \begin{aligned}
        \frac{1}{B(a,a)}B(\frac{1}{2}; a+C, a+N-C), \quad & \sigma = -1\\
        \frac{1}{2^N} ,  \quad & \sigma = 0\\
         \frac{1}{B(a,a)}\left( 1 -B(\frac{1}{2}; a+C, a+N-C) \right), \quad & \sigma = +1\\
    \end{aligned}
    \right.
\end{equation}
Consequently,
\begin{equation}\label{eq:3.37}
L_q(T; 0) = \sum_{Z(\overrightarrow{\sigma})=0} V_q(\sigma_1, \dots, \sigma_k) W(\sigma_1; C_1, N_1) \bullet \dots \bullet W(\sigma_k; C_k, N_k),
\end{equation}

\begin{equation}\label{eq:3.38}
\begin{aligned}
\frac{d}{dh} L_q(T; h) = \sum_{\overrightarrow{\sigma}} &
\left( Z(\overrightarrow{\sigma}) h^{Z(\overrightarrow{\sigma}) - 1} (1-h)^{k-Z(\overrightarrow{\sigma})}
- (k-Z(\overrightarrow{\sigma}) )h^{Z(\overrightarrow{\sigma})} (1-h)^{k-Z(\overrightarrow{\sigma})-1} \right) \bullet \\
&V_q(\sigma_1, \dots, \sigma_k) W(\sigma_1; C_1, N_1) \bullet \dots \bullet W(\sigma_k; C_k, N_k),
\end{aligned}
\end{equation}
and
\begin{equation}\label{eq:3.39}
\begin{aligned}
\left.\frac{d}{dh} L_q(T; h) \right\rvert_{h=0} = & \sum_{Z(\overrightarrow{\sigma})=1} V_q(\sigma_1, \dots, \sigma_k) W(\sigma_1; C_1, N_1) \bullet \dots \bullet W(\sigma_k; C_k, N_k)\\
& -k \sum_{Z(\overrightarrow{\sigma})=0} V_q(\sigma_1, \dots, \sigma_k) W(\sigma_1; C_1, N_1) \bullet \dots \bullet W(\sigma_k; C_k, N_k).
\end{aligned}
\end{equation}

For the numerator of \eqref{eq:3.30}, the small- \(h\) behavior is proportional to \eqref{eq:3.39},
because \(L_{\text{umi}}(T;0)=0\), as the nonzero values of \(V_{\text{umi}}(\overrightarrow{\sigma})\)
all have \(Z(\overrightarrow{\sigma}) \geq 1\). 
The denominator of \eqref{eq:3.30}, \(L_{\text{sym}}(T; 0)\) is nonzero, because \(V_{\text{sym}}(\overrightarrow{\sigma})=1\) for six triplets of nonzero arguments
(the cases \(n_{half}=0\) in eq. \eqref{eq:3.33}).
Thus, for small \(h\), the likelihood ratio \eqref{eq:3.30} is proportional to \(h\).
This proportionality indicates to what extent adding a small amount of mass to the prior at
\(R=\frac{1}{2}\) leads to triplets of choice probabilities that are compatible with the ultrametric property.

Since the triplets in each triad form non-overlapping sets (as was the case for the analysis of symmetry),
we can combine the likelihood ratios for each triplet to form a likelihood ratio for the entire dataset:

\begin{equation}\label{eq:3.40}
LR_{\text{umi}}(\overrightarrow{C}, \overrightarrow{N}; h) = \prod_{T \in \mathfrak{trip}} LR_{\text{umi}} (T; h).
\end{equation}
The analysis of the limiting behavior of \(LR_{\text{umi}} (T; h)\) then motivates an index of the extent to which a set of observations is compatible with an ultrametric distance:

\begin{equation}\label{eq:3.41}
I_{\text{umi}}(\overrightarrow{C}, \overrightarrow{N}) = \frac{1}{\#\mathfrak{trip}} \lim_{h\rightarrow0} \left( \log LR_{\text{umi}}(\overrightarrow{C}, \overrightarrow{N}; h)\right) - \log h.
\end{equation}

This index is an average measure of pointwise compatibility of choice probabilities with an
ultrametric distance across all triads. Moreover, in view of the remark following Proposition~\ref{pr:2} and noting that the triads in each triplet are independent, it can also be considered as a measure of setwise compatibility across the entire stimulus space \(S\).
In applying \eqref{eq:3.41} to data, the limiting behavior can be determined by setting
\(h\) to a small nonzero value, e.g., 0.01 or 0.001, as we will show below.

As is the case for symmetry \(I_{\text{sym}}\) (eq. \eqref{eq:3.24}), a useful benchmark is the
\textit{a priori} value, \(I_{\text{umi}}(\overrightarrow{0}, \overrightarrow{0})\). 
To calculate this, it suffices to consider a single triad \(T_0\) for which there are no observations:

\begin{equation}\label{eq:3.42}
I_{\text{umi}}(\overrightarrow{0}, \overrightarrow{0}) = \lim_{h\rightarrow0} (\log LR_{\text{umi}}(T_0;h)) -  \log h .
\end{equation}
The numerator of \(LR_{\text{umi}}(T_0;h)\) (eq. \eqref{eq:3.30}), and its behavior for small \(h\) is given by \(\left.\frac{d}{dh} L_{\text{umi}}(T_0; h) \right\rvert_{h=0}\). 
This can be computed from \eqref{eq:3.39}, noting that (from eq. \eqref{eq:3.36})

\begin{equation}\label{eq:3.43}
    W(\sigma; 0, 0) = \left\{ 
    \begin{aligned}
        \frac{1}{2}, \quad & \sigma = \pm1\\
        1,  \quad & \sigma = 0\\
    \end{aligned}
    \right.
\end{equation}
and that there are three nonzero contributors to \(V_{\text{umi}}\) with \(Z(\overrightarrow{\sigma}) = 1\) (eq. \eqref{eq:3.32}). Thus, \eqref{eq:3.39} yields

\begin{equation}\label{eq:3.44}
L_{\text{umi}}(T_0;h) = \frac{3}{2^2}h + O(h^2).
\end{equation}

The denominator of \(LR_{\text{umi}}(T_0;h)\) is \(L_{\text{sym}}(T_0;h)\) (eq. \eqref{eq:3.30}), 
which, as mentioned above, has a nonzero value at \(h=0\). 
Using \eqref{eq:3.43} and noting that there are six nonzero contributors to \(V_{\text{sym}}\) with
\(Z(\overrightarrow{\sigma}) = 0\) (the cases \(n_{half}=0\) in eq. \eqref{eq:3.33}), \eqref{eq:3.37} yields

\begin{equation}\label{eq:3.45}
L_{\text{sym}}(T_0; 0) = \frac{6}{2^3}.
\end{equation}

Combining \eqref{eq:3.30}, \eqref{eq:3.42}, \eqref{eq:3.44}, and \eqref{eq:3.45} yields the \textit{a priori} value of the index:

\begin{equation}\label{eq:3.46}
I_{\text{umi}}(\overrightarrow{0}, \overrightarrow{0}) = \lim_{h\rightarrow0} 
\left( \log \left( \frac{\frac{3}{4}h+O(h^2)}{\frac{3}{4}}\right) -  \log h \right) = 0.
\end{equation}

In sum, the index \(I_{\text{umi}}(\overrightarrow{C}, \overrightarrow{N})\) (eq. \eqref{eq:3.41}) 
evaluates whether an experimental set of dis-similarity responses is compatible with an ultrametric model, and does so in a way that recognizes the intrinsic limitation that experimental data can never show that a choice probability is exactly \(\frac{1}{2}\). 
If the index is greater than 0, the observed data are more likely to be compatible with an ultrametric model than a set of
unstructured choice probabilities; values less than 0 indicate progressively greater deviations from an ultrametric model.

\subsubsection*{Addtree}

Formulation of the addtree index follows along similar lines, but with a focus on tents rather than triplets -- corresponding to the necessary and sufficient conditions for pointwise compatibility 
(respectively Propositions~\ref{pr:3} and \ref{pr:4}).
For consistency with the machinery developed above for \(I_{\text{umi}}\), we continue with the more general prior
\eqref{eq:3.27}, though the specialization to \(h=0\) 
(the beta-function prior \eqref{eq:3.4}), is of primary interest.
For each tent \(\mathcal{T}\), we consider the likelihood ratio

\begin{equation}\label{eq:3.47}
LR_{\text{addtree}}(\mathcal{T}; h) = \frac{L_{\text{addtree}}(\mathcal{T}; h)}{L_{\text{symtent}}(\mathcal{T}; h)}. 
\end{equation}
Here,
\begin{equation}\label{eq:3.48}
L_{\text{addtree}}(\mathcal{T}; h) = \int\limits_{\Omega_{\text{addtree}}} p(\overrightarrow{R}_\mathcal{T}| \overrightarrow{C}_\mathcal{T}, \overrightarrow{N}_\mathcal{T}) d\overrightarrow{R}_\mathcal{T},
\end{equation}
where \(\Omega_{\text{addtree}}\) is the space in which the six choice probabilities in the tent \(\mathcal{T}\) are compatible with 
symmetry but \textit{falsify} the conditions \eqref{eq:2.14} and 
\eqref{eq:2.15} of Proposition~\ref{pr:3}, and
\begin{equation}\label{eq:3.49}
L_{\text{symtent}}(\mathcal{T}; h) = \int\limits_{\Omega_{\text{symtent}}} p(\overrightarrow{R}_\mathcal{T}| \overrightarrow{C}_\mathcal{T}, \overrightarrow{N}_\mathcal{T}) d\overrightarrow{R}_\mathcal{T},
\end{equation}
where \(\Omega_{\text{symtent}}\) is the space in which the six choice probabilities in the tent \(\mathcal{T}\) are merely compatible with symmetry.

These integrals have the same form as \eqref{eq:3.31}, so it suffices to specify the values of 
\(V_q \left( sgn(R_1-\frac{1}{2}), \cdots, sgn(R_6-\frac{1}{2})\right)\).
For definiteness, given a tent \(\mathcal{T} = \{z; a, b, c\}\) with \(z\) at the vertex and \(\{a, b, c\}\) at the base, we specify the six choice probabilities needed to compute \(V\) as follows:
for the tripod component, 
\(R_1 \triangleq R(z; b,c)\), 
\(R_2 \triangleq R(z; c,a)\),
\(R_3 \triangleq R(z; a,b)\);
for the base, 
\(R_4 \triangleq R(a; b,c)\), 
\(R_5 \triangleq R(b; c,a)\),
\(R_6 \triangleq R(c; a,b)\).
These choice probabilities or those of the mirror triads (but no other choice probabilities) enter into determining whether the conditions \eqref{eq:2.14} and \eqref{eq:2.15} are falsified for this 
tent: the choice probabilities \(R_1\), \(R_2\), \(R_4\), and \(R_5\) 
are explicit in \eqref{eq:2.14} and \eqref{eq:2.15} and all of the 
\(R_i\) are used equally as the base elements \(\{a, b, c\}\) are permuted.
Since \(V_{\text{addtree}}\) and \(V_{\text{symtent}}\) has six arguments, each of which can take on any of three values 
\(\{-1, 0, 1\}\), there are \(3^6 = 729\) values to specify.

For \(V_{\text{addtree}}\), these values may be determined as follows. For the choice probabilities to be pointwise compatible with an addtree distance (i.e., for \(V_{\text{addtree}}=1\)),
the conditions \eqref{eq:2.14} and \eqref{eq:2.15} cannot hold for any of the permutations of \(\{a, b, c\}\). 
Since these conditions are symmetric under interchange of \(a\) and \(b\), it suffices to consider the cyclic permutations.
So the region of \([0,1]^6\) in which \(V_{\text{addtree}}=1\)
is the intersection of the region that falsifies the conditions
\eqref{eq:2.14} and \eqref{eq:2.15}, which we denote \(V^{[c]}_{\text{addtree}}\), with the regions that falsify these conditions after cyclic permutation of \((a, b, c)\), which we denote 
\(V^{[a]}_{\text{addtree}}\) and \(V^{[b]}_{\text{addtree}}\).
Additionally, \(V_{\text{addtree}}=0\) for sets of choice probabilities that are incompatible with a symmetric distance. Thus,

\begin{equation}\label{eq:3.50}
V_{\text{addtree}} = V^{[a]}_{\text{addtree}} V^{[b]}_{\text{addtree}} V^{[c]}_{\text{addtree}} V_{\text{symtent}}.
\end{equation}
\(V^{[c]}_{\text{addtree}}=1\) except when all of the inequalities \eqref{eq:2.14} hold, or, as specified by \eqref{eq:2.15}, when
\(D(z, c) > D(z, a)\) or  \(D(a,b) > D(b,c)\) (but not both) is replaced by equality, and
\(D(z, c) > D(z, b)\) or  \(D(a,b) > D(c,a)\) (but not both) is replaced by equality.
Thus, \(V^{[c]}_{\text{addtree}}=0\) as follows:

\begin{equation}\label{eq:3.51}
\begin{aligned}
& V^{[c]}_{\text{addtree}}(+\tau_1, -\tau_2, \sigma_3, -\tau_4, +\tau_5, \sigma_6) = 0 \\
& \text{for} \, \, (\tau_1, \tau_4) \, \, \text{and} \, \, (\tau_2, \tau_5) \in \{(1,1), (1,0), (0,1)\}; \, \, \sigma_3 \, \, \text{and} \, \, \sigma_6 \in \{-1, 0, +1\}.
\end{aligned}
\end{equation}
Here, the paired \(\tau_i\)'s -- not  both of which can be zero -- handle the allowed equalities specified by \eqref{eq:2.15} and
the \(\sigma\)'s handle the lack of a dependence on the third and sixth arguments.
\(V^{[a]}_{\text{addtree}}\) and \(V^{[b]}_{\text{addtree}}\) are
then determined by cyclic permutation:
\begin{equation}\label{eq:3.52}
\begin{aligned}
V^{[a]}_{\text{addtree}}(\sigma_1, \sigma_2,\sigma_3,\sigma_4, \sigma_5, \sigma_6) & = V^{[c]}_{\text{addtree}}(\sigma_2, \sigma_3,\sigma_1,\sigma_5, \sigma_6, \sigma_4)\\
V^{[b]}_{\text{addtree}}(\sigma_1, \sigma_2,\sigma_3,\sigma_4, \sigma_5, \sigma_6) & = V^{[c]}_{\text{addtree}}(\sigma_3, \sigma_1,\sigma_2,\sigma_6, \sigma_4, \sigma_5)
\end{aligned}
\end{equation}

\(V_{\text{symtent}}\) occurs both in the likelihood 
\(L_{\text{addtree}}(\mathcal{T}; h)\) as a factor via \eqref{eq:3.50} 
and alone in the likelihood \(L_{\text{symtent}}(\mathcal{T}; h)\).
The choice probabilities of the three triads in the base depend on dis-similarities between the elements of the triplet \(\{a, b, c\}\), 
so the choice probabilities compatible with symmetry 
correspond to \(V_{\text{sym}}(\sigma_4, \sigma_5, \sigma_6)=1\) 
(eq. \eqref{eq:3.33}). 
The three triads in the tripod are comparisons between
\(D(z,a)\), \(D(z,b)\), and \(D(z,c)\). 
While these are unconstrained by symmetry, they must be consistent with
transitivity. That is, all of the inequalities:

\begin{equation}\label{eq:3.53}
\left.\begin{aligned}
    D(z, a) < D(z, b) \\
    D(z, b) < D(z, c) \\
    D(z, c) < D(z, a)
\end{aligned}\right\}
\end{equation}
cannot hold, nor can it hold if up to two of the inequalities are non-strict, nor if all signs of comparison are inverted.
This precisely matches the constraints on three dis-similarities required for compatibility with symmetry in Proposition~\ref{pr:1}, 
so it is captured by  \(V_{\text{sym}}(\sigma_1, \sigma_2, \sigma_3)\).
Thus,

\begin{equation}\label{eq:3.54}
V_{\text{symtent}}(\sigma_1, \sigma_2,\sigma_3,\sigma_4, \sigma_5, \sigma_6) = V_{\text{sym}}(\sigma_1, \sigma_2,\sigma_3) V_{\text{sym}}(\sigma_4, \sigma_5,\sigma_6)
\end{equation}
In sum, the likelihood ratio \(LR_{\text{addtree}}(\mathcal{T}; h) = \frac{L_{\text{addtree}}(\mathcal{T}; h)}{L_{\text{symtent}}(\mathcal{T}; h)}\) is determined by
\begin{equation}\label{eq:3.55}
\begin{aligned}
L_{\text{addtree}}(\mathcal{T}; h) = \int_{[0,1]^6} & V_{\text{addtree}} 
\left( sgn(R_1 - \frac{1}{2}), \cdots , sgn(R_6 - \frac{1}{2})\right) R_1^{C_1}(1-R_1)^{N_1 - C_1} \bullet \cdots \\
& \bullet R_6^{C_6}(1-R_6)^{N_6 - C_6} \bullet p_{a,h}(R_1) \bullet \cdots \bullet p_{a,h}(R_6) dR_1 \cdots dR_6
\end{aligned}
\end{equation}
and
\begin{equation}\label{eq:3.56}
\begin{aligned}
L_{\text{symtent}}(\mathcal{T}; h) = \int_{[0,1]^6} & V_{\text{symtent}} 
\left( sgn(R_1 - \frac{1}{2}), \cdots , sgn(R_6 - \frac{1}{2})\right) R_1^{C_1}(1-R_1)^{N_1 - C_1} \bullet \cdots \\
& \bullet R_6^{C_6}(1-R_6)^{N_6 - C_6} \bullet p_{a,h}(R_1) \bullet \cdots \bullet p_{a,h}(R_6) dR_1 \cdots dR_6
\end{aligned}
\end{equation}
where \(V_{\text{addtree}}\) and \(V_{\text{symtent}}\) are given by eqs. \eqref{eq:3.50} and \eqref{eq:3.54}.

As with the other indices, we can now average the likelihood ratios across all tents to form an index of compatibility with an addtree distance: 

\begin{equation}\label{eq:3.57}
I_{\text{addtree}}(\overrightarrow{C}, \overrightarrow{N}; h) = \frac{1}{\#(\mathfrak{tent})} \sum_{\mathcal{T} \in \mathfrak{tent}} \log  LR_{\text{addtree}}(\mathcal{T}; h),
\end{equation}
where the sum is over all tents \(\mathcal{T}\). 
By virtue of Propositions~\ref{pr:3} and \ref{pr:4}, for \(h=0\), this is the average log likelihood for pointwise compatibility with an addtree distance at each tent (\(h=0\) is needed since Proposition~\ref{pr:4} only considers strict inequalities).
Interpreted as a global measure of compatibility, the caveat mentioned for \(I_{\text{sym}}\) applies here too, as setwise compatibility for the stimuli in every tent considered individually does not imply setwise compatibility on the entire stimulus space \(S\).
Here, though, there is an additional caveat: averaging the likelihood ratios in \eqref{eq:3.57} is tantamount to assuming that each tent’s contribution to the log likelihood is independent. 
This is only an approximation since tents may have overlapping triads.

Finally, we calculate the benchmark value of \(I_{\text{addtree}}\) based on the prior alone, for \(h=0\). The two instances of 
\(V_{\text{sym}}\) in its denominator (eq. \eqref{eq:3.54}) each contribute a factor of \(\frac{3}{4}\) (for each, six of \(2^3\) combinations of nonzero signs yield values of \(1\), as in the calculation of \eqref{eq:3.25}). 
In the numerator, by direct enumeration, 24 of $2^6$ nonzero sign combinations yield a value of \(V_{\text{addtree}} =1\). So 

\begin{equation}\label{eq:3.58}
I_{\text{addtree}}(\overrightarrow{0}, \overrightarrow{0}; 0) = 
\log \left( \frac{\frac{24}{64}}{\frac{3}{4}\cdot\frac{3}{4}}\right) = \log \frac{2}{3} \approx -0.4055. 
\end{equation}

\section{Application to simulated datasets}
\subsection*{Methods}
To explore the utility of the indices $I_{\text{sym}}$, $I_{\text{umi}}$, and $I_{\text{addtree}}$  we applied them to a range of simulated datasets. The main simulations considered a domain $S$ with 15 stimuli, with distances assigned according to a range of geometries (Figure~\ref{fig:4}, top). In four of the geometries, stimuli were nodes in a four-level binary tree, and the geometries are distinguished by how the distances are calculated: $Tree - UM$, an ultrametric space in which distances are given by height of the first common ancestor, $Tree - Add$, an addtree space in which distances are given by the graph distance, i.e., the number of links in the shortest path, $Tree - AddWt$, an addtree space in which distances are given by a weighted graph distance, i.e., the sum of the lengths of the links in the shortest path, and $Tree - Eucl$, a non-addtree space in which distances are given by the Euclidean distance between the node as embedded in the plane. The fifth geometry, $Line$, is an addtree space in which the stimuli are arranged in a straight line and distances are given by the graph distance (which is equivalent to the Euclidean distance). The final geometry, $Circle$, is a non-addtree space in which the stimuli are arranged in a circle and distances are given by the graph distance. Within each geometry, distances were scaled to have a root-mean-squared value of 1. \par

\begin{figure}
    \centering
    \includegraphics[width=\linewidth]{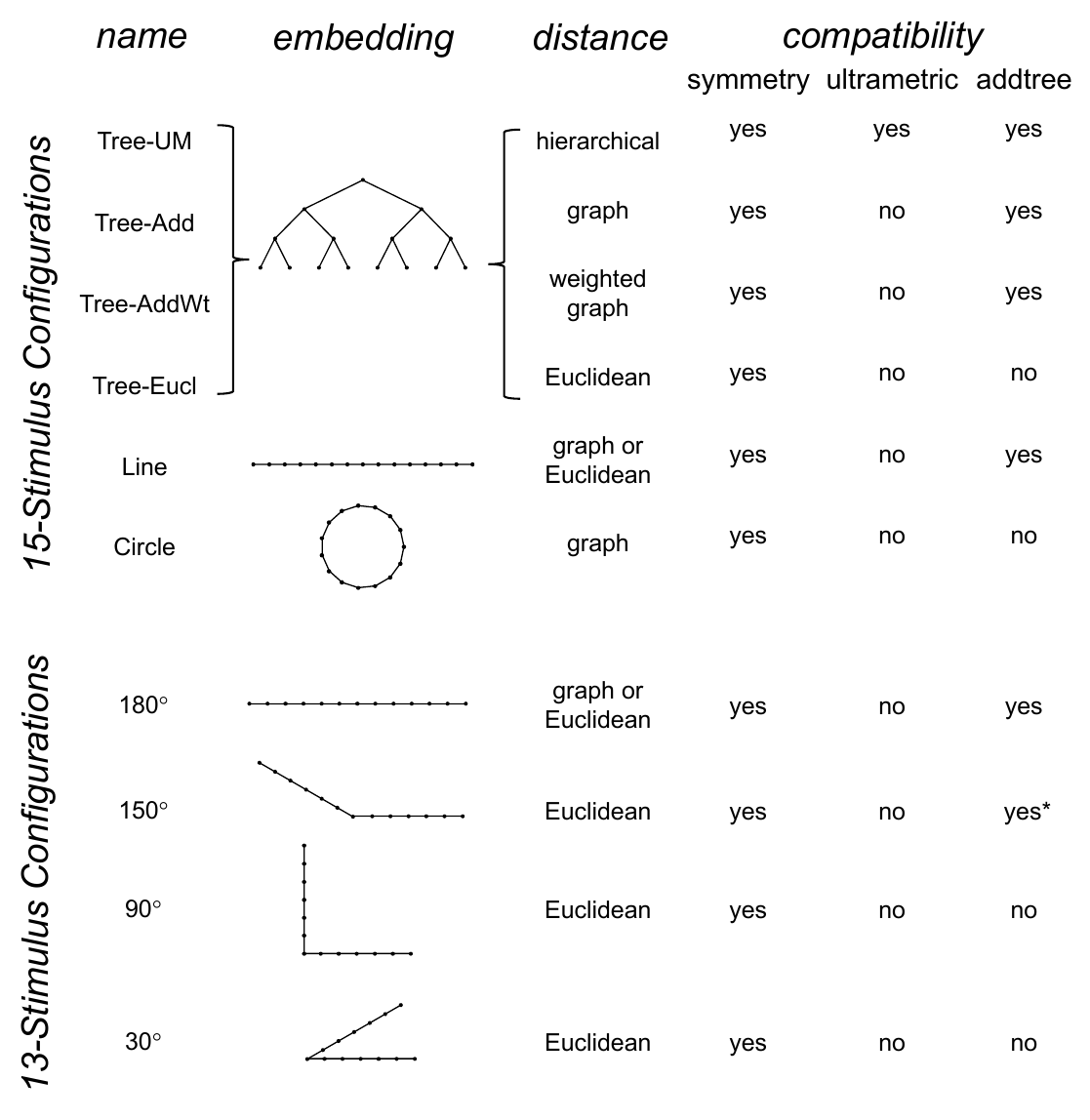}
    \caption{Configurations used for generation of simulated datasets. Simulated dis-similarity judgments are determined by comparing the distances between two stimuli, followed by adding a Gaussian noise to simulate uncertainty. The distance types are: hierarchical, the height of the first common ancestor; graph, number of links in the shortest path; weighted graph: total length of links in the shortest path; and Euclidean, the Euclidean distance in the illustrated embedding. The final three columns indicate setwise compatibility based on the rank order of dis-similarities with a symmetric distance, an ultrametric distance, and an addtree distance. Distances within each configuration are normalized to have a root-mean-squared value of 1. The asterisk for the 150$^{\circ}$ configuration indicates that, while the dis-similarities are setwise compatible with an addtree distance, the distances themselves are not addtree.}
    \label{fig:4}
\end{figure}
For each geometry, we simulated an experiment in which all possible triads were presented an equal number of times $N$, for $N \in \{1, 2, 4, 8, 16, 32\}$. The decision rule compared the distances and then added a random draw from a Gaussian of standard deviation $\sigma$ for a range of values $\sigma \in \{0.0625, 0.125, 0.25, 0.5, 1, 2\}$, to model uncertainty. (In typical experiments, $N$ is in the range of 4-8, and a value of $\sigma = 0.25$ yields a distribution of choice probabilities that is similar to those observed \cite{Waraich2022}). \par

These simulated responses were then used to compute the indices, using the beta-function prior \eqref{eq:3.4} for $I_{\text{sym}} $(eq. \eqref{eq:3.24}) and $I_{\text{addtree}}$ (eq. \eqref{eq:3.57} with $h=0$), and the modified prior \eqref{eq:3.27} for $I_{\text{umi}}$, with $h=0.001$ to approximate the limit in eq. \eqref{eq:3.42}. For most calculations, the parameter $a$ was determined by maximizing the log likelihood (eq. \eqref{eq:3.8}) or its generalization for the modified prior (eq. \eqref{eq:3.28}). In these cases, indices were only computed for $N \geq 2$. This is because for $N=1$, all empiric choice probabilities are 0 or 1 and the likelihood function is independent of $a$, so this parameter cannot be determined. Alternatively, to test sensitivity to the shape parameter $a$, we recalculated the indices for a range of fixed values for $a$ (including the flat prior, $a=0$). To test sensitivity to the point mass parameter of the prior, we recalculated indices for a range of fixed values of $h$, and also with maximum likelihood fitted values for $h$, fitting $a$ by maximum likelihood in all cases. When both $a$ and $h$ are determined by maximum likelihood, indices are only computed for $N \geq 4$.

\subsubsection*{Surrogates}

To provide insight as to whether the indices are dominated by the structure of the data and not
merely by the prior on choice probabilities, we take a heuristic approach based on parallel calculations
with surrogate datasets. The logic is as follows. If the values of the indices are dominated by the prior,
then surrogate datasets, drawn from the prior, would yield values of the indices similar to the values that
were derived from the experimental data. To construct these surrogate datasets, we make use of the fact
that the prior on choice probabilities is an independent one: the \textit{a priori} likelihood of a set of choice
probabilities is the product of the likelihoods assigned to each choice probability. Moreover, the prior
for each choice probability $R$,  either the beta-function prior \eqref{eq:3.4} or its modification \eqref{eq:3.27}  that includes
a discrete mass, is symmetric about an exchange of $R$ for $1-R$. Thus, exchanging any of the observed
choice probabilities $R$ for $1-R$ yields a set of choice probabilities which is \textit{a priori} as likely as the
experimental observations to result from the prior. The range of indices obtained from surrogates with random replacements of $R$ for $1-R$ is thus the range of indices expected when the relationships among
the choice probabilities in the data are removed, but the choice probability distribution is retained.
Conversely, values of the indices outside of this range indicate that the index values are dominated by
relationships in the data that are not inherent in the prior. \par
We implemented this strategy with the ``flip any'' surrogate: the triads were independently
selected with a probability of $\frac{1}{2}$, and the responses to the chosen triads were replaced by their
complement (so that $C$ out of $N$ responses that $D(r, x) < D(r, y)$ were replaced by $N-C$ out of $N$
responses that $D(r, x) < D(r, y)$). Note also that these surrogates match the original data in terms of
whether ``easy'' ($R$ near 0 or 1) and ``hard'' ($R$ near $\frac{1}{2}$ ) triads tend to co-occur in triplets, so this
approach tests not only the effect of the distribution of choice probabilities, but also, the effect of the co-occurrence of easy and hard choice probabilities within triplets.\par
Since (as shown below) indices calculated from the data typically deviated substantially from the
``flip any'' surrogate, we then extended this logic to a second kind of surrogate, to test whether the index
values could be accounted for by a more refined alteration of the choice probabilities. In this surrogate,
called ``flip all,'' the choice probabilities have the same marginal distribution as the independent prior,
but they are no longer independent. Instead, for the ``flip all'' surrogate, the choice probabilities within a
triplet retain all of the pairwise correlation present in the data. To achieve this, we select triplets with a
probability of $\frac{1}{2}$, and, for each selected triplet, replace \textit{all} of the responses of the three triads within the
triplet by their complements. These surrogate sets of choice probabilities have the same likelihood as
the data in the independent prior. They also match the pairwise correlations of choice probabilities
within triplets, since the choice probabilities corresponding to two triads within a triplet are either both
inverted, or both unchanged. However, this surrogate destroys any systematic relationship between the choice probabilities in different triplets, as well as any third-order correlations of the choice probabilities
within triplets. \par
The range of indices obtained from the ``flip all'' surrogates is thus the range of indices expected
when the pairwise correlations among the choice probabilities in a triplet are retained, as well as the
choice probability distribution. For $I_{\text{umi}}$ and $I_{\text{addtree}}$, if the indices calculated from the data deviate from the range of values calculated from the surrogates, this indicates that the extent of ultrametricity or
addtree structure in the data cannot be accounted for merely by pairwise correlations within triplets.
(For $I_{\text{sym}}$, the value calculated from the ``flip all'' surrogate necessarily matches the value calculated
from the data, since the likelihood ratio \eqref{eq:3.15}  is unchanged.) \par
Since the indices are sums of values that are independently computed either from triplets or tents,
the exact means and standard deviations of the surrogate indices could be computed efficiently by
exhaustive sampling of each triplet or tent separately, rather than approximating them via a random
sampling procedure.

\subsection*{Results}

\begin{figure}
    \centering
    \includegraphics[width=0.95\linewidth]{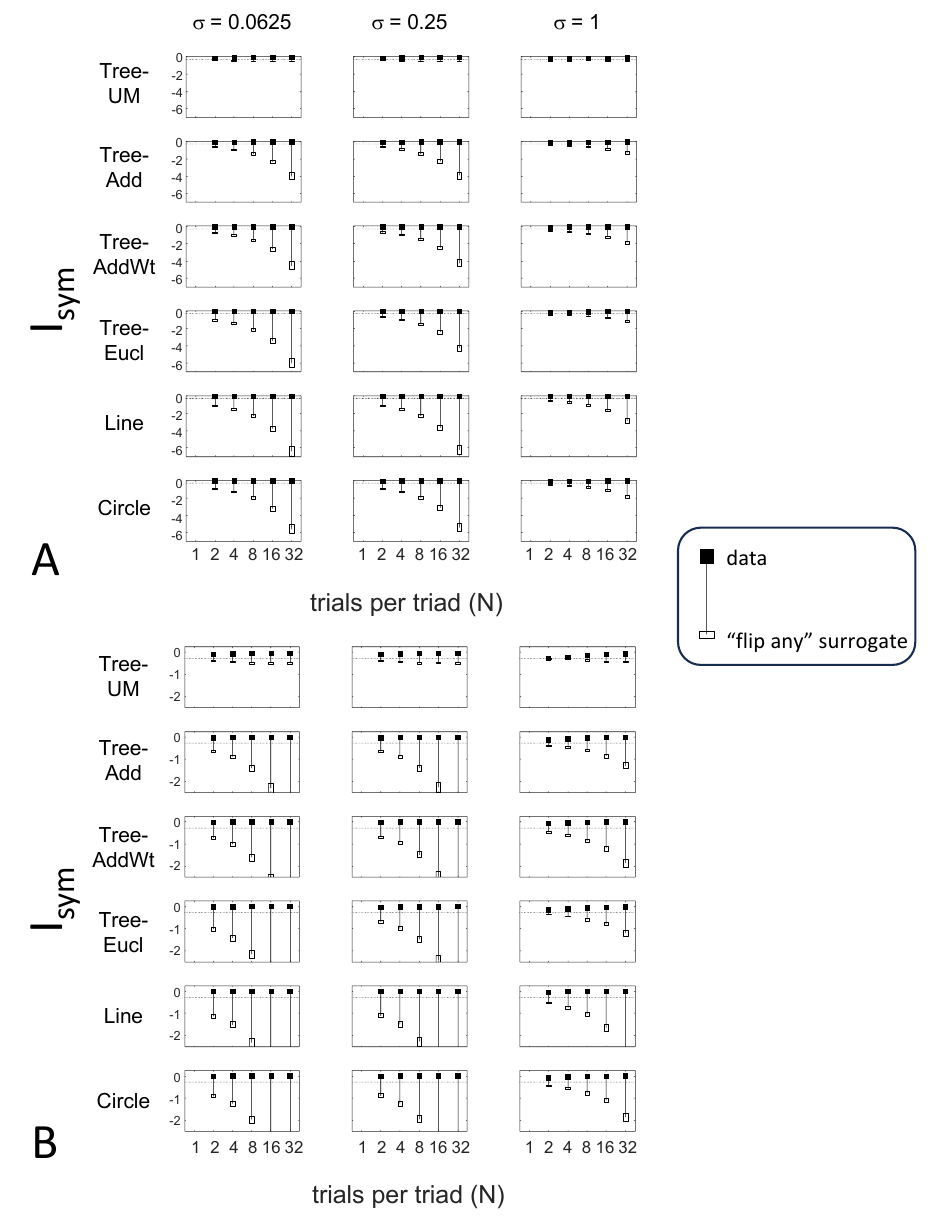}
    \caption{Panel A. Behavior of $I_{\text{sym}}$ for the 15-point configurations, a range of trials per triad ($N$, abscissa), and a range of decision rules ($\sigma$, standard deviation of added noise). Solid symbols indicate values of $I_{\text{sym}}$ computed from simulated decisions. The dashed horizontal line indicates the \textit{a priori} value for $I_{\text{sym}}$. Hollow boxes indicate values of $I_{\text{sym}}$ computed from ``flip any'' surrogates, showing the $\pm1$ standard deviation range. The thin vertical lines are to aid visualization, and do not represent ranges. The beta-function prior \eqref{eq:3.4} is used, with values of the parameter $a$ determined by maximum likelihood. Panel B expands the vertical scale. See Supplementary Figure~\ref{fig:s1} for additional values of $\sigma$ and the values of the prior parameter $a$.}
    \label{fig:5}
\end{figure}

\begin{figure}
    \centering
    \includegraphics[width=0.93\linewidth]{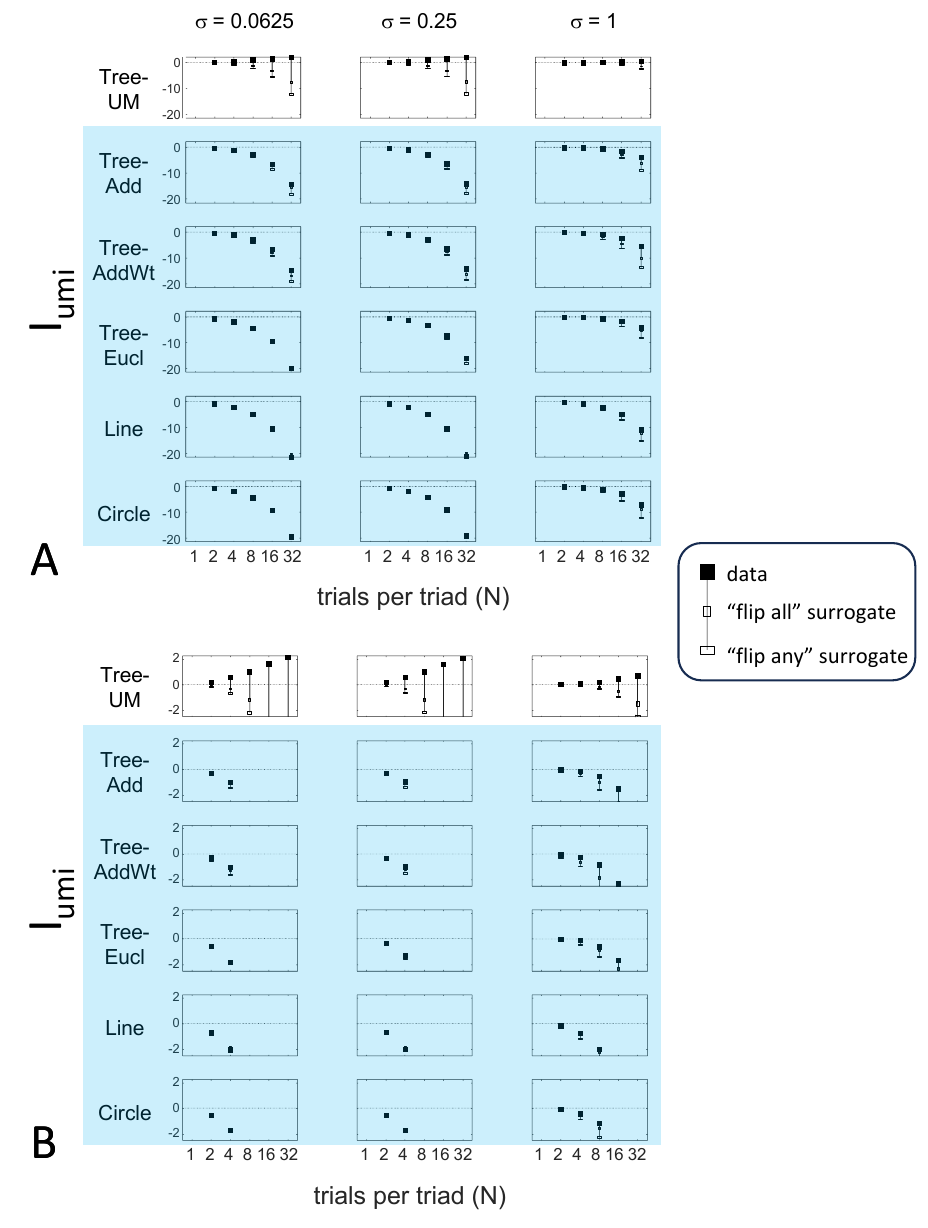}
    \caption{Panel A. Behavior of $I_{\text{umi}}$ for the 15-point configurations, a range of trials per triad ($N$, abscissa), and a range of decision rules ($\sigma$, standard deviation of added noise). Solid symbols indicate values of $I_{\text{umi}}$ computed from simulated decisions. The dashed horizontal line indicates the \textit{a priori} value for $I_{\text{umi}}$. Hollow boxes indicate values of $I_{\text{umi}}$ computed from surrogates, showing the $\pm1$ standard deviation range: wider boxes for ``flip any'' surrogates, narrower boxes for ``flip all'' surrogates. The thin vertical lines are to aid visualization, and do not represent ranges. The modified prior \eqref{eq:3.27} is used with the value of the parameter $a$ determined by maximum likelihood and $h=0.001$. Panel B expands the vertical scale. See Supplementary Figure~\ref{fig:s2} for additional values of $\sigma$ and the values of the prior parameter $a$. Blue overlay indicates the simulated datasets that are incompatible with the ultrametric property.}
    \label{fig:6}
\end{figure}

\subsubsection*{Identifying compatibility with underlying distances}

The results for $I_{\text{sym}}$ (Figure~\ref{fig:5}) demonstrate that the approach efficiently distinguishes choice data that are compatible with symmetric distances from their ``flip any'' surrogates, which are not. 
For all geometries, as the number of observations for each triad ($N$) increases, $I_{\text{sym}}$ ascends from its \textit{a priori} value (approximately $-0.2877$, eq. \eqref{eq:3.25}), indicating compatibility. 
When judgments are reliable ($\sigma \leq 0.5$, first two columns of Figure~\ref{fig:5} and first four columns of Supplementary Figure~\ref{fig:s1}), $I_{\text{sym}}$ reaches its maximal value of 0 with only a few trials ($N=2$ or $4$) per triad. For $\sigma=1$ (third column of Figure~\ref{fig:5}), the maximal value is reached for $N=8$ or $16$ trials per triad, and for highly unreliable judgments ($\sigma=2$, sixth column of Supplementary Figure~\ref{fig:s1}), the maximal value is not reached even for $N=32$ judgments. 
The increasing values of $I_{\text{sym}}$ as calculated from the simulated data are paralleled by decreasing values of $I_{\text{sym}}$ as calculated from the ``flip any'' surrogate, since this manipulation destroys the necessary relationships within each triad ($I_{\text{sym}}$ calculated from the ``flip all'' surrogate is identical to the calculation from the original data because the requirements for symmetry in Proposition~\ref{pr:1} are preserved when triads are replaced by their complements). 
These findings are all as expected and reassuring: the datasets are constructed from symmetric distances, and the ability to determine that the responses are compatible with symmetric distances increases with the number of observations $N$ and decreases with judgment uncertainty $\sigma$. \par

The results for $I_{\text{umi}}$ (Figure~\ref{fig:6}) demonstrate that the approach is able to distinguish a dissimilarity structure that is compatible with ultrametric structure ($Tree - UM$) from those that are not (all other geometries). Specifically, $I_{\text{umi}}$ rises progressively from its \textit{a priori} value of 0 for this geometry (first row of Figure~\ref{fig:6}), but falls for all of the others. Additionally, for the $Tree - UM$ geometry, analyses of both kinds of surrogates show a fall from the \textit{a priori} value. The distinction between the $Tree - UM$ geometry and the surrogate datasets, or between the $Tree - UM$ geometry and the other geometries, is evident for $N=4$ trials per triad when judgments are reliable ($\sigma \leq 0.5$, first two columns of Figure~\ref{fig:6} and first four columns of Supplementary Figure~\ref{fig:s2}), but requires $8$ or more trials per triad when they are unreliable ($\sigma \geq 1$, third column of Figure~\ref{fig:6} and last two columns of Supplementary Figure~\ref{fig:s2}). \par
\begin{figure}
    \centering
    \includegraphics[width=0.96\linewidth]{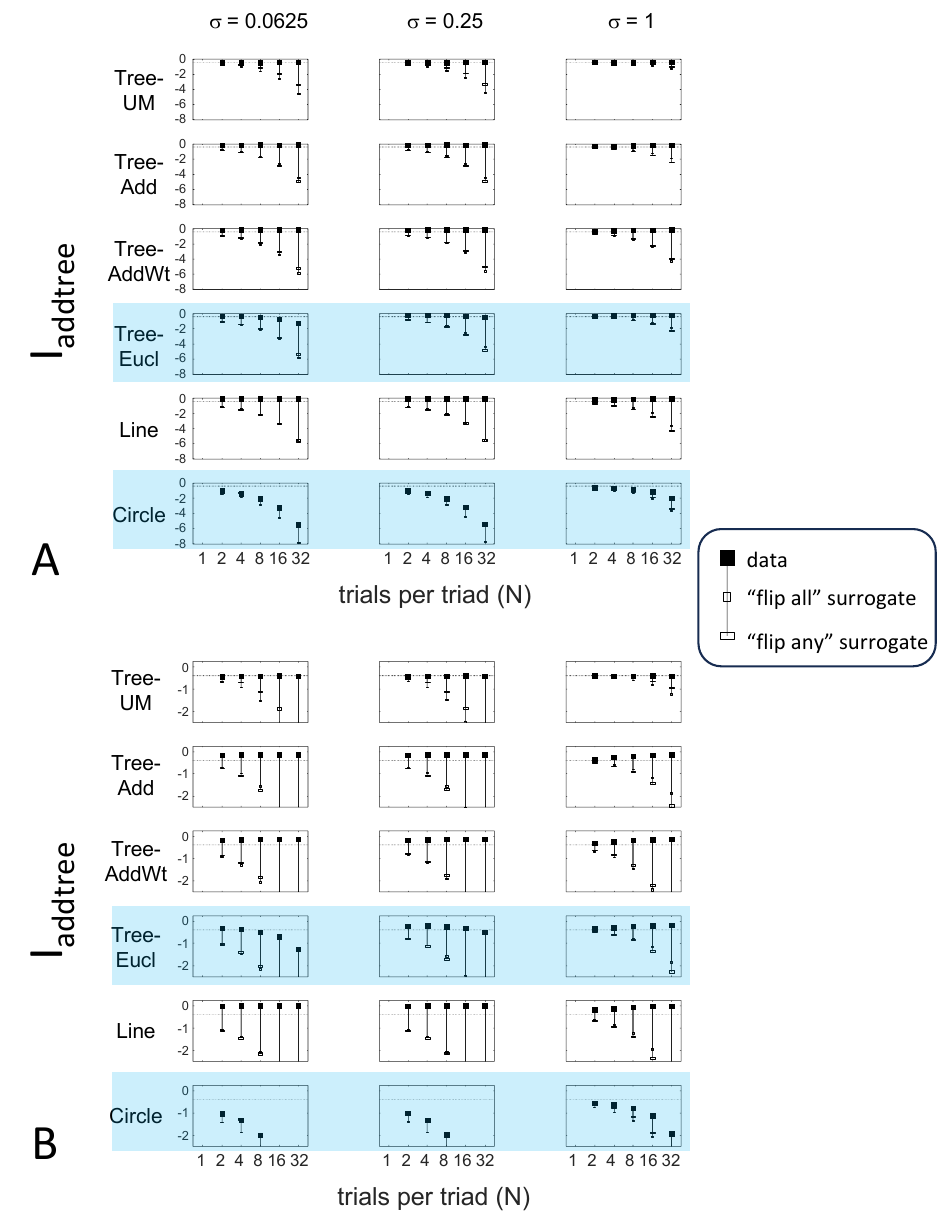}
    \caption{Panel A. Behavior of $I_{\text{addtree}}$ for the 15-point configurations, a range of trials per triad ($N$, abscissa), and a range of decision rules ($\sigma$, standard deviation of added noise). Solid symbols indicate values of $I_{\text{addtree}}$ computed from simulated decisions. The dashed horizontal line indicates the \textit{a priori} value for $I_{\text{addtree}}$. Hollow boxes indicate values of $I_{\text{addtree}}$ computed from surrogates, showing the $\pm1$ standard deviation range: wider boxes for ``flip any'' surrogates, narrower boxes for ``flip all'' surrogates. The thin vertical lines are to aid visualization, and do not represent ranges. The beta-function prior \eqref{eq:3.4} is used, with the value of the parameter $a$ determined by maximum likelihood. Panel B expands the vertical scale. See Supplementary Figure~\ref{fig:s3} for additional values of $\sigma$ and the values of the prior parameter $a$. Blue overlay indicates the simulated datasets that are incompatible with the addtree property.}
    \label{fig:7}
\end{figure}

The results for $I_{\text{addtree}}$ (Figure~\ref{fig:7}) demonstrate that the approach is able to distinguish dissimilarity data that are compatible with an addtree structure ($Tree - UM$, $Tree - Add$, $Tree - AddWt$, $Line$), from those that are not ($Tree- Eucl$, $Circle$). 
For the graph and weighted-graph geometries ($Tree - Add$, $Tree - AddWt$) as well as the $Line$, $I_{\text{addtree}}$ ascends from its \textit{a priori} (approximately $-0.4055$, eq. \eqref{eq:3.58}) to a saturating value of 0, with similar dependence on the number ($N$) and reliability ($\sigma$) of judgments seen above. For the surrogates and for $Circle$, $I_{\text{addtree}}$ has the opposite behavior, descending from its \textit{a priori} value as the amount of data increases. The behavior of $I_{\text{addtree}}$ for the other two non-addtree geometries also contrasts with addtree behavior, but also suggests some caveats. For the Euclidean distance on the tree ($Tree - Eucl$), values of $I_{\text{addtree}}$ below the \textit{a priori} value -- indicating incompatibility with an addtree distance -- are only prominent when judgments are highly reliable ($\sigma \leq 0.25$) and with a larger number of judgments $N \geq 8$. This is not surprising, since the Euclidean and weighted graph distances are often identical or similar. \par
\begin{figure}
    \centering
    \includegraphics[width=\linewidth]{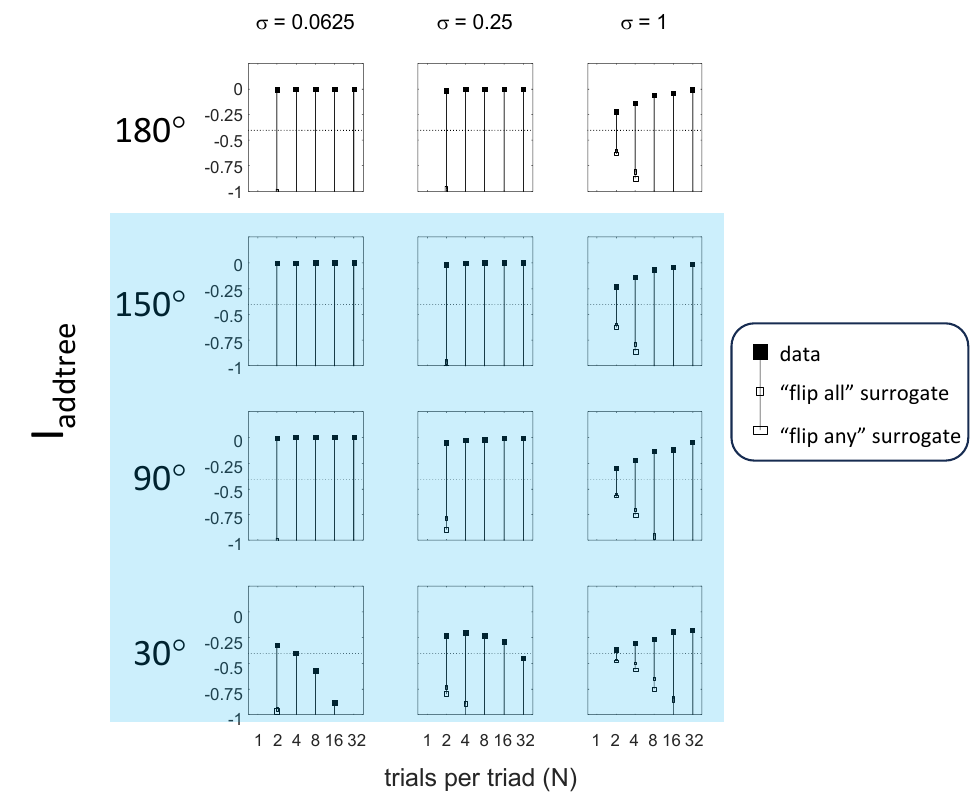}
    \caption{Behavior of $I_{\text{addtree}}$ for the 13-point configurations, a range of trials per triad (abscissa), and a range of decision rules  ($\sigma$, standard deviation of added noise). Graphical conventions as in Figure~\ref{fig:7}. See Supplementary Figure~\ref{fig:s4} for additional values of $\sigma$ and the values of the prior parameter $a$. Blue overlay indicates the simulated datasets that are incompatible with the addtree property.}
    \label{fig:8}
\end{figure}

This point is explored in greater detail in Figure~\ref{fig:8} and shows a limitation of $I_{\text{addtree}}$. We consider configurations of a stimulus set consisting of 13 points, arranged along two rays with a common origin, as shown in Figure~\ref{fig:4}, bottom. (These specifics are motivated by the textures dataset below). The angle between the rays is set at 180$^{\circ}$, 150$^{\circ}$, 90$^{\circ}$, and 30$^{\circ}$, producing a series of simulations in which the
distinction between the graph distance and the Euclidean distance is parametrically varied. For an angle of 180$^{\circ}$, which is equivalent to a straight line of 13 points, the graph distance and the Euclidean distance
are identical. The graph distance is independent of the angle, but the Euclidean distance departs
progressively from the graph distance as the angle decreases. For an angle of 150$^{\circ}$, the underlying choice probabilities for the Euclidean distance differs from that of the graph distance, but the rank orders of the distances are identical. All triadic choice probabilities based on a Euclidean distance are thus pointwise compatible with an addtree metric. In fact, they are setwise compatible: a monotonic transformation that rounds up each Euclidean distance to the nearest integer  multiple of the nearest-neighbor distance yields the graph distance, which is addtree. In this scenario, the values of $I_{\text{addtree}}$ are nearly identical to those for the 180$^{\circ}$ configuration (they are not precisely identical, because the maximum likelihood beta-function prior for the two datasets are slightly different since the choice probabilities are not identical). 
For an angle of 90$^{\circ}$, some of rank-orders of distances are different in the two geometries, and at this point, $I_{\text{addtree}}$ appears to show a slight difference in behavior for the Euclidean model at $\sigma = 0.25$. However, this is more a manifestation of uncertainty in determining the choice probabilities, rather than evidence that the choice probabilities are incompatible with an addtree geometry -- since $I_{\text{addtree}}$ never falls below the \textit{a priori} value, and rises with increasing amounts of data ($N$), or decreasing $\sigma$. In contrast, for an angle of 30$^{\circ}$, where there are major differences between rank orders of graph and Euclidean distances, $I_{\text{addtree}}$ clearly identifies the incompatibility with an addtree geometry, as its value remains below 0.2 and, with sufficient data of sufficient quality (e.g., $N = 8$ with $\sigma = 0.0625$ in Figure~\ref{fig:8}, or $N=8$ with $\sigma = 0.125$ in Supplementary Figure~\ref{fig:s4}), it falls below the \textit{a priori} value. In sum, these examples underscore the fact that $I_{\text{addtree}}$ is only sensitive to rank orders of dis-similarities, not to metrical structure, and show that modest deviations of rank orderings from compatibility with addtree distances may not be detectable.\par

The second notable point in Figure~\ref{fig:7} is that for the $Tree - UM$, values of $I_{\text{addtree}}$ remain close to the \textit{a priori} value across the entire range of $N$ and $\sigma$, rather than rise, even though $Tree- UM$ has an ultrametric geometry. This might appear surprising, since the ultrametric geometry is a special case of an addtree geometry, and thus $Tree - UM$ is compatible with an addtree distance. The reason for this behavior is a severe mismatch between the best-fitting prior and the choice probability distribution.
That is, the simulated data had some choice probabilities are \textit{exactly} $\frac{1}{2}$, while the analysis used the unmodified beta-function prior (eq. \eqref{eq:3.4}), which does not have a point mass at $\frac{1}{2}$. 
Consequently, the prior often forced the inferred choice probabilities to one side or the other of $\frac{1}{2}$ with little evidence, and thus the pattern of \textit{a posteriori} probabilities is spuriously close to the pattern of \textit{a priori} probabilities. 
This scenario is unlikely to arise with real data (where underlying choice probabilities are not likely to be exactly $\frac{1}{2}$). 
More importantly, this potential confound can be simply addressed by using the modified prior with a point mass at $\frac{1}{2}$ (eq. \eqref{eq:3.27}), as we show below. \par
\begin{figure}
    \centering
    \includegraphics[width=\linewidth]{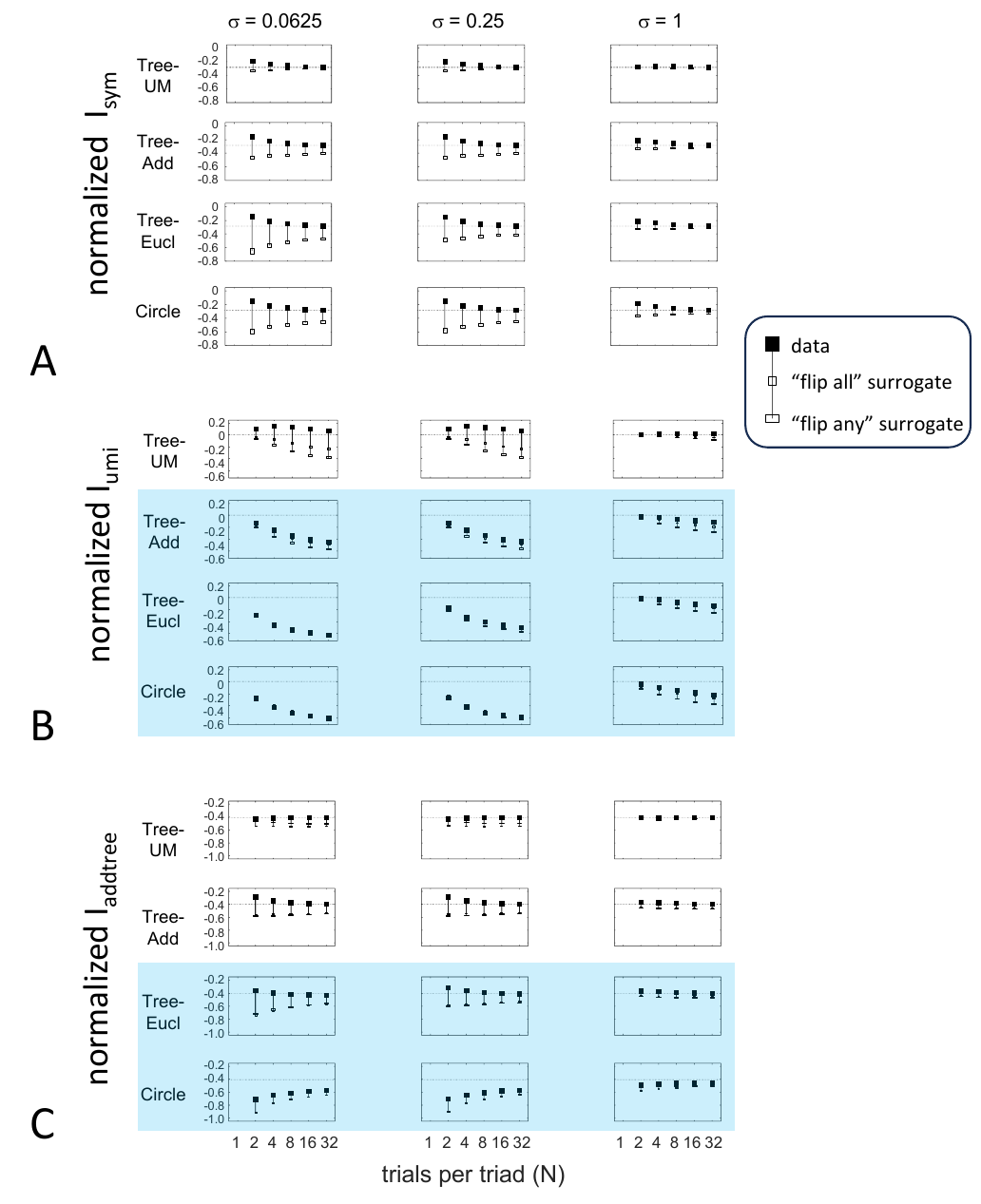}
    \caption{Normalized values of $I_{\text{sym}}$ (panel A), $I_{\text{umi}}$ (panel B), and $I_{\text{addtree}}$ (panel C), for the 15-point configurations, a range of trials per triad ($N$, abscissa), and a range of decision rules ($\sigma$, standard deviation of added noise). Data from Figure~\ref{fig:5}, Figure~\ref{fig:6}, and Figure~\ref{fig:7}, replotted after normalizing the deviation from the \textit{a priori} value by the number of trials per triad: $I_0 + \frac{I-I_0}{N}$, where $I$ is the index from Figures~\ref{fig:5}-\ref{fig:7}, $I_0$ is the \textit{a priori} value, and $N$ is the number of trials per triad. Blue overlay indicates the simulated datasets that are incompatible with the ultrametric property (B) or the addtree property (C). Other graphical conventions as in Figure~\ref{fig:5}, Figure~\ref{fig:6}, and Figure~\ref{fig:7}. See Supplementary Figure~\ref{fig:s5} for additional values of $\sigma$ and results for the $Tree -AddWt$ and $Line$ geometries.}
    \label{fig:9}
\end{figure}

\subsubsection*{Considerations for experimental design}

An experimenter is often faced with the question of the best way to deploy a limited total number of observations: whether to sample a large number of triads with few repeats, or alternatively to concentrate on a small number of triads with proportionately more repeats. To provide some guidance on this point, Figure~\ref{fig:9} replots the indices from Figure~\ref{fig:5} to Figure~\ref{fig:7} after dividing their departure from the \textit{a priori} value by the number of trials in which each triad is sampled, $N$.\par
The values in Figure~\ref{fig:9} can be interpreted as the expected contribution to the log likelihood ratio on a per-observation basis. For the $I_{\text{sym}}$ and $I_{\text{addtree}}$ examples shown, the greatest departures occur for $N=2$. That is, once a set of triads have been sampled twice, it is more informative to observe additional triads, than to observe the same triads more extensively. This optimal number of repeated observations, however, may well be larger if the incompatibility with an addtree geometry is more subtle, as in Figure~\ref{fig:8}. It may therefore be advisable to carry out targeted simulations if there is a particular model to be excluded.\par
For $I_{\text{umi}}$ (Figure $\ref{fig:9}$B), the normalized upwards departure from \textit{a priori} values peaks for $N = 4$ or
$N = 8$ for the compatible geometry $Tree - UM$, and the normalized downwards departure for the incompatible geometries does not appear to saturate. Thus, for a focus on $I_{\text{umi}}$, a more intensive sampling of a limited number of triads is expected to be more informative. \par

\subsubsection*{Dependence on the prior}

Here we explore the influence of the prior for the distribution of choice probabilities. The prior
plays a key role in this approach in that the likelihoods of the underlying choice probabilities are
estimated by combining the observed responses with a prior in a Bayesian fashion. As we will show,
while the prior is necessary to carry out this estimation, it typically has only a minor influence on the
conclusions. \par
\begin{figure}
    \centering
    \includegraphics[width=\linewidth]{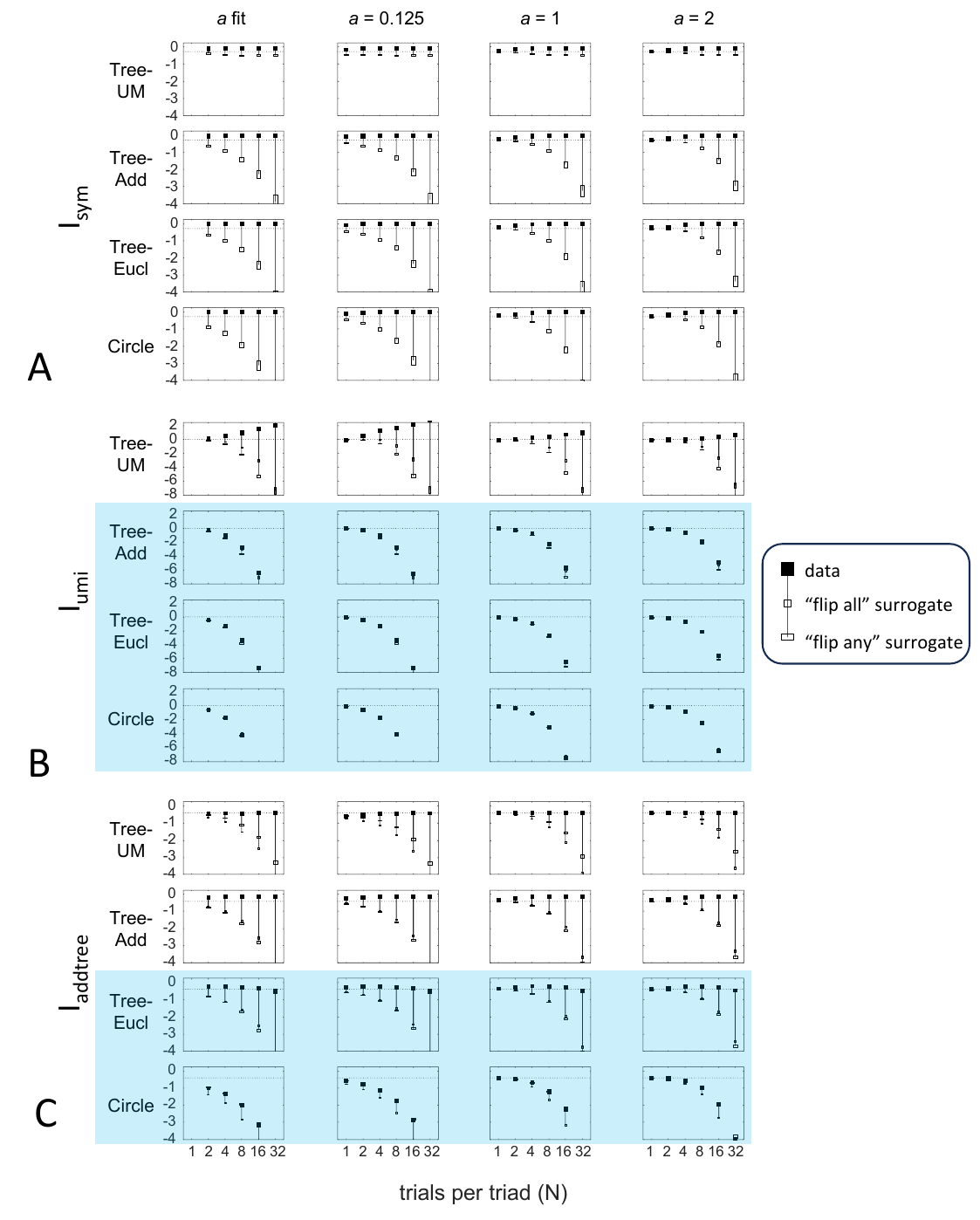}
    \caption{Dependence on the prior’s shape parameter, $a$, for $I_{\text{sym}}$ (panel A), $I_{\text{umi}}$ (panel B), and $I_{\text{addtree}}$ (panel C), for the 15-point configurations and a range of trials per triad ($N$, abscissa). The decision rule $\sigma = 0.25$ is used. Column 1: $a$ is determined by maximum likelihood; other columns: $a$ is assigned the value indicated over each column. The beta-function prior \eqref{eq:3.4} is used for $I_{\text{sym}}$ and $I_{\text{addtree}}$; the modified prior \eqref{eq:3.27} is used for $I_{\text{umi}}$ with $h=0.001$. Blue overlay indicates the simulated datasets that are incompatible with the ultrametric property (B) or the addtree property (C). Other graphical conventions as in Figure~\ref{fig:5}, Figure~\ref{fig:6}, and Figure~\ref{fig:7}. See Supplementary Figure~\ref{fig:s6} for additional fixed values of $a$ and values of the prior parameter $a$ determined by maximum likelihood.}
    \label{fig:10}
\end{figure}
\begin{figure}
    \centering
    \includegraphics[width=\linewidth]{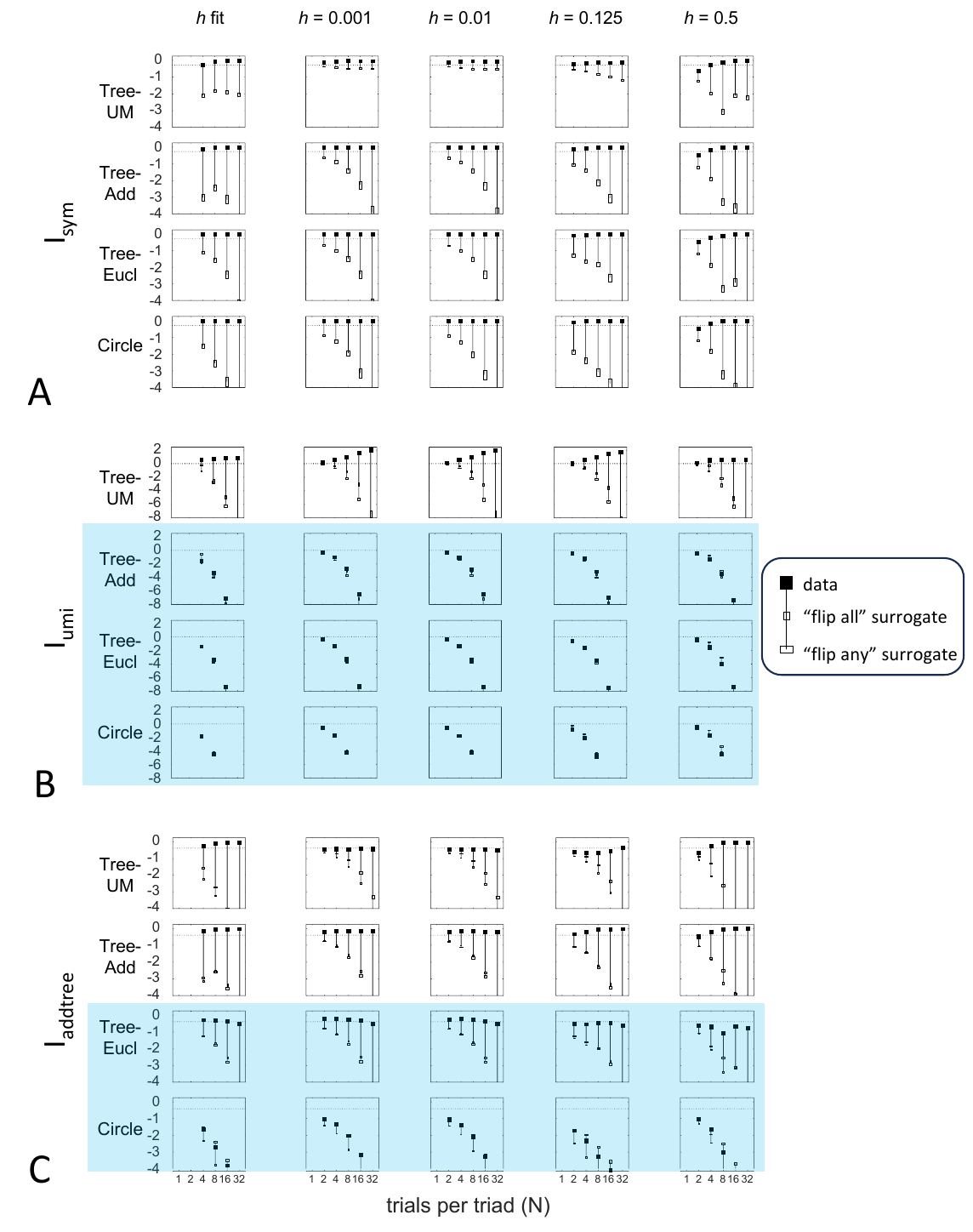}
    \caption{Dependence on the prior’s point mass parameter, $h$, for $I_{\text{sym}}$ (panel A), $I_{\text{umi}}$ (panel B), and $I_{\text{addtree}}$ (panel C), for the 15-point configurations and a range of trials per triad ($N$, abscissa). The decision rule $\sigma = 0.25$ is used. Column 1: $h$ is determined by maximum likelihood; other columns: $h$ is assigned the value indicated over each column. The parameter $a$ is determined by maximum likelihood in all cases. The modified prior \eqref{eq:3.27} is used for all indices. Blue overlay indicates the simulated datasets that are incompatible with the ultrametric property (B) or the addtree property (C). Other graphical conventions as in Figure~\ref{fig:5}, Figure~\ref{fig:6}, and Figure~\ref{fig:7}. See Supplementary Figure~\ref{fig:s7} for additional fixed values of $h$ and the values of the prior parameters $a$ and $h$ determined by maximum likelihood.}
    \label{fig:11}
\end{figure}
Up to this point, we used the prior specified in the Implementation section: the beta-function prior $\eqref{eq:3.4}$ for $I_{\text{sym}}$ and $I_{\text{addtree}}$, in which the shape parameter $a$ was determined by maximum-likelihood, and the modified prior \eqref{eq:3.27} for $I_{\text{umi}}$, with a point mass of weight $h=0.001$ to approximate the limit of eq. $\eqref{eq:3.41}$. Here we consider instead specific choices of the shape parameter $a$ (Figure~\ref{fig:10}) and alternative choices for the point mass $h$, including fitting $h$ by maximum likelihood (Figure~\ref{fig:11}). Note that these choices provided distributions that cover a wide range of shapes: for $h=0$, these choices
specify priors that range from U-shaped ($a< 1$, typical of experimental data, where most judgments are
clearcut) to an inverted U ($a>1$), and that $h>0$ adds a point mass at $\frac{1}{2}$, leading to trimodal distributions for $a<1$. The flat prior corresponds to $a=1$ and $h=0$. We consider four representative geometries: $Tree - UM$ (ultrametric), $Tree - Add$, (addtree), $Tree - Eucl$ (non-addtree but only modestly incompatible), and $Circle$ (non-addtree but highly incompatible), and the decision rule specified by $\sigma = 0.25$. \par
Figure~\ref{fig:10} shows the behavior of the indices when $a$ is set to a fixed value rather than determined by maximum likelihood. Supplementary Figure~\ref{fig:s6} shows their behavior for additional values of $a$ and also gives the values of $a$ determined by the maximum-likelihood fit. As expected, if $a$ is set to a value that is far from the maximum likelihood value, the indices require more data to depart from their \textit{a priori} value. This effect, however, is modest: for a flat prior ($a=1$) and even for a poorly-chosen prior
($a=2$), recovery of the behavior seen for the fitted prior is typically seen with $N =4$ or $N =8$ trials per triad. \par
Figure~\ref{fig:11} shows a parallel analysis of alternative choices for the point mass $h$. Supplementary Figure~\ref{fig:s7} shows their behavior for additional values of $h$ and also gives the values of $a$ and $h$ determined by the maximum-likelihood fit. For $I_{\text{sym}}$ and $I_{\text{addtree}}$, the influence of this choice is again modest. The main finding is that the inability of $I_{\text{addtree}}$ to detect compatibility of the strictly ultrametric dataset $Tree - UM$ with addtree geometry when using a beta-function prior (Figure~\ref{fig:7}, top row) is remedied by allowing $h > 0$ (Figure~\ref{fig:11}C, top row) as this prior allows for capturing the large fraction of underlying choice probabilities that are exactly $\frac{1}{2}$. For $I_{\text{umi}}$, the dependence on the prior is somewhat
complex. For the non-ultrametric geometries, there is a modest dependence on the prior: $I_{\text{umi}}$ descends similarly from its \textit{a priori} value with increasing amounts of data. But for $Tree - UM$, there is a sharper ascent of $I_{\text{umi}}$ when $h$ is chosen to be small ($0.001$ or $0.01$), than when $h$ is close to its fitted value of $\sim 0.4$. This is due to the fact that $I_{\text{umi}}$ is defined by its limiting behavior when $h$ approaches zero. This limiting behavior is accurately estimated with $h=0.001$ or $0.01$, but not larger values.

\section{Application to sample experimental datasets}

Here we demonstrate the utility of the present approach via application to sample datasets from three psychophysical experiments, 
encompassing two methods for acquiring similarity judgments and
spanning low- and high-level visual domains.

\subsection*{Methods}

The first two experiments (``textures'' and ``faces'') make use of the method of Waraich et al. \cite{Waraich2022}: on each trial,
participants rank the eight comparison stimuli \(c_1,\dots, c_8\), in order of similarity to a common reference \(r\). 
These rank-orderings are then interpreted as a set of similarity 
judgments: ranking \(c_i\) as more similar than \(c_j\) to the reference \(r\) is interpreted as a triadic judgment that
\(D(c_j, r) > D(c_i,r)\).
Data are accumulated across all trials in which \(c_i\) and \(c_j\) are presented along with the reference \(r\), leading to an estimate of 
\(R_{obs}(r; c_i, c_j)\). 
Stimulus sets consisted of 24 or 25 items (described in detail with Results below), and 10 sessions of 100 trials each are presented. 
On each trial, stimuli are randomly chosen to be the reference or the comparison stimuli. 
As there were 10 sessions of 100 self-paced trials each and each trial yielded \(\binom{8}{2} = 28\) triadic judgments, each participant’s dataset contained 28,000 triadic judgments.
This corresponds to an average of \(\frac{28000}{25 \cdot \binom{24}{2}} \approx 6.08 \) trials per triad for the 25-item sets and 
\(\frac{28000}{24 \cdot \binom{23}{2}} \approx 6.92 \) trials per triad for the 24-item sets.

For the ``texture'' and ``faces'' datasets (described in detail below), stimuli were generated in MATLAB, and were displayed and sequenced using open-source PsychoPy software on a 22-inch LCD screen (Dell P2210, resolution 1680x1050, color profile D65, mean luminance 52 cd/m\textsuperscript{2}).
The display was viewed binocularly from a distance of 57 cm. 
The visual angle of the stimulus array was 24 degrees;
each stimulus (a texture patch or a face) subtended 4 degrees. 
Tallying of responses and multidimensional scaling as described in \cite{Waraich2022} were carried out via Python scripts.
Computation of the indices and visualization was carried out in MATLAB using code that is publicly available at 
\url{https://github.com/jvlab/simrank}.

The third experiment (``brightness'') uses an odd-one-out paradigm. On each trial, three stimuli are presented, each consisting of a central disk drawn from one of eight luminances, and an annular surround. 
The surround was either of minimal or maximal luminance, and was perceived as black or white, respectively. 
The participant is asked to judge the brightness of the central disk, and to choose which of the three is the outlier. 
We interpret selection of a stimulus \(x_j\) out of a triplet \(\{x_j, x_k, x_l\}\) as a judgment that the pairwise dis-similarities involving this stimulus are larger than the dis-similarity of the two non-outliers, i.e., that \(D(x_k,x_j) > D(x_k, x_l )\) and also 
that \(D(x_l, x_j) > D(x_l, x_k)\). 
Each trial thus contributes to estimates of choice probabilities for two triads, \((x_k; x_j, x_l)\) and \((x_l; x_j, x_k)\), 
and these judgments are tallied across the experiment.
Note though that, in contrast to the ``textures'' and ``faces''
datasets, here the specific triadic comparisons that enter into the tallies depend on the participant's responses.

For the ``brightness'' dataset, stimuli were generated in Python 3.10 and the NumPy library. Stimuli were displayed on a calibrated 24-inch ViewPixx monitor (1920x1080 pixel resolution, mean luminance 70 cd/m\textsuperscript{2}, Vpixx Technologies, Inc.), running custom Python libraries that handle high bit-depth grayscale images (\url{https://github.com/computational-psychology/hrl}). 
Monitor calibration was accomplished using a Minolta LS-100 photometer (Konica Minolta, Tokyo, Japan). 
The display was viewed binocularly from a distance of 76 cm.
The visual angle of the display was 39 degrees; each stimulus subtended 5 degrees, with the central disk subtending 1.67 deg. 
The three stimuli were arranged in a triangular manner, 4 degrees equidistant from the center (Fig.~\ref{fig:15}A). 
There were 16 unique stimuli, consisting of all pairings of 8 values for the luminance of the center disk (14, 33, 55, 78, 104, 131, 163 and 197 cd/m\textsuperscript{2}) and 2 values of luminance for the surrounding annulus (0.77 and 226 cd/m\textsuperscript{2}).
The 16 stimuli generated \(\binom{16}{3} = 560\) possible triplet combinations, which were presented in randomized order and position, constituting one block. 
Each session consisted of two blocks, and each participant ran four
sessions. In total, we collected 4480 trials per participant.
As each trial gives information for two triadic judgments (as mentioned above), there were 8960 triadic judgments per participant. 
The 560 triplets contain 3 triads each, so there were an average of
\(\frac{8960}{3 \cdot 560} = 5.33\) trials per triad.

The texture and faces experiments were performed at Weill Cornell Medical College, in four participants (3F, 1M), ranging in age from 23 to 62. 
Participants MC and SAW (an author) were experienced observers and familiar with the ``texture'' stimulus set from previous studies; participants BL and ZK were naïve observers.
All participated in the textures experiment; 2F (SAW and MC) participated in the faces experiment and neither had prior familiarity with those stimuli.
The brightness experiment was performed at Technische Universität Berlin in three participants (1F, 2M), ranging in age from 31 to 39. 
Participant JP was a naïve observer; participants GA (an author) and JXV were experienced observers. 
All participants had normal or corrected-to-normal vision. 
They provided informed consent following institutional guidelines and the Declaration of Helsinki, according to a protocol that was approved by the relevant institution.

In addition to the calculations described above, we also calculated the indices \(I_{\text{sym}}\), \(I_{\text{umi}}\), and \(I_{\text{addtree}}\) for surrogate datasets, as detailed in the Simulated Datasets section. 
Briefly, the ``flip any'' surrogate was created by randomly selecting individual triads, and flipping the choice probabilities for the selected triads. 
The ``flip all'' surrogate was created by randomly selecting triplets and replacing \(R_{\text{obs}}(r; x, y)\) by \(1-R_{\text{obs}}(r; x, y)\) in the selected triplets.

Finally, we estimated the standard errors for the indices calculated from the original datasets via a jackknife on triplets (for \(I_{\text{sym}}\), \(I_{\text{umi}}\)) or tents (for \(I_{\text{addtree}}\)).
Maximum likelihood parameters \(a\) and \(h\) were not re-calculated for the jackknife subsets, as pilot analyses confirmed that removal of one triplet or tent made very little change in the maximum likelihood value.

% FIG 12
\begin{figure}
    \centering
    \includegraphics[width=0.59\linewidth]{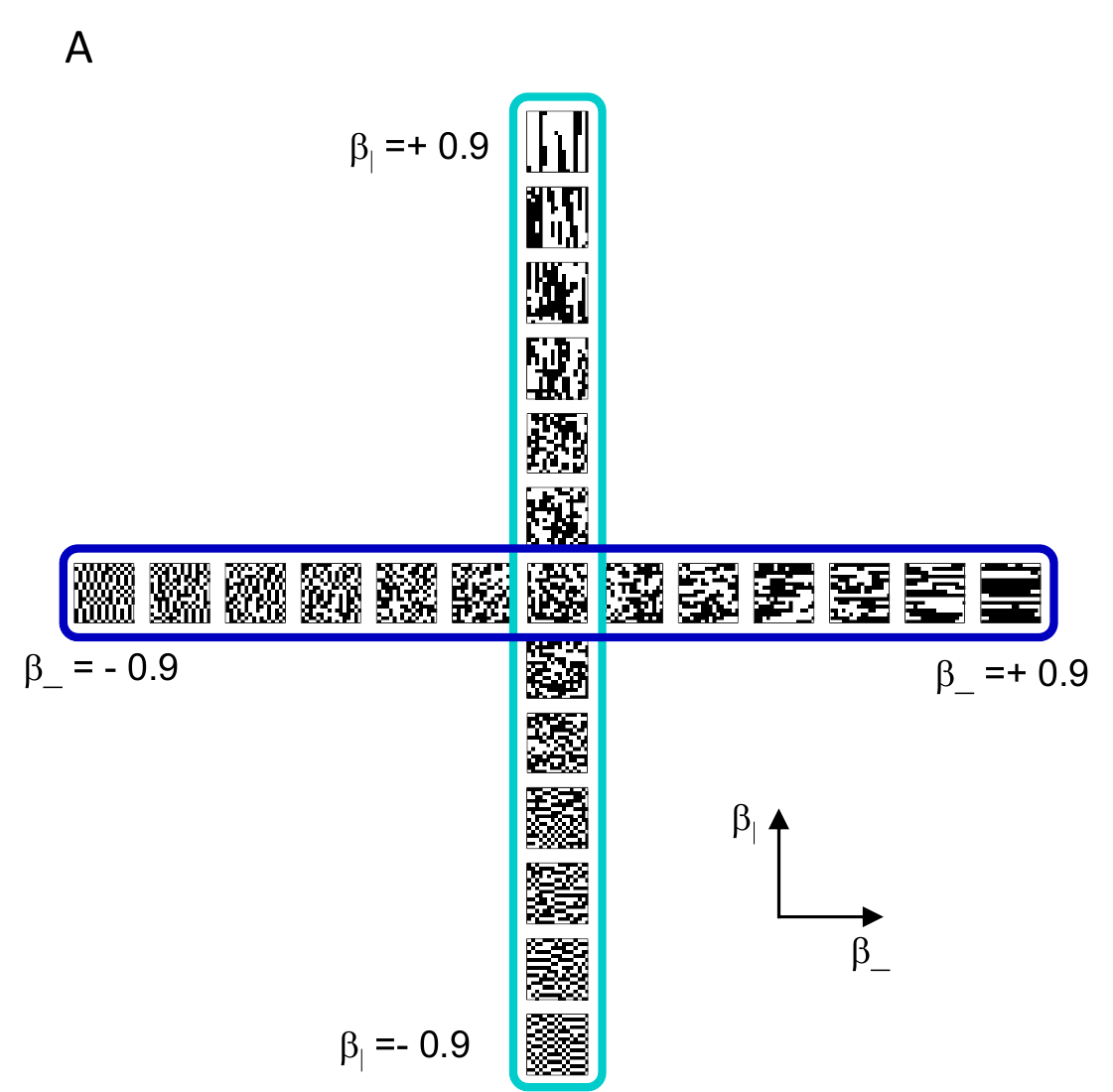}
    \\
    \includegraphics[width=0.59\linewidth]{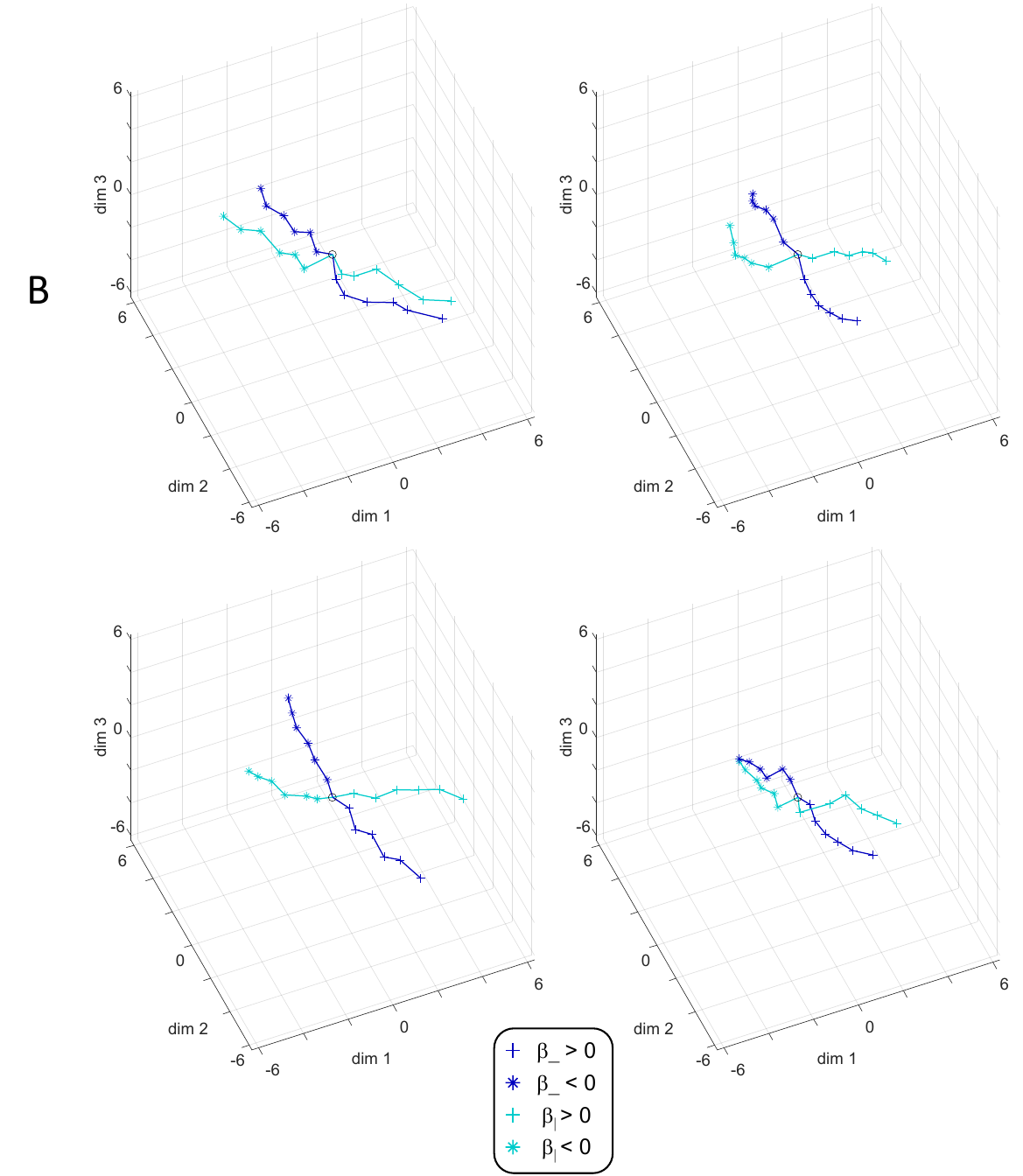}
    \caption{Panel A. Stimuli used in the texture experiment.
    Each stimulus is an array of 16$\times$16 black or white checks. 
    For stimuli enclosed in dark blue, checks are correlated (or anticorrelated) along rows. 
    Correlation strength is parameterized by \(\beta_{-}\), where \(\beta_{-}> 0\) indicates positive correlation and \(\beta_{-} < 0\) indicates negative correlation.
    For stimuli enclosed in light blue, checks are correlated (or anticorrelated) along columns, similarly parameterized by \(\beta_{|}\).
    The full stimulus set consists of 6 equally-spaced values
    positive and negative values of \(\beta_{-}\) and \(\beta_{|}\), and an uncorrelated stimulus (center), where \(\beta_{-}=\beta_{|}=0\).
    Panel B. Multidimensional scaling of similarity judgments for the stimuli in panel A for four participants. 
    The data from each participant have been rotated into a consensus alignment via the Procrustes procedure (without rescaling).
    Lines connect stimuli along each of the rays in Panel A. 
    One unit indicates one just-noticeable difference in an additive noise model \cite{Waraich2022}.
    }
    \label{fig:12}
\end{figure}

\subsection*{Results}
\subsubsection*{Textures}

% FIG 13
\begin{figure}
    \centering
    \includegraphics[width=\linewidth]{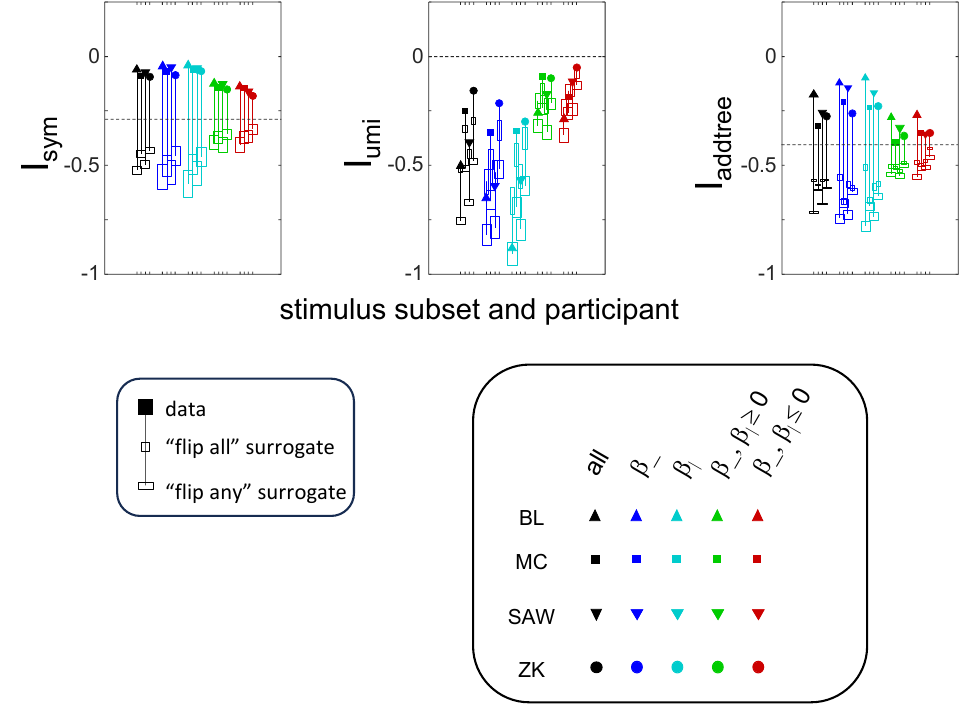}
    \caption{Indices \(I_{\text{sym}}\), \(I_{\text{umi}}\), and \(I_{\text{addtree}}\) for the texture experiment. 
    Stimulus subsets are indicated by symbol color; participants by symbol shape. 
    The vertical extent of the wide boxes indicate \(\pm1\) standard deviation for the ``flip any'' surrogates; the vertical extent of the narrow boxes indicate \(\pm1\) standard deviation for the ``flip all'' surrogates (not plotted for \(I_{\text{umi}}\), as this index is unchanged by the ``flip all'' operation). 
    The thin vertical lines are to aid visualization, and do not represent ranges. Standard errors for the experimental datasets are smaller than the symbol sizes. Null-hypothesis values of the indices are indicated by the horizontal dashed line.
    }
    \label{fig:13}
\end{figure}

The textures experiment made use of the stimulus space described in \cite{Victor2012}, a 10-dimensional space of binary textures with well-characterized discrimination thresholds \cite{Victor2015}. 
We chose a two-parameter component of this domain (Figure~\ref{fig:12}A) that allowed a focus on testing for compatibility for addtree structure.
The two parameters chosen, \(\beta_{-}\) and \(\beta_{|}\), determine the nearest-neighbor correlations in the horizontal or vertical direction: the probability that a pair of adjacent checks have the same luminance (either both black or both white) is 
\((1+\beta)/2\), and the probability that a pair of adjacent checks have the opposite luminance (one black, the other white) is \((1-\beta)/2\). Other than these constraints, the textures are maximum-entropy (see \cite{Victor2012} for details). 
For these experiments, we chose values of \(\beta_{-}\) or \(\beta_{|}\) from -0.9 to 0.9 in steps of 0.15. 
That is, six stimuli had positive values of \(\beta_{-}\)($0.15, 0.30, 0.45, 0.60, 0.75, 0.90$) with \(\beta_{|}= 0\), 
six had the corresponding negative values of \(\beta_{-}\), 
six had positive values of \(\beta_{|}\) with \(\beta_{-}=0\), 
six had negative values of \(\beta_{|}\) with \(\beta_{-}=0\),
and one had \(\beta_{-}=\beta_{|}=0\).
In the experiment, each stimulus example was unique -- that is, a stimulus is specified by a particular \((\beta_{-}, \beta_{|})\) pair, but the texture example used on each trial was a different random sample from that texture.

The rationale for this stimulus set is that we anticipated that certain subsets of stimuli would be more compatible with the addtree model than others. 
The basis for these expectations is shown in Figure~\ref{fig:12}B, which presents non-metric multidimensional scaling of the similarity data.
This analysis, carried out with the procedure detailed in \cite{Waraich2022}, uses a maximum likelihood criterion to place the 25 stimulus points in a space, so that the Euclidean distances between them best account for the choice probabilities (assuming a uniform, additive decision noise). Consistently across participants, the points along each stimulus axis (\(\beta_{-}\) or \(\beta_{|}\)) map to a gradually curving trajectory. 
For this reason, we anticipate that comparison data from the stimuli on one of these trajectories (the 13 points with either \(\beta_{-}\) or \(\beta_{|}\) equal to zero, here called an ``axis'') when analyzed in isolation, will be close to an addtree model.
However, the two trajectories are not perpendicular: rays with same signs of \(\beta\) meet at an acute angle of 45° or less. 
That is, stimuli with strong positive correlations (\(\beta_{-}> 0\)  compared to \(\beta_{|}> 0\)) are seen as relatively similar to each other. 
This is anticipated to make the subset consisting of the 13 points with either \(\beta_{-}\)  or \(\beta_{|}\) positive (a ``vee'') incompatible with the addtree model, as the shortest perceptual path between two points at the end of the positive \(\beta_{-}\) or \(\beta_{|}\) rays is much shorter than a path that traverses each ray back to the origin. 
Similar reasoning indicates that the ``vee'' formed by the two negative rays should also be incompatible with an addtree model. 
Note, though, that this intuition assumes that the Euclidean distances in Figure~\ref{fig:12}B are an accurate account of the perceptual dis-similarities; the analysis via \(I_{\text{addtree}}\) does not make this assumption.

Figure~\ref{fig:13} shows the indices \(I_{\text{sym}}\), \(I_{\text{umi}}\), and \(I_{\text{addtree}}\) computed from the full datasets for each participant, and for the ``axis'' and ``vee'' subsets.
As expected from the above analysis, the addtree index \(I_{\text{addtree}}\) is substantially higher for the ``axis'' subsets than for the ``vee'' subsets -- comparable to the difference between the values of \(I_{\text{addtree}}\) for the simulated 180° and 30° datasets in Figure~\ref{fig:8} -- and the values for the ``axis'' and ``vee'' subsets straddle the value for the full dataset.
Note that ``axis'' and ``vee'' subsets are equated in terms of the number of stimuli, and were collected simultaneously within a single experiment. 
This finding supports the efficacy of \(I_{\text{addtree}}\) in determining compatibility with addtree distances: it is close to zero for data along an axis, which is anticipated to be compatible with an addtree distance, and decreases reproducibly for subsets that form an acute ``vee'', for which the rank orders of similarity expected to be incompatible with an addtree distance. 
Note also that in all cases, it is higher than the \textit{a priori} value, and substantially higher than values computed from surrogate datasets in which choice probabilities are randomly flipped. 
This latter point indicates (not surprisingly) that for the full dataset and the selected subsets, there are portions of the data that are more compatible with an addtree model than datasets with the same choice probabilities, but no relationship between the triads.

For this dataset, values of \(I_{\text{sym}}\) were quite close to zero (usually \(>-0.1\)), indicating that nearly all (\(> e^{-0.1}\)) of the posterior distribution of choice probabilities was compatible with a symmetric dis-similarity.
\(I_{\text{umi}}\), which measures compatibility with the ultrametric model, was typically $-0.25$ or less, substantially below the \textit{a priori} value of zero. But interestingly, the highest values of \(I_{\text{umi}}\) were seen in
the ``vee'' subsets, suggesting a partially hierarchical structure -- e.g., that the two directions of correlation formed categories. 
As was the case for \(I_{\text{addtree}}\), all indices were higher than for the surrogates. 
For \(I_{\text{sym}}\), this is unsurprising, as randomly flipping choice probabilities would be unlikely to lead to a set of symmetric judgments.
For \(I_{\text{umi}}\), this finding indicates that, even though the ultrametric model is excluded, the data has islands of relatively greater compatibility with the ultrametric structure.

The above results were insensitive to the parameters \(a\) and \(h\) of the prior for the distribution of choice probabilities (eqs. \eqref{eq:3.4} and \eqref{eq:3.27}). 
The beta-function parameter \(a\) obtained by maximum likelihood (eq. \eqref{eq:3.8}) ranged from $0.25$ to $1.25$ (with the lowest values for the full texture dataset), but very similar results to Figure~\ref{fig:13} were obtained with \(a=0.5\) for all datasets. 
For \(I_{\text{umi}}\), the limit in \eqref{eq:3.41} was estimated by setting \(h=0.001\) but similar values were obtained for \(h=0.01\).
The findings for \(I_{\text{sym}}\) and \(I_{\text{addtree}}\), here shown for \(h=0\), were not substantially changed when \(h\) was determined by maximum likelihood. 
These values of \(h\) were typically quite small (median, \(7 \cdot10^{-5}\)).

\subsubsection*{Faces}

% FIG 14
\begin{figure}
    \centering
    \includegraphics[width=\linewidth]{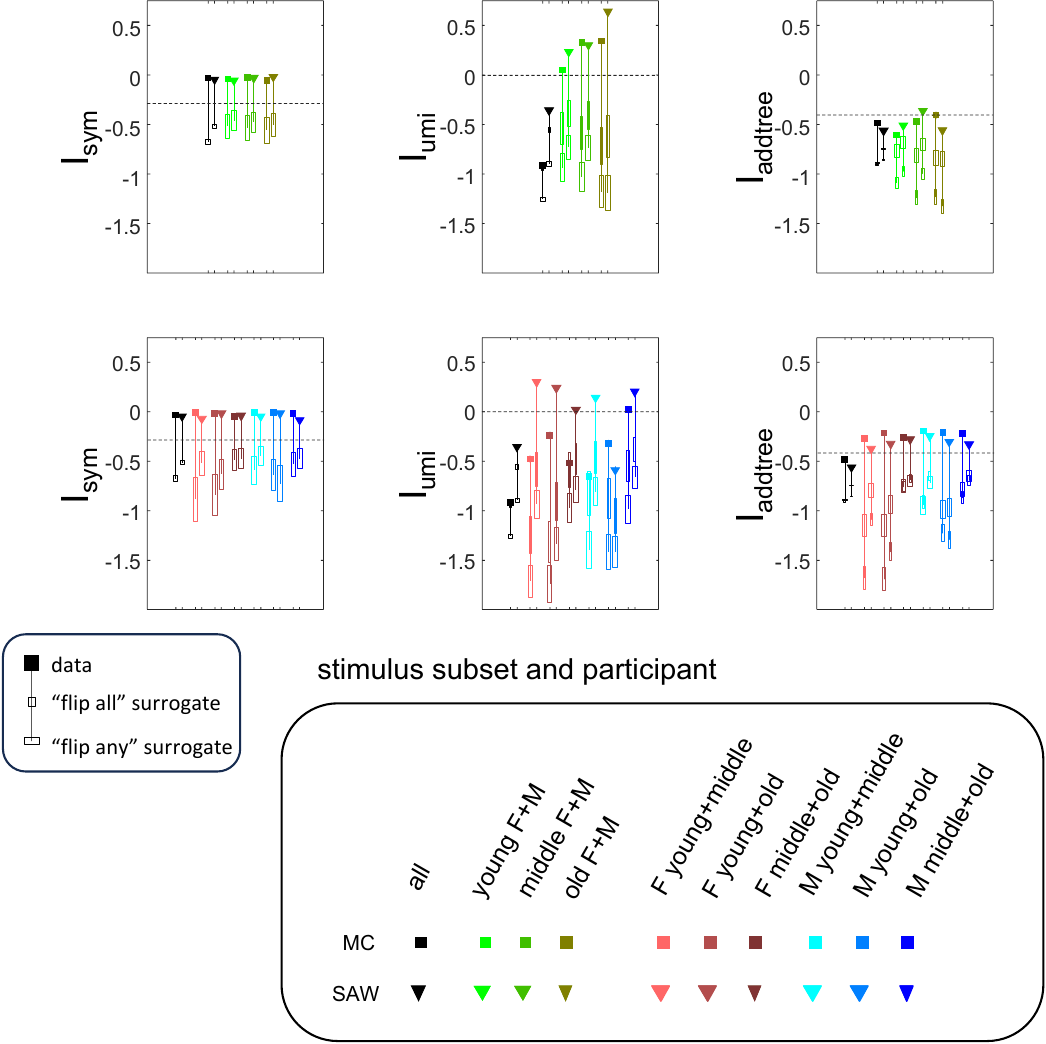}
    \caption{Indices \(I_{\text{sym}}\), \(I_{\text{umi}}\), and \(I_{\text{addtree}}\) for the faces experiment. 
    Stimulus subsets are indicated by symbol color; participants by symbol shape. 
    Upper row: full stimulus set (black symbols), and subsets partitioned by age. 
    Lower row: full stimulus set (black symbols, repeated), and subsets partitioned by gender, with two age ranges each. Other graphical conventions as in Figure~\ref{fig:13}.
    }
    \label{fig:14}
\end{figure}

The faces experiment used stimuli drawn from the public-domain library of faces, at \url{https://faces.mpdl.mpg.de/imeji/}, which contained color photographs of 171 individuals, stratified in three age ranges (``young'', ``middle'', ``old''). We randomly selected two males and two females from each age range, and for each individual in the faces database, used the two example photographs with neutral expressions, for a total of 24 unique images (2 genders $\times$ 3 age ranges $\times$ 2 individuals $\times$ 2 photographs of each).

The rationale for this choice of stimuli was that the above hierarchical organization might lead to a similarity structure close to ultrametric behavior. As shown in Figure~\ref{fig:14}, upper row, while this was not the case for analysis of the full dataset (\(I_{\text{umi}}<0\) the \textit{a priori} level), it was the case for the 8-stimulus subsets within each age bracket (\(I_{\text{umi}}>0\)).
Values of \(I_{\text{umi}}>0\) were also seen in data from some observers, for some subsets subdivided by gender (restricted to two age ranges, to equate the number of stimuli), as shown in Figure~\ref{fig:14}, lower row. 
Values of \(I_{\text{sym}}\) were again quite close to zero (usually \(>-0.1\)), indicating compatibility with a symmetric dis-similarity.
Values of \(I_{\text{addtree}}\) were similar to the \textit{a priori} value, but much larger than for the surrogates. 
As was the case for the texture experiment, these results were insensitive to the parameters \(a\) and \(h\) of the prior for the distribution of choice probabilities. 
Here, values of the beta-function parameter $a$ obtained by maximum likelihood ranged from approximately $0.1$ to $0.5$; results similar to those of Figure~\ref{fig:14} were obtained with setting \(a=0.3\) for all datasets. 
Also as was the case for the texture experiment, findings for \(I_{\text{sym}}\) and \(I_{\text{addtree}}\), were not substantially changed when \(h\) was determined by maximum likelihood -- even though the typical values of \(h\) were larger (median, \(6 \cdot10^{-2}\)), supporting the idea that some underlying choice probabilities were close to 0.5.

\subsubsection*{Brightness}

% FIG 15
\begin{figure}
    \centering
    \includegraphics[width=\linewidth]{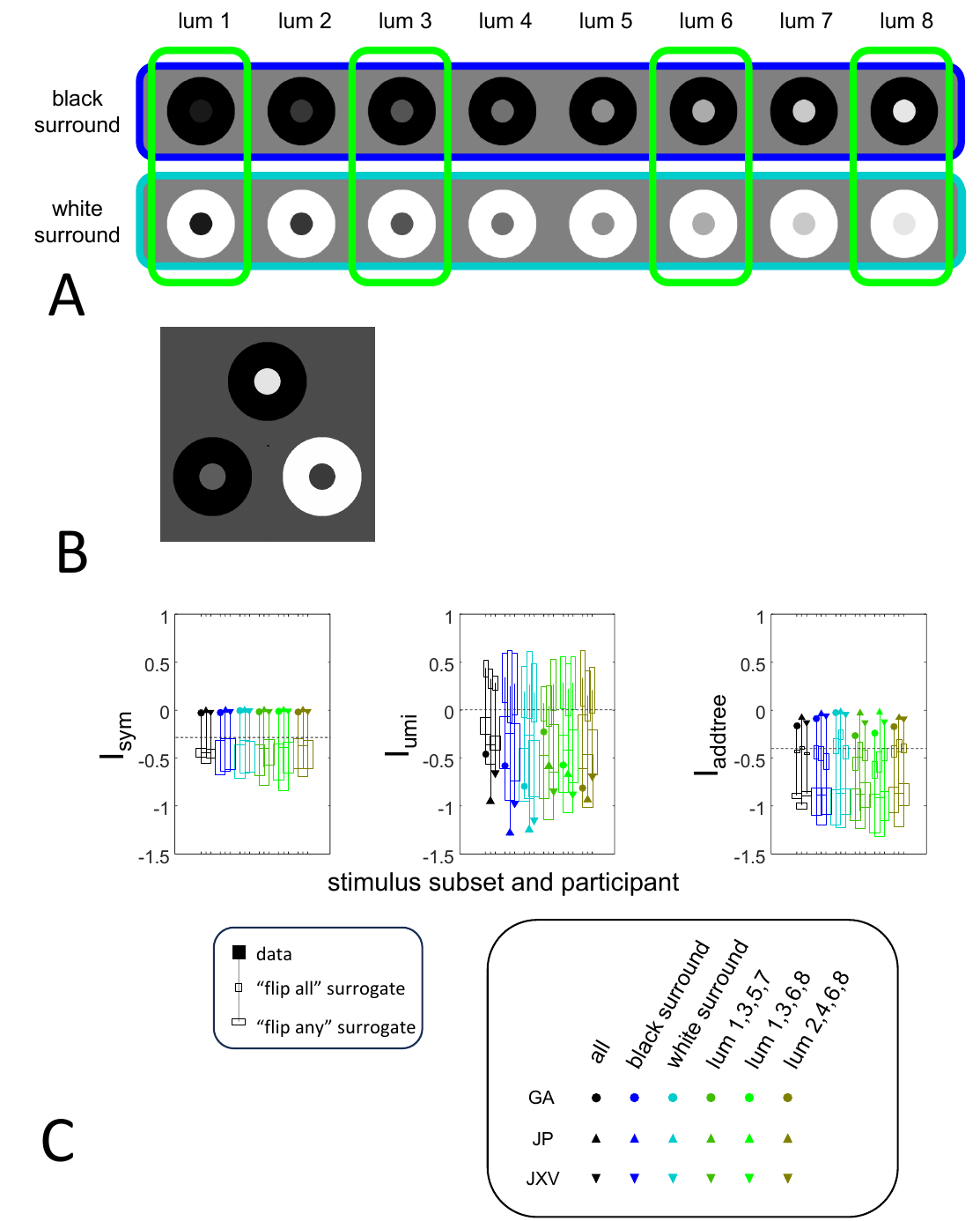}
    \caption{Panel A. Stimuli for the brightness experiment. Each stimulus had a disk-and-annulus configuration, in which the disk had one of 8 luminances (columns) and either a black (upper row) or a white (lower row) surround annulus. The colored lines encircle three of the stimulus subsets used in Panel C. 
    Panel B. A sample trial. 
    Panel C. Indices \(I_{\text{sym}}\), 
    \(I_{\text{umi}}\), and \(I_{\text{addtree}}\) for the brightness experiment calculated from all trials (black symbols), from trials with 8-element subsets with only one of the two kinds of surround (blue symbols), and from 8-element subsets with both surrounds (green symbols). 
    Other graphical conventions as in Figure~\ref{fig:13}.
    }
    \label{fig:15}
\end{figure}

The brightness experiment consisted of judgments of brightness dis-similarity for the set of disk-and-annulus stimuli as shown in Figure~\ref{fig:15}A.
This disk-and-annulus configuration has been extensively used to study the effect of the context surround on the appearance of the inner disk \cite{Gilchrist2006, Wallach1948}. 
A light surround is expected to make the inner disk appear darker, and conversely,
darker surround is expected to make the inner disk appear lighter. While it is generally assumed that this shift in appearance is along a one-dimensional brightness continuum, the evidence is ambiguous \cite{Murray2021}. 
For example, Madigan and Brainard \cite{Madigan2014} found that one dimension suffices to explain brightness similarity ratings, while Logvinenko \& Maloney \cite{Logvinenko2006} found that dis-similarity ratings under different illuminations required a 2-dimensional perceptual space.

This open question motivated the stimuli used in the present experiment: a gamut of 8 disk luminances, presented with either of 2 surround contexts (Figure~\ref{fig:15}A).
Participants judged the brightness of the inner disk for triplets constructed from all possible combinations of disk luminance and surround (Figure~\ref{fig:15}B). 
If brightness is one dimensional, then dis-similarity judgments for the full set of stimuli should be described by a one-dimensional model, which is a special case of an addtree model (Appendix~\ref{appendix2}). 
If, on the other hand, the surround produces differences in appearance that are not one dimensional, then the full set of judgments should be incompatible with an addtree distance. 
Under this hypothesis, restricting the judgments to stimuli with the same surround (the subsets encircled by the dark and light blue lines in Figure~\ref{fig:15}A) should recover a one-dimensional structure and compatibility with an addtree model, while restriction to a same-sized set but with two kinds of surrounds (green lines in Figure~\ref{fig:15}A) should preserve the incompatibility.

Figure~\ref{fig:15}C shows the results. For the full stimulus set (black symbols), \(I_{\text{addtree}}\) is close to zero (\(>-0.17\)), and substantially higher than the \textit{a priori} value, for all three participants. Even higher indices are found for the 8-element stimulus subsets of only black (\(I_{\text{addtree}}> 0.09\)) or only white 
(\(I_{\text{addtree}}> 0.05\) ) surrounds (blue symbols in Figure~\ref{fig:15}C). This is consistent with the notion that, when context (i.e., the surround) is held constant, dis-similarity judgments correspond to distances in a one-dimensional space.
However, when \(I_{\text{addtree}}\) was computed for 8-element subsets of the stimuli in which judgments were made across two surrounds (green symbols in Figure~\ref{fig:15}C), \(I_{\text{addtree}}\) was lower, and varied substantially across participants. Data from GA always yielded the lowest value (\(I_{\text{addtree}}\) was $-0.26$ to $-0.17$) and JP the highest value close to zero (\(I_{\text{addtree}}\) was $-0.08$ to $-0.03$).

These findings show that when the surround context is constant, judgments are compatible with an addtree model, but there is inter-observer variability when judgments are made across two surround contexts. The variability is not surprising, as previous research has shown that individual idiosyncrasies can play a substantial role when disk-in-context stimuli are used to study brightness or color \cite{Radonjic2015}. 
Our method seems to be capturing these inter-individual differences, but -- as we are focusing on a demonstration of the analysis methods -- we do not attempt to probe the basis for this difference here.

Similar to the texture and faces experiments, values of \(I_{\text{sym}}\) are all close to zero for the brightness dataset, indicating compatibility with symmetric dis-similarity judgments. 
Ultrametric indices \(I_{\text{umi}}\) are below the \textit{a priori} value for all cases, indicating incompatibility with an ultrametric model, as expected for a one-dimensional geometry.

Also as in the texture and faces experiments, results were robust to changes of analysis details. 
Values of the beta-function parameter $a$ obtained by maximum likelihood ranged from approximately $0.07$ to $0.22$; 
results similar to those of Figure~\ref{fig:15}C were obtained with setting \(a=0.1\) for all datasets. 
Findings for \(I_{\text{sym}}\) and \(I_{\text{addtree}}\) were not substantially changed when \(h\) was determined by maximum likelihood, yielding values of \(h\) with a median of \(5 \cdot10^{-2}\), comparable to the faces dataset.

\section{Discussion}

The main goal of this paper is to advance a strategy for connecting similarity judgments of a collection of stimuli to inferences about the structure of the domain from which the stimuli are drawn. 
The starting point is an experimental dataset in which the judgments are assumed to be independently drawn binary choices distributed according to the underlying choice probabilities. 
We assume that for each triad (a reference stimulus and two comparison stimuli), the comparison stimulus that is more often judged to be more similar to the reference is at a shorter distance from it, but we do not assume, or attempt to infer, a quantitative relationship between choice probabilities and the distances. 
This approach also takes into account the possibility that each triad may have its own ``pointwise'' transformation that links choice probability and distance. 
While we recognize that judgments may be uncertain, we refrain from postulating a noise model or a decision model -- or even that sensory or decision noise is uniform throughout the space.

Despite the relative paucity of assumptions, we show that it is possible to characterize dis-similarity judgments along three lines: compatibility with symmetry, compatibility with an ultrametric model, and compatibility with an additive tree model. 
These characteristics are functionally significant aspects of a domain’s organization. 
Symmetry (or its absence) has implications for the mechanism(s) by which comparisons are made \cite{Tversky1977, Tversky1982}. 
For symmetric similarity judgments, addtree models, but not ultrametric models, are consistent with the Tversky contrast model \cite{Sattath1977}.
More broadly, semantic domains are anticipated to be consistent with a hierarchical model of similarity judgments (ultrametric or addtree), while domains of features are not \cite{Kemp2008, Saxe2019}.
It is also worth noting that one-dimensional domains are a special case of the addtree model, so the present approach can address whether the apparent ``curvature'' in a one-dimensional perceptual space can be eliminated by alternative choices of the linkage between distance and decision -- a limitation of the analysis in \cite{Bujack2022}.
Furthermore, our method is sensitive enough to reveal inter-individual differences: for some participants data are compatible with the addtree model and for others, not (or less so) -- consistent with other studies of the influence of context \cite{Radonjic2015}, and an interesting area for further investigation.

\subsection*{Comparisons to other methods}

The present strategy, in which a main consideration is to keep assumptions about metrical 
distances to a minimum, is complementary to other ways of analyzing similarity judgments.
A classical and commonly-used approach, non-metric multidimensional scaling \cite{deLeeuw1982, Tsogo2000}, explicitly postulates that the original data (here, the choice probabilities) reflect a monotonic transformation of a metric distance. 
The distance is taken to be the Euclidean distance, but distances in a hyperbolic or spherical geometry can also be used. 
An important related approach for one-dimensional models is maximum likelihood difference scaling \cite{Knoblauch2008, Maloney2003}, which -- via a decision model -- takes into account the noisy nature of psychophysical judgments. 
This approach can also be extended to multidimensional models, but the need for a decision model remains \cite{Victor2017, Waraich2022}.

The spirit of our approach is similar in some ways to that of \cite{Tversky1986}, who used analyses of the statistics of nearest neighbors (``centrality'' and ``reciprocity'') to assess models of perceptual spaces. 
Their study showed that tree-like models are preferred to low-dimensional metrical models for ``cognitive'' categories (furniture, musical instruments), but not for low-level domains (lights, shapes, tones). 
As in the current approach their analyses depended only on rank orders of distances, rather than metrical relationships. 
However, the starting point for their analysis is a complete set of pairwise distances with a definite rank order. 
In contrast, the present approach begins with a set of triadic judgments and is explicit that these serve only as estimates of an underlying choice probability.

Our approach also is related to that of topological data analysis (TDA) via persistent homologies \cite{Dabaghian2012, Giusti2015, Singh2008, Zhou2018}. 
Like our approach, TDA avoids the need to postulate a specific relationship between dis-similarities and distances, as the Betti numbers are calculated from a sequence of graphs that depend only on the rank order of similarity judgments. 
However, construction of this sequence of graphs requires a globally uniform linkage between triadic judgments and relative distance, and also, that every pairwise distance is included in the measured triads. 
The characterizations yielded by TDA are also complementary: they focus on dimensionality and homology class, rather than the characterizations considered here.

Finally, we note another approach that directly seeks to identify features of ultrametric behavior in neural data. 
Treves \cite{Treves1997} developed an information-theoretic measure of dis-similarity of neural responses to faces. 
The strategy for seeking evidence of ultrametric behavior was to examine the ratio, within each triplet, of the middle distance to the largest. 
This ratio, which would be 1 for an ultrametric, was found in that study to be closer to 1 than for expected by chance. 
Nonparametric generalizations of this approach may permit a relaxation of the assumed linkage between extensions the information-theoretic measure and dis-similarity, and even an evaluation of addtree models -- but in contrast to our approach, it begins with a set of responses to each stimulus, rather than a sampling of triadic comparisons.

\subsection*{Experimental design}

The analysis of simulated and experimental datasets carry implications for experimental design.
From the point of view of demonstrating departure of the indices from their \textit{a priori} value, the simulated datasets show that as few as 2 observations per triad suffice under some circumstances, and may even be optimal (Figure~\ref{fig:9}). 
However, this optimal number of samples per trial is likely to be dependent on how severely the properties tested by each index are violated and the actual geometry of the underlying space: for subtle departures (e.g., Figure~\ref{fig:8}), a larger number of observations (8 or more) may be needed.
Note, however, that observations contribute to the indices only if the triads are in the same triplet or tent, a consideration that may have bearing on the way that triads are selected for presentation in an experiment.

The experimental datasets show that analyses of subsets of the full dataset can be revealing: in the texture dataset (Figure~\ref{fig:13}), the responses to the subset of stimuli contained in the ``vees'' are less compatible with an addtree distance than for subsets contained in the ``axes''; in the faces dataset (Figure~\ref{fig:14}), data from images restricted to one gender and including the three age ranges were less compatible with an ultrametric distance than data restricted to one age range but including both genders; for the brightness dataset (Figure~\ref{fig:15}), subsets that included both luminance contexts were less compatible with an addtree distance than subsets that included only one luminance context. These observations suggest that the current approach can be incorporated into an adaptive design strategy, in which calculation of indices within subsets can be used to select stimuli or triads for more intensive study. This is a potential area for future work.

\subsection*{Limitations and caveats}

Keeping assumptions to a minimum necessarily leads to certain limitations. The indices for 
symmetry and addtree structure reflect necessary, but not sufficient, conditions for global compatibility. 
Moreover, though the indices can be used to compare the behavior of the data with behavior of a model, they are not measures of goodness of fit: the indices merely measure to what extent the data act to concentrate the \textit{a priori} distribution of choice probabilities into the subset of choice probabilities that have a particular characteristic. 
Critically, the extent of concentration will depend on the typical coverage of each triad: a greater number of trials of each triad provides better estimates of the underlying choice probabilities, and thus, can move the indices further from their \textit{a priori} values. 
For this reason, the examples above (``axes'' vs. ``vees'' in Figure~\ref{fig:13}; subdivision by age vs. gender in Figure~\ref{fig:14}, same context vs. different contexts in Figure~\ref{fig:15}C) focus on comparisons of indices between datasets with a similar number of stimuli, and a similar coverage of each.

Because the approach relies on rank-order comparisons and does not attempt to estimate 
perceptual distances in a metrical sense, it cannot distinguish an addtree geometry from a Euclidean one unless the geometries lead to differences in the rank orders of distances. Thus, while the incompatibility of a circular stimulus space with an addtree geometry is readily detected by \(I_{\text{addtree}}\) (Figure~\ref{fig:7}, bottom row), the Euclidean distance in a two-ray stimulus space (Figure~\ref{fig:8}) is only distinguishable from an addtree geometry when the angle between the rays is acute.

A necessary component of putting this approach into practice is the need to infer underlying choice probabilities from a finite set of experimental observations, and this step -- even if principled -- entails caveats. Here, we use a Bayesian framework and a two-parameter family of priors: the beta-function prior of eq. \eqref{eq:3.4}, which suffices for \(I_{\text{sym}}\) and \(I_{\text{addtree}}\), and its generalization that includes point mass, eq. \eqref{eq:3.27}, which is needed for \(I_{\text{umi}}\). 
While this choice of prior has some theoretical justification, it is primarily based on heuristics and computational advantages. 
For the model scenarios considered here, simulations (Figure~\ref{fig:10}) show that within this family (which includes concave, convex, and flat distributions) the influence of the prior is mild, especially when there are several observations per triad. 
The need to consider priors with a point mass at \(\frac{1}{2}\) for \(I_{\text{umi}}\) arises because a strictly ultrametric geometry requires that some choice probabilities are exactly \(\frac{1}{2}\); had we only considered the beta-function prior, this has a posterior probability of zero, for any finite dataset. 
So, in order to formalize a useful sense in which the data have an ultrametric character, we ask how rapidly this probability moves away from zero as we allow the prior to include a fraction \(h\) of choice probabilities that are exactly \(\frac{1}{2}\).
As we show in the simulations (Figure~\ref{fig:6}) and also in the analysis of the faces dataset (Figure~\ref{fig:14}), this is an effective strategy -- but there may be other approaches more suitable in other specific circumstances.

The analysis of the synthetic ultrametric dataset $Tree- UM$ reveals one situation in which the choice of the prior is important. In this dataset, the underlying choice probabilities have a substantial point mass at \(\frac{1}{2}\). 
Excluding this possibility from the prior (Figure~\ref{fig:7}, top row) fails to reveal
compatibility with addtree structure, while including it (Figure~\ref{fig:11}C, top row) demonstrates compatibility. 
With this in mind, it may be prudent to analyze experimental datasets both with assuming the beta-function prior (\(h=0\)) and allowing \(h\) to be set by maximum likelihood. 
In the experimental datasets presented here, this did not result in a substantial change.

Our approach also provides two strategies to help ensure that the indices are dominated by the data, rather than by the prior. First, each index has a definite \textit{a priori} value, independent of the prior for the choice probabilities, provided only that it is symmetrical: \(\log (3/ 4)\) (eq. \eqref{eq:3.25}) for \(I_{\text{sym}}\), 0 (eq. \eqref{eq:3.46}) for \(I_{\text{umi}}\), and \(\log (2/ 3)\) (eq.\eqref{eq:3.58}) for \(I_{\text{addtree}}\).
Second, the two kinds of surrogates (``flip any'' and ``flip all'') have the same \textit{a priori} probability as the data, though they lack the inter-relationships among choice probabilities within triplets or tents. 
Computed values of the indices that do not deviate substantially from the \textit{a priori} values, or do not deviate substantially from indices computed from the surrogates, should be viewed with caution, as they suggest that the indices are dominated by the prior, rather than by the data. The ``flip all'' surrogate provides a stronger test, since it retains not only the \textit{a priori} choice probabilities, but also within-triplet pairwise correlations.

\subsection*{Variations, extensions and open issues}

There are a number of variations and alternative analytic strategies that might be used to modify and enhance the present framework, as well as some open issues that would enable it to be further extended.

One such variation relates to how the locally-computed likelihood ratios are combined to form an overall characterization of the dataset. Here, we elected to simply average their logs, i.e., to compute an index that can be interpreted as an average measure of local compatibility. Instead, one could look at other aspects of the distribution of these ratios. For example, one could also transform the likelihood ratio computed for each triplet or triad into a Bayes factor, via \(BF = \frac{LR}{1-LR}\) 
(see comment following eq. \eqref{eq:3.23}), and then combine the Bayes factors -- yielding a Bayes factor that compares the hypotheses that all triplets are compatible, vs. no triplets are compatible.

One might also consider circumventing the difficulties related to underlying choice probabilities that are exactly \(\frac{1}{2}\) by allowing subjects to respond that the two comparison stimuli are equally dis-similar. 
At first glance, this sounds promising, but the problem is that, by allowing subjects to be explicitly indecisive, it (i) reduces the extent to which their responses will reveal subtle but systematic differences in perceived dis-similarity, and (ii) introduces a major confound, namely, that subjects may differ substantially in their criterion for using this alternative. 
The situation worsens if subjects change their criterion during the experiment, which introduces yet another uncontrolled confound. 
So although some psychophysicists have studied the usefulness of allowing an ``uncertain'' response (e.g. \cite{Watson1973}), 
it is not recommended (e.g. \cite{Kingdom2016, Wickens2001}),
and consequently, it is rarely used in psychophysical studies. 
Accordingly, our approach and our sample datasets all used forced-choice paradigms without the possibility of an ``uncertain'' response.

There are also alternative strategies for surrogate generation. The ``flip any'' and ``flip all'' surrogates used here rely on replacing the observed responses for a triad by their complement. Alternatively, one could replace the data from each triad with the data from a randomly-selected triad. 
As is the case for the ``flip any'' and ``flip all'' surrogates, this permuted surrogate would have the same \textit{a priori} probability as the data. 
Correlations within a triad be destroyed, and, in addition, judgments that were more certain (i.e,. observed choice probabilities near 0 or 1) would be randomly replaced by judgments that were less certain. 
Thus, a permutation test would sample the independent prior more evenly than the ``flip'' surrogates, but it would not determine whether the values of the indices could be due to co-occurrence of easy vs. hard judgments within triplets.

While the present approach aims at determining compatibility with an ultrametric or addtree distance, it stops short of attempting to determine the specific ultrametric or addtree structure. 
Existing methods for taking this step require a complete set of dis-similarity measures \cite{Abraham2015, Barthelemy1991, Hornik2005, Sattath1977}, along with the assumption that the transformation from triadic choice probabilities into dis-similarities is uniform across the space. 
Choice probabilities provide constraints even without this assumption -- for example certain relationships among the triadic judgments involving four points are sufficient for a pointwise addtree model -- but it is unclear whether these constraints are sufficient for a global model, or how such a global model can be determined.

The central strategy of the present framework is to translate the consequences of metrical constraints on distances and their sums into constraints on rank orders of dis-similarities. 
Here, we have chosen to implement this strategy in its simplest form. In considering compatibility with a symmetric distance, we only considered triadic comparisons within a triplet -- even though (see eq. \eqref{eq:2.6}) there are further constraints on triadic comparisons within larger sets of triads. 
On larger stimulus sets, a complete analysis would also need to take into account the constraints of transitivity (eq. \eqref{eq:3.53}). 
If these constraints are taken into account, the theoretical benefit that each of the local contributions to \(I_{\text{sym}}\) is determined from a non-overlapping set of triads (eq. \eqref{eq:3.24} would be lost. 
Relatedly, an alternative strategy for assessing addtree structure is to consider all 12 of the triadic comparisons among four stimuli together, rather than just the six-triad subsets that constitute a tent. 
It is unclear whether the additional complexity entailed by any of these considerations would translate into practical benefit.

We speculate that the overall framework of relating graphical distance models to rank-order comparisons can be extended to more complex graphs. Specifically, the observation that a metric that obeys the four-point condition can always be realized by the path metric on a weighted acyclic graph \cite{Buneman1974} suggests the possibility of a succession of further characterizations. 
By definition, acyclic graphs have no three-cycles. 
An isolated 3-cycle with nodes \(a_1\), \(a_2\), and \(a_3\) can always be removed by adding a node \(c\) , with distances 
\(d(c,a_1) = (d(a_1,a_2) + d (a_1,a_3) - d(a_2, a_3))/2\), etc.; 
this quantity is guaranteed to be non-negative via the triangle inequality, and \(d(c, a_1) + d(c,a_2) = d(a_1, a_2)\).
Thus, ruling out compatibility with an addtree model via Proposition~\ref{pr:3} also implies that the dis-similarity structure is incompatible with distances on a weighted acyclic graph, or on a weighted graph with only isolated three-cycles. 
Consequently, a graph with two non-disjoint three-cycles or a four-cycle, is required. Perhaps more elaborate conditions analogous to those of Proposition~\ref{pr:3} could then rule out compatibility with distances on graphs with some greater level of cycle structure or connectivity.

In this regard, it is interesting to note that the ultrametric condition and the four-point condition have a similar structure: both state that among three numbers (three single distances for the ultrametric, three pairwise sums for the four-point condition), the largest two must be identical. This intriguing similarity raises the further possibility of a sequence of analogous conditions, each specifying a progressively less-restrictive aspect of a set of dis-similarity judgments -- such as compatibility with planar graphs (and, more generally, dimensionality), or other topological characterizations, such as statements about Betti numbers.

%%%%%%%%%%%%%%%%%%%%% APPENDICES %%%%%%%%%%%%%%%%%%%%%
% We included the Appendices here, using the \appendix command. 
% With it,  every appendix get a different letter, and equations follow
% that letter. Other MNA papers are also formatted this way, 
% e.g. https://mna.episciences.org/8775/pdf
\appendix
\titleformat{\section}{\normalfont\large\bfseries}{Appendix \thesection: }{0em}{}

\section{Metric Spaces}
\label{appendix1}

This appendix demonstrates that pointwise or setwise compatibility of a set of choice probabilities with a symmetric dis-similarity
\(D(x,y)\) implies pointwise or setwise compatibility with a metric-space structure.

A metric space is defined to be a set of points \(\{x, y, z, \dots \}\), along with a \textit{metric} \(d(x,y)\) that associates each pair of points with a real number, and that satisfies three properties:

\begin{equation}\label{eq:4.1}
d(x,y) \geq 0 \quad \text{and} \quad d(x,y)=0 \iff x=y,
\end{equation}

\begin{equation}\label{eq:4.2}
\text{symmetry}: d(x,y) = d(y,x),
\end{equation}
and
\begin{equation}\label{eq:4.3}
\text{the triangle inequality}: d(x,y) \leq d(x, z) + d(y,z).
\end{equation}

As we have generically assumed (in eq. \eqref{eq:1.4}) that dis-similarity \(D(x, y)\)  satisfies \eqref{eq:4.1} and we now add that it is symmetric (eq. \eqref{eq:4.2}), we need to show that we can replace \(D(x, y)\) -- which need not satisfy the triangle inequality -- by a metric \(d(x, y)\) which both satisfies the triangle inequality and also accounts for the choice probabilities via

\begin{equation}\label{eq:4.4} 
    R(r;x, y) > \frac{1}{2} \iff d(r, x) < d(r, y).
\end{equation}
It suffices to find a monotonic transformation of the dis-similarity
\begin{equation}\label{eq:4.5}
    d(x, y) = G(D(x, y)).
\end{equation}
for which \(d(x, y)\) satisfies the triangle inequality. A suitable transformation for this purpose is 
\begin{equation}\label{eq:4.6}
    G(u)  = 
     \left\{\begin{aligned}0, \quad u=0 \\
     1 + \frac{u}{1+u}, u> 0
     \end{aligned}
     \right.
\end{equation}
For distinct elements $a$ and $b$, $D(a, b) > 0$ and consequently $d(a, b) = G(D(a, b))$ is between 1 and 2. So the left-hand side of \eqref{eq:4.3} is always $<2$ and each term on the right-hand side is $>1$, implying that the triangle inequality holds. 

\section{Addtree Models and the Four-Point Condition}
\label{appendix2}

This Appendix demonstrates some well-known and basic \cite{Buneman1974, Dobson1974, Sattath1977} relationships between the four-point condition, the triangle inequality, and the ultrametric inequality.

To show that the four-point condition \eqref{eq:2.13}
implies the triangle inequality, we set \(w=x\). Then \eqref{eq:2.13} becomes 
\begin{equation}\label{eq:5.1}
    \text{None of the three quantities} \quad \left\{
    \begin{aligned}
           &d(u, v) \\
        d(u, x) &+ d(v, x) \\
        d(u, x) &+ d(v, x)
    \end{aligned}
    \right\} \quad \text{is strictly greater than the other two.}
\end{equation}

That is, $d(u, v) \leq d(u, x) + d(v, x)$, which is the triangle inequality. \\
To show that the ultrametric inequality \eqref{eq:2.8} implies the four-point condition, we first relabel the points \(u, v, w,\) and $x$ so that $d(u, v)$ is the smallest distance, and 
\begin{equation}\label{eq:5.2}
    d(u, v) \leq d(u, w) \leq d(u, x).
\end{equation}

Consider the triangle with vertices $u, v,$ and $w$. Since $d(u, v) \leq d(u, w)$, the ultrametric inequality, which says that the two longest sides must have the same length, requires that the length of the third side, $d(v, w)$, is equal to $d(u, w)$:
\begin{equation}\label{eq:5.3}
    d(u, v) \leq d(u, w) \implies d(v, w) = d(u, w). 
\end{equation}

Similarly, applied to the triangle with vertices $u, v,$ and $x$, the ultrametric inequality yields 
\begin{equation}\label{eq:5.4}
    d(u, v) \leq d(u, x) \implies d(v, x) = d(u, x).
\end{equation}

Combining these two yields
\begin{equation}\label{eq:5.5}
    d(u, w) + d(v, x) = d(v, w) + d(u, x). 
\end{equation}
Applied to the triangle with vertices $u, w,$ and $x,$ the ultrametric inequality yields, 
\begin{equation}\label{eq:5.6}
    d(u, w) \leq d(u, x) \implies d(w, x) = d(u, x).
\end{equation}
Combining this with the assumption that $d(u, v)$ is the shortest distance, \eqref{eq:5.2} yields 
\begin{equation}\label{eq:5.7}
    d(u, v) + d(w, x) = d(u, v) + d(u, x) \leq d(v, w) + d(u, x). 
\end{equation}
Together, \eqref{eq:5.5} and \eqref{eq:5.7} constitute the four-point condition. \par
\textbf{Examples:} \par
For four distinct points on a line and distances given by the ordinary Euclidean distance $d(y,z) = \vert z-y \vert$, the addtree conditions hold. Taking $u<v<w<x$,
\begin{equation}\label{eq:5.8}
    d(u, w) + d(v, x) = (w-u) + (x-v) = (w-v) + (x-u) = d(u, x) + d(v, w),
\end{equation}
while
\begin{equation}\label{eq:5.9}
    \begin{aligned}
            d(u, v)+ d(w, x) = (v-u) + (x-w) = (v+x) - (u+w) < \\(w+x) - (u+v) = (x-u) + (w-v) = d(u, x) + d(v, w).
    \end{aligned}
\end{equation} \par
However, the ultrametric inequality does not in general hold, since \(d(u, v) < d(u, w) < d(u, x)\).\par
For the standard distance between four distinct points in general position in a plane, the four-point condition does not hold, since the three pairwise sums are typically distinct. 

\section{Setwise Sufficiency for the Four-Point Condition}\label{appendix3}

Here, we prove a partial converse to Proposition~\ref{pr:3}: that falsification of the conjunction \eqref{eq:2.16}
suffices to ensure that the dis-similarities are setwise-compatible with an addtree model. More precisely:

\begin{proposition}[sufficient ordinal conditions for addtree]\label{pr:4}
    If, for each relabeling of a set of four points \(\{z, a, b, c\}\), \eqref{eq:2.16} does not hold, and all of the dis-similarities are unequal, then there is a strictly monotonic-increasing transformation $d(x, y) = F(D(x, y))$ of the dis-similarities into distances, for which the four-point condition \eqref{eq:2.13} holds. 
\end{proposition}
\begin{proof}
    We choose $F$ to be of the form 
    \begin{equation}\label{eq:6.1}
        F(D) = \left\{ 
        \begin{aligned}
            D + k, D \geq D_0 \\
            D, D < D_0
        \end{aligned}
        .\right.
    \end{equation}
    so the demonstration rests on finding an appropriate choice of $D_0$ and $k$. We note that we can ignore whether the quantity that results from the monotonic transformation satisfies the triangle inequality, since (as is well-known) Appendix~\ref{appendix2} shows that the four-point condition implies the triangle inequality. \par
    We refer to the three quantities in the four-point condition \eqref{eq:2.13} and the corresponding sums of pairs of dis-similarities as ``pairsums.'' For convenience, we reproduce \eqref{eq:2.16} below, replacing the non strict inequalities by strict ones, since the dis-similarities are assumed to be unequal:
    \begin{equation}\label{eq:c_2}       
    	 \begin{aligned}        
    	 	D(z, c) \geq D(z, b) \quad\text{and}\quad D(z,c) \geq D(z, a) \quad\text{and}\\        
    	 	D(a, b) \geq D(c, a) \quad\text{and}\quad D(a, b) \geq D(b, c).        
    	 \end{aligned}    
     \end{equation}

WLOG, relabel the points so that $D(z, c)$ is the greatest dis-similarity. \par

\textbf{\textit{Case $\mathbf{A.}$}} $D(z, c) + D(a, b)$ is not the largest pairsum. Choose $D_0 = D(z, c)$ in \eqref{eq:6.1}. Then $d(z, c) + d(a, b) = F(D(z, c)) + F(D(a, b)) = D(z, c) + D(a, b) + k$, but the other pairsums are unchanged by $F$ since their components are all strictly less than $D_0$. Setting $k$ equal to the difference between $D(z, c)+D(a, b)$ and the next-largest pairsum yields a transformation to distances that satisfies the four-point condition. \par
\textbf{\textit{Case $\mathbf{B.}$}} $D(z,c)$ is the largest dis-similarity and $D(z, c) + D(a, b)$ is the largest pairsum. Take $D_0 > D(a, b)$ but smaller than the next-largest dis-similarity. Then again \par
\(d(z,c) + d(a,b) = F(D(z,c)) + F(D(a,b)) = D(z,c) + D(a,b) +k\). We show that applying $F$ increases at least one of the other two pairsums by $2k$. Then, choose $k$ to be the difference between $D(z,c) + D(a,b)$ and the next-largest pairsum. Since $F$ cannot increase a pairsum by more than $2k$, applying $F$ results in equating the two largest pairsums. \par
To show that applying $F$ increases at least one of the other two pairsums by at least $2k$, we need to show that for at least one of the other two pairsums, both terms are greater than $D(a, b)$. Since $D(z, c)$ is the largest dis-similarity, the second-largest dis-similarity cannot be $D(a, b)$, since then \eqref{eq:2.16} would hold. So the second-largest dis-similarity must be between a point in $\{z, c\}$ and a point in $\{a, b\}$. We then relabel the points (switching $a$ and $b$ if necessary, and switching $z$ and $c$ if necessary) so that the second-largest dis-similarity is $D(z, a)$. In this relabeling, $D(a, b)$ cannot be the third-largest dis-similarity, since then \eqref{eq:c_2} would hold. Thus, $D(z,c)$ is the largest dis-similarity, $D(z,a)$ is the second-largest, and there are three possibilities for the third-largest dis-similarity: $D(z, b), D(a, c)$, and $D(b, c)$ (Figure~\ref{fig:16}).\par
\begin{figure}
    \centering
    \includegraphics[width=0.75\linewidth]{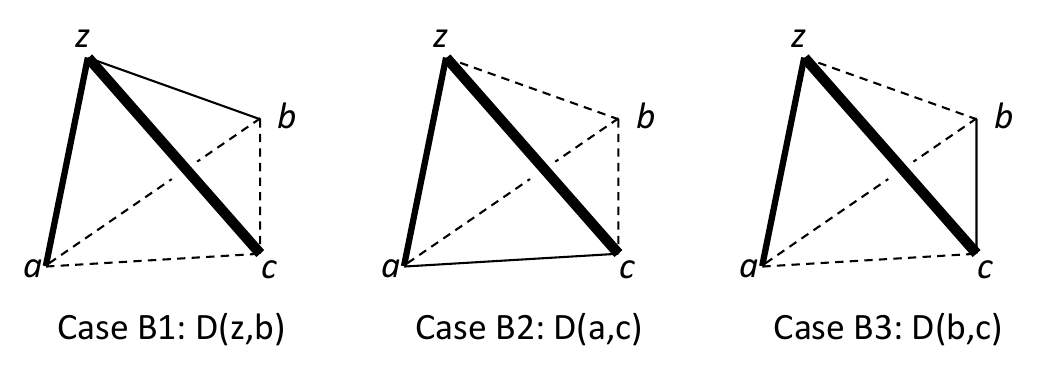}
    \caption{Three rank-orderings of the dis-similarities among four points consistent with falsification of the conjunction \eqref{eq:2.16}  used in the proof of Proposition~\ref{pr:4} (Appendix~\ref{appendix3}). In all cases, $D(z, c)$ is largest (heavy line) and $D(z, a)$ is second-largest (intermediate line). The third-largest dis-similarity can be either $D(z, b)$, $D(a, c)$, or $D(b, c) $ (thin solid lines). See Appendix~\ref{appendix3} for details.}
    \label{fig:16}
\end{figure}
\textbf{\textit{Case $\mathbf{B_1}$.}} The third-largest dis-similarity is $D(z, b)$: Either $D(a, c) > D(a, b)$ or $D(b, c) > D(a, b)$ since otherwise \eqref{eq:c_2} would hold. Since $D(z, b)$ and $D(z, a)$ are both greater than $D(a, b)$, both terms of at least one of \(D(z, b) + D(a, c)\) or \(D(z, a) + D(b, c)\) are greater than $D(a, b).$\par
\textbf{\textit{Case $\mathbf{B_2}$}}. The third-largest dis-similarity is $D(a, c)$: $D(a, b)$ cannot be larger than both $D(z, b)$ and $D(b, c)$, because otherwise, a conjunction like \eqref{eq:c_2} would hold for the tent $\{b;z,c,a\}$, as $D(b, a)= D(a, b)$ would be the largest of the tripod at $b$, and $D(z, c)$ would be the largest of the base. So at least one of $D(z, b) > D(a, b)$ or $D(b, c) > D(a, b)$. Thus, both terms of either pairsum $D(z, b) + D(a, c)$, or $D(z, a) + D(b, c)$ are greater than $D(a, b).$\par
\textbf{\textit{Case $\mathbf{B_3}$.}} The third-largest dis-similarity is $D(b, c)$: Both terms of $D(z, a) + D(b, c)$ are larger than $D(a, b)$.
\end{proof}

%%%%%%%%%%%%%%%%%%%%%%%%%%%%%%%%%%%%%%%%%%%%%%%%%%%%%%%%%%%%%%%%%%%
%%                                                               %%
%% Supplementary Material, if any, should be provided in         %%
%% {supplement} environment  with title and short description.   %%
%%                                                               %%
%%%%%%%%%%%%%%%%%%%%%%%%%%%%%%%%%%%%%%%%%%%%%%%%%%%%%%%%%%%%%%%%%%%
\renewcommand{\thefigure}{S\arabic{figure}}
\renewcommand{\theHfigure}{S\arabic{figure}}
\setcounter{figure}{0}

\begin{supplement}
\stitle{Supplementary figures}
\sdescription{Supplement to main text Figures~\ref{fig:5} - \ref{fig:11}}

\begin{figure}
    \centering
    \includegraphics[width=\linewidth]{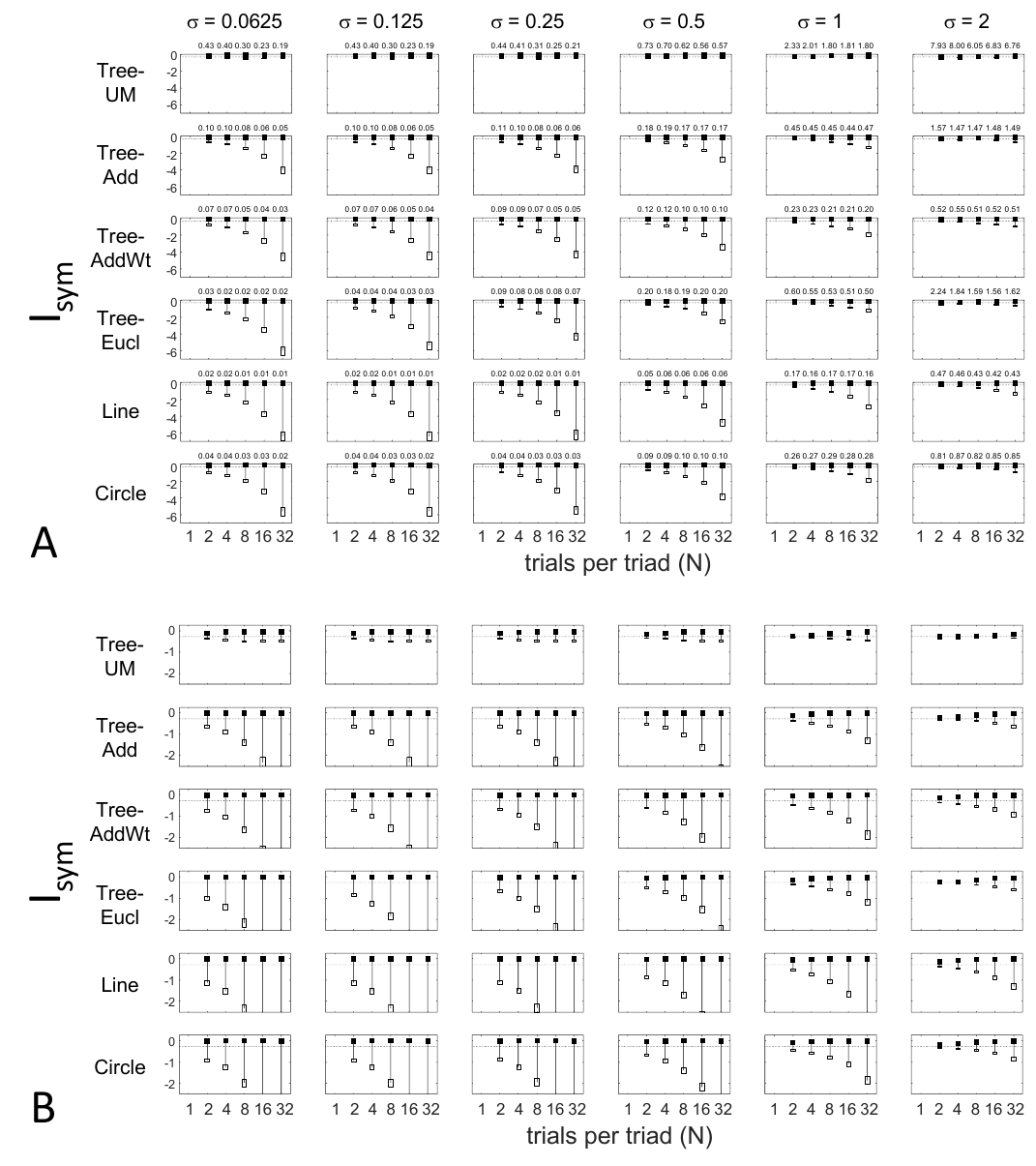}
    \caption{(supplement to main text Figure~\ref{fig:5}).
    Panel A. Behavior of $I_{\text{sym}}$ for the 15-point configurations, a range of trials per triad ($N$, abscissa), and a range of decision rules ($\sigma$, standard deviation of added noise). 
    Solid symbols indicate values of $I_{\text{sym}}$ computed from simulated decisions. 
    The dashed horizontal line indicates the \textit{a priori} value for $I_{\text{sym}}$. 
    Hollow boxes indicate values of $I_{\text{sym}}$ computed from ``flip any'' surrogates, showing the $\pm1$ standard deviation range. 
    The thin vertical lines are to aid visualization, and do not represent ranges. The beta-function prior \eqref{eq:3.4} is used, 
    with the value of the parameter $a$ (shown in fine print over each data point) determined by maximum likelihood. 
    Panel B expands the vertical scale.}
    \label{fig:s1}
\end{figure}

\begin{figure}
    \centering
    \includegraphics[width=\linewidth]{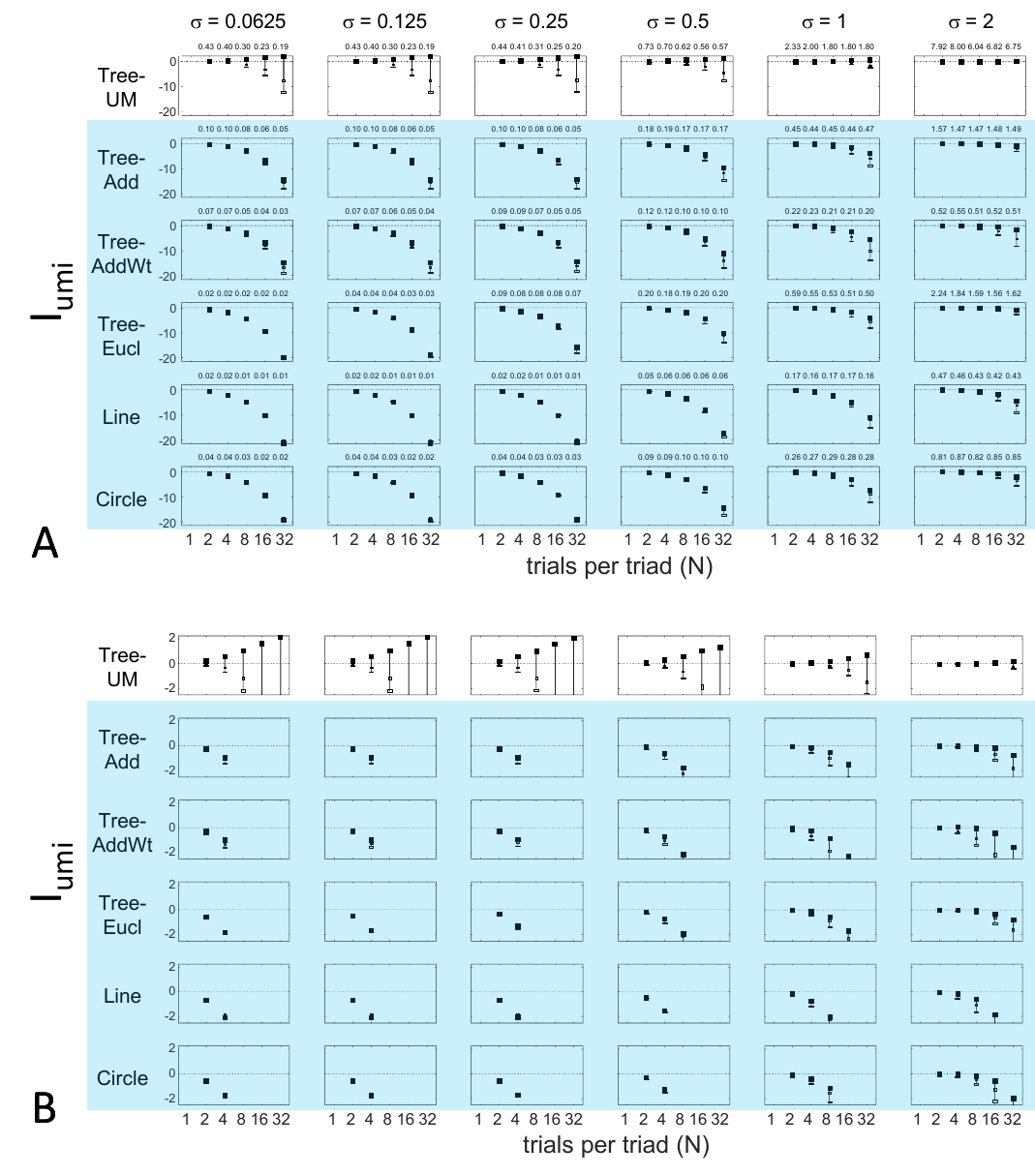}
    \caption{(supplement to main text Figure~\ref{fig:6}).
    Panel A. Behavior of $I_{\text{umi}}$ for the 15-point configurations, a range of trials per triad ($N$, abscissa), and a range of decision rules ($\sigma$, standard deviation of added noise). 
    Solid symbols indicate values of $I_{\text{umi}}$ computed from simulated decisions. 
    The dashed horizontal line indicates the \textit{a priori} value for $I_{\text{umi}}$. 
    Hollow boxes indicate values of $I_{\text{umi}}$ computed from surrogates, showing the $\pm1$ standard deviation range:
    wider boxes for ``flip any'' surrogates, narrower boxes for ``flip all'' surrogates. 
    The thin vertical lines are to aid visualization, and do not represent ranges. 
    The modified prior \eqref{eq:3.27} is used, with the value of the parameter $a$ (shown in fine print over each data point) determined by maximum likelihood and $h = 0.001$. 
    Panel B expands the vertical scale. Blue overlay indicates the simulated datasets that are incompatible with the ultrametric property.
    }
    \label{fig:s2}
\end{figure}

\begin{figure}
    \centering
    \includegraphics[width=\linewidth]{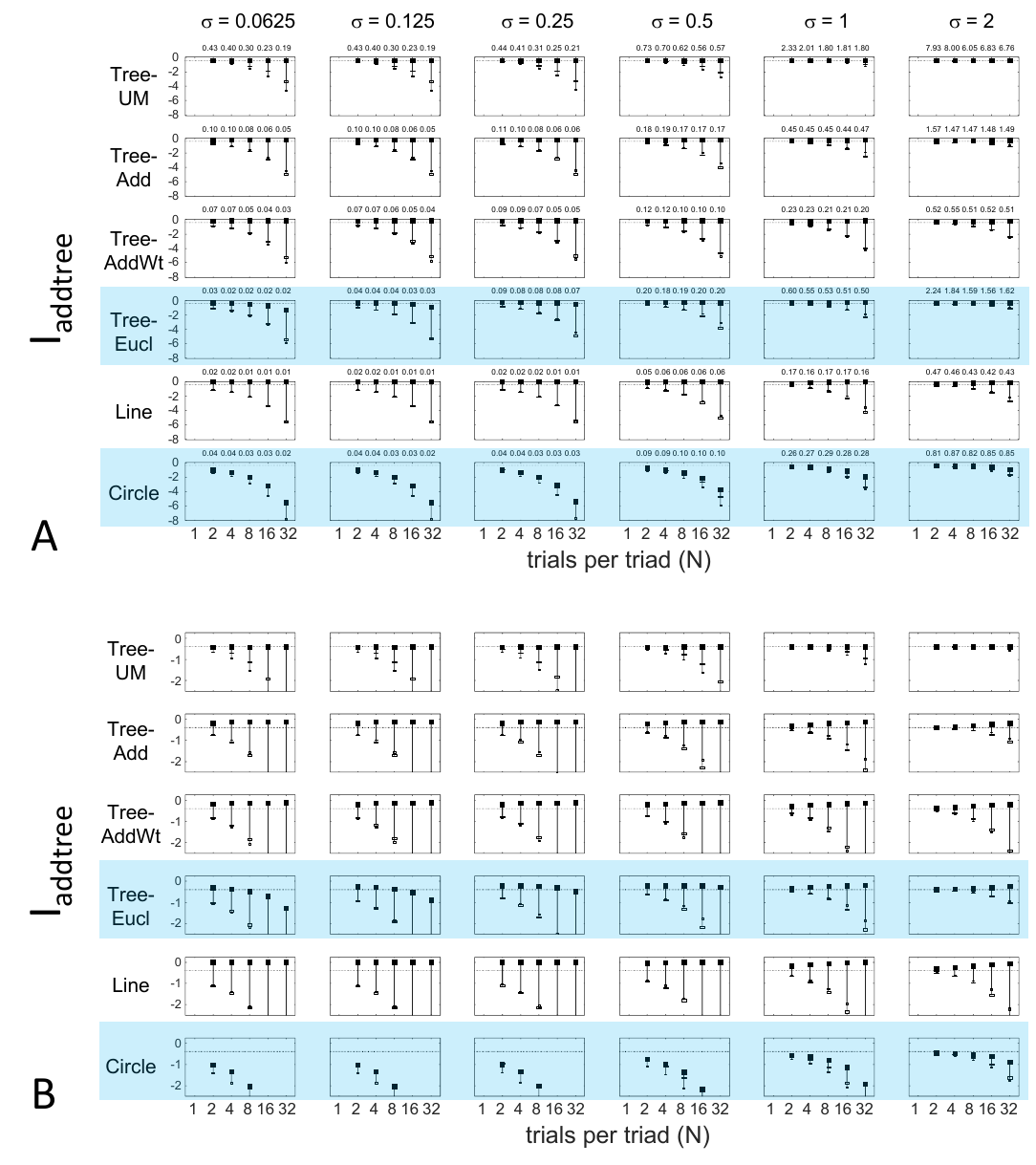}
    \caption{(supplement to main text Figure~\ref{fig:7}).
    Panel A. Behavior of $I_{\text{addtree}}$ for the 15-point configurations, a range of trials per triad ($N$, abscissa), and a range of decision rules ($\sigma$, standard deviation of added noise). 
    Solid symbols indicate values of $I_{\text{addtree}}$ computed from simulated decisions. 
    The dashed horizontal line indicates the \textit{a priori} value for $I_{\text{addtree}}$. 
    Hollow boxes indicate values of $I_{\text{addtree}}$ computed from surrogates, showing the $\pm1$ standard deviation range:
    wider boxes for ``flip any'' surrogates, narrower boxes for ``flip all'' surrogates. 
    The thin vertical lines are to aid visualization, and do not represent ranges.
    The beta-function prior \eqref{eq:3.4} is used, 
    with the value of the parameter $a$ (shown in fine print over each data point) determined by maximum likelihood. 
    Panel B expands the vertical scale. Blue overlay indicates the simulated datasets that are incompatible with the addtree property.
    }
    \label{fig:s3}
\end{figure}

\begin{figure}
    \centering
    \includegraphics[width=\linewidth]{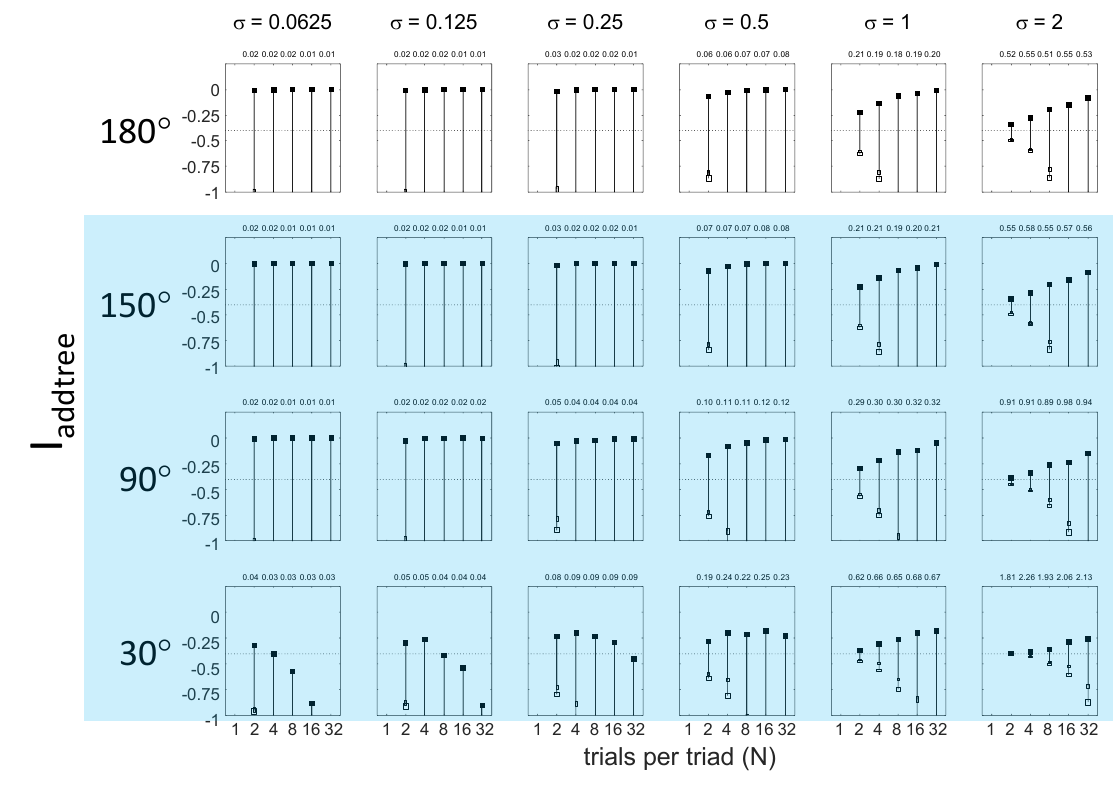}
    \caption{(supplement to main text Figure~\ref{fig:8}).
    Behavior of $I_{\text{addtree}}$ for the 13-point configurations, 
    a range of trials per triad (abscissa), and a range of decision rules ($\sigma$, standard deviation of added noise). 
    Blue overlay indicates the simulated datasets that are incompatible with the addtree property. Other graphical conventions as in Supplementary Figure~\ref{fig:s3}.}
    \label{fig:s4}
\end{figure}

\begin{figure}
    \centering
    \includegraphics[width=\linewidth]{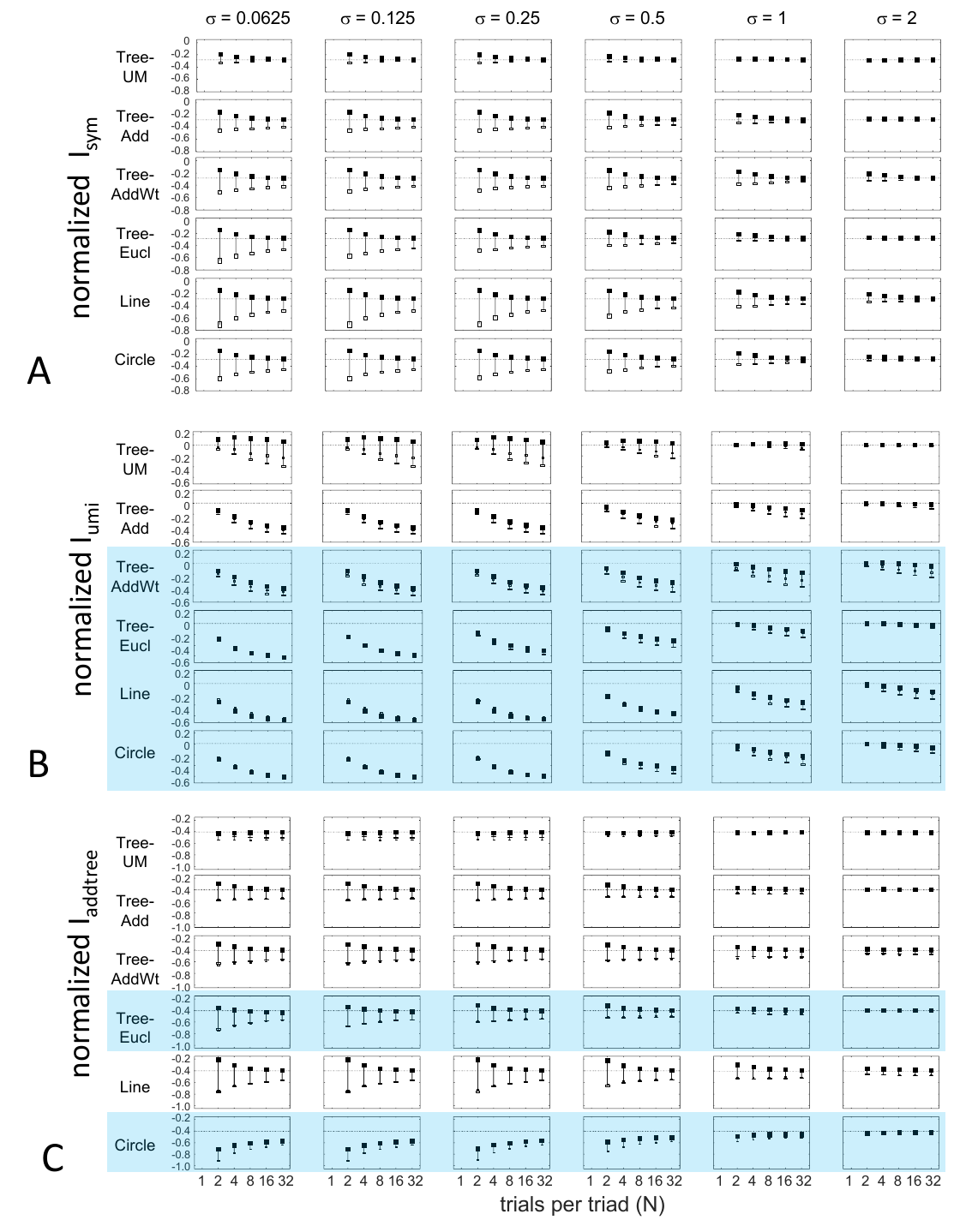}
    \caption{(supplement to main text Figure~\ref{fig:9}).
    Normalized values of $I_{\text{sym}}$ (panel A), $I_{\text{umi}}$ (panel B), and $I_{\text{addtree}}$ (panel C), for the 15-point configurations, a range of trials per triad ($N$, abscissa), and a range of decision rules ($\sigma$, standard deviation of added noise).
    Data from Supplementary Figures~\ref{fig:s1}-\ref{fig:s4}, replotted after normalizing the deviation from the a priori value by the number of trials per triad: $I_0 + (I - I_0) / N$, where 
    $I$ is the index from Supplementary Figures~\ref{fig:s1}-\ref{fig:s4}, $I_0$ is the \textit{a priori} value, and $N$ is the number of trials per triad. 
    Blue overlay indicates the simulated datasets that are incompatible with the ultrametric property (B) or the addtree property (C). Other graphical conventions as in Supplementary Figures~\ref{fig:s1}-\ref{fig:s3}.}
    \label{fig:s5}
\end{figure}

\begin{figure}
    \centering
    \includegraphics[width=\linewidth]{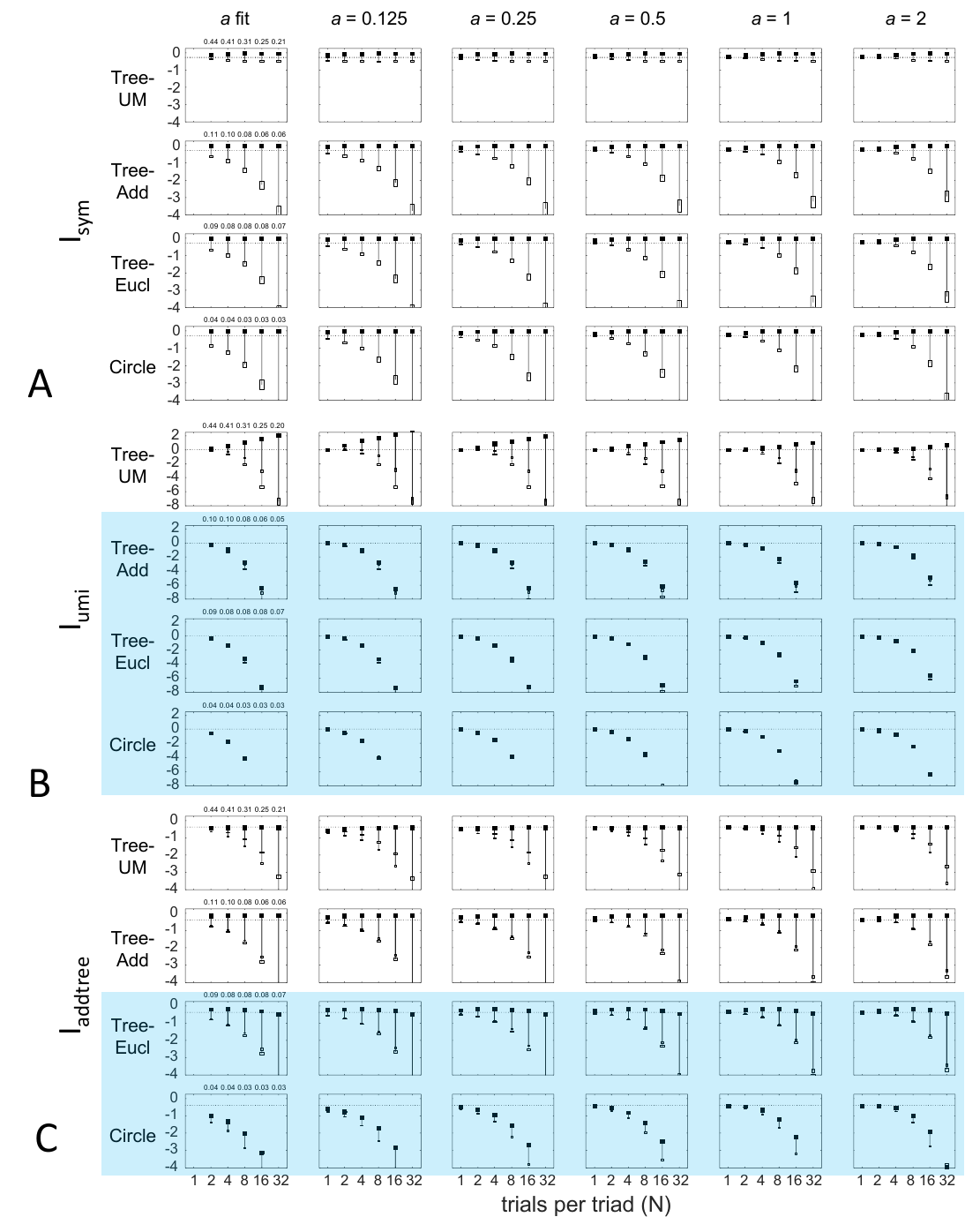}
    \caption{(supplement to main text Figure~\ref{fig:10}).
    Dependence on the prior’s shape parameter, $a$, for $I_{\text{sym}}$ (panel A), $I_{\text{umi}}$ (panel B), and $I_{\text{addtree}}$ (panel C), for the 15-point configurations and a range of trials per triad ($N$, abscissa). 
    The decision rule $\sigma=0.25$ is used. 
    Column 1: $a$ is determined by maximum likelihood; other columns: $a$ is assigned the value indicated over each column. 
    The beta-function prior \eqref{eq:3.4} is used for $I_{\text{sym}}$ and $I_{\text{addtree}}$; the modified prior \eqref{eq:3.27} is used for $I_{\text{umi}}$ with $h=0.001$.
    Blue overlay indicates the simulated datasets that are incompatible with the ultrametric property (B) or the addtree property (C). Other graphical conventions as in Supplementary Figures~\ref{fig:s1}-\ref{fig:s3}.
    }
    \label{fig:s6}
\end{figure}

\begin{figure}
    \centering
    \includegraphics[width=0.93\linewidth]{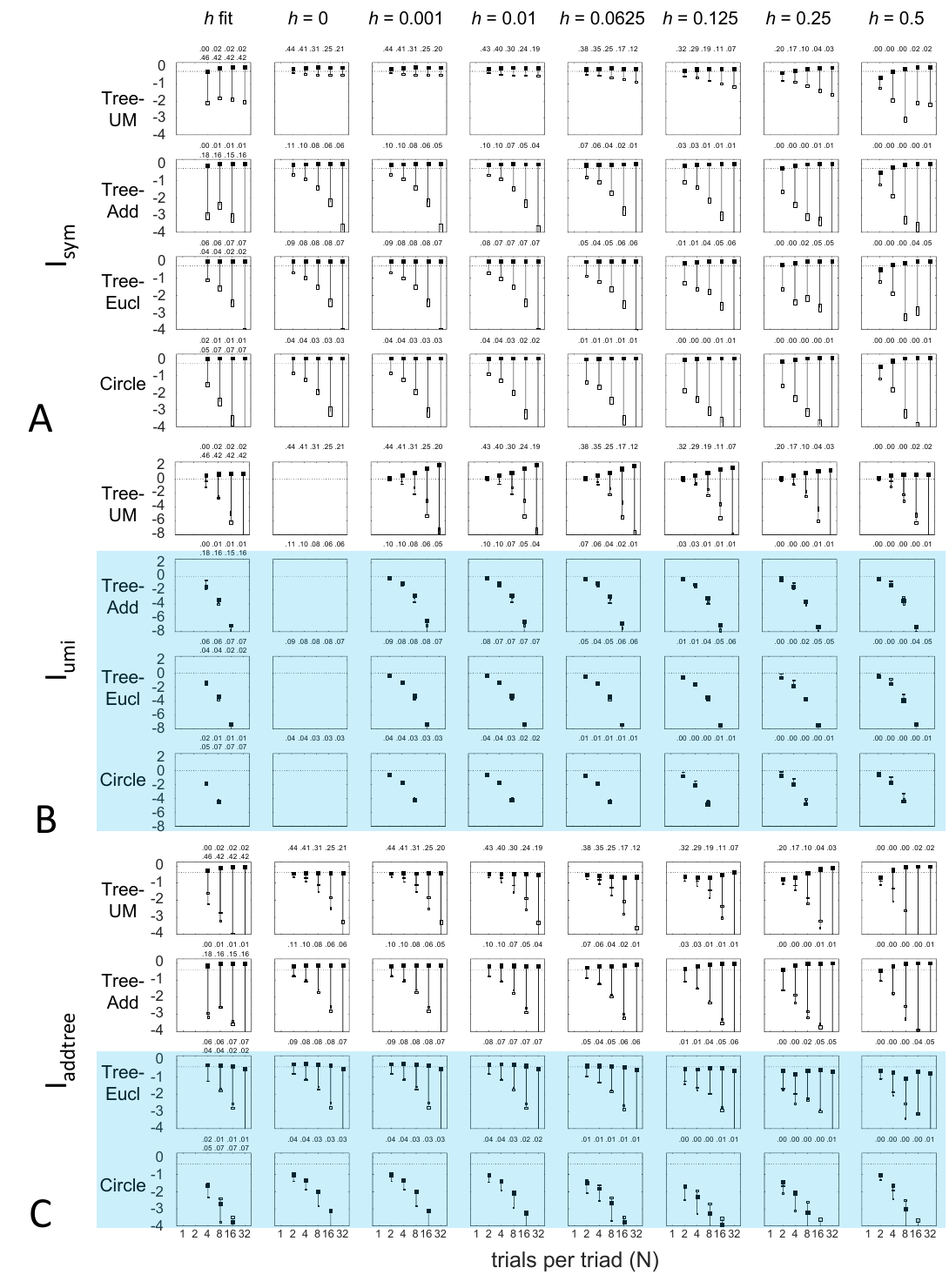}
    \caption{(supplement to main text Figure~\ref{fig:11}).
    Dependence on the prior’s point mass parameter, $h$, for $I_{\text{sym}}$ (panel A), $I_{\text{umi}}$ (panel B), and $I_{\text{addtree}}$ (panel C), for the 15-point configurations and a range of trials per triad ($N$, abscissa). 
    The decision rule $\sigma=0.25$ is used. 
    Column 1: $h$ is determined by maximum likelihood (shown in fine print in lower row over each data point); 
    other columns: $h$ is assigned the value indicated over each column. 
    The parameter $a$ (shown in fine print over each data point) is determined by maximum likelihood in all cases.
    The modified prior \eqref{eq:3.27} is used for all indices.
    Note that $I_{\text{umi}}$ is undefined for $h=0$ (panel B, second column). Blue overlay indicates the simulated datasets that are incompatible with the ultrametric property (B) or the addtree property (C).
    Other graphical conventions as in Supplementary Figures~\ref{fig:s1}-\ref{fig:s3}.
    }
    \label{fig:s7}
\end{figure}

\end{supplement}

\newpage
%%%%%%%%%%%%%%%%%%%%%%%%%%%%%%%%%%%%%%%%%%%%%%%%%%%%%%%%%%%%%%%%%%%
%%                                                               %%
%% Use the two commands below for producing your bibliography    %%
%% with bibtex, then comment again the commands and include the  %%
%% content of the .bbl file in this file below the commands.     %%
%%                                                               %%
%%%%%%%%%%%%%%%%%%%%%%%%%%%%%%%%%%%%%%%%%%%%%%%%%%%%%%%%%%%%%%%%%%%

%\bibliographystyle{amsplainhyper}  
%\bibliography{bibliomna}

% add below the content of your .bbl file produced by bibtex.
%
\providecommand{\bysame}{\leavevmode\hbox to3em{\hrulefill}\thinspace}
\providecommand{\MR}{\relax\ifhmode\unskip\space\fi MR }
\providecommand{\MRhref}[2]{%
  \href{http://www.ams.org/mathscinet-getitem?mr=#1}{#2}
}
\providecommand{\bysame}{\leavevmode\hbox to3em{\hrulefill}\thinspace}

%%% HERE STARTS .BBL FILE 

%%%%%%%%%%%%%%%%%%%%%%%%%%%%%%%%%%%%%%%%%%%%%%%%%%%%%%%%%%%%%%%%%%%
%%                                                               %%
%% You may add acknowledgments (optional).                       %%
%%                                                               %%
%%%%%%%%%%%%%%%%%%%%%%%%%%%%%%%%%%%%%%%%%%%%%%%%%%%%%%%%%%%%%%%%%%%

\ACKNO{This work was supported by NIH NEI EY7977 (JV), NSF: 2014217 (JV), and the Fred Plum
Fellowship in Systems Neurology and Neuroscience (SAW). GA would like to thank Marianne
Maertens for all her support. We thank Laurence T. Maloney for his many helpful discussions,
especially concerning the addtree model.}

%%%%%%%%%%%%%%%%%%%%%%%%%%%%%%%%%%%%%%%%%%%%%%%%%%%%%%%%%%%%%%%%%%%
%%                                                               %%
%% You have reached the end of your document.                    %%
%%                                                               %%
%%%%%%%%%%%%%%%%%%%%%%%%%%%%%%%%%%%%%%%%%%%%%%%%%%%%%%%%%%%%%%%%%%%
\nocite{*}
\end{document}